\documentclass[a4paper,11pt]{article}

\usepackage[left=2.5cm,right=2.5cm,top=2cm,bottom=2cm]{geometry}

\usepackage{amsmath}
\usepackage{amssymb,fancyhdr}
\usepackage{amsthm}

\usepackage{pgf,tikz}
\usepackage{graphicx}
\usepackage{setspace,xspace}
\usepackage[version=3]{mhchem}

\usepackage{enumerate,hyperref}
\usepackage{url}

\usepackage[sort&compress,comma,square,numbers]{natbib}

\newcommand{\R}{\mathbb{R}}

\newcommand{\mC}{\mathcal{C}}
\newcommand{\mP}{\mathcal{P}}

\DeclareMathOperator*{\sign}{sign}
\DeclareMathOperator*{\im}{im}

\DeclareMathOperator{\bd}{bd}
\DeclareMathOperator{\cl}{cl}
\DeclareMathOperator{\id}{id}

\newcommand{\N}{{\mathbb N}}
\renewcommand{\k}{{\kappa}}
\newcommand{\st}{\mid}

\newtheorem{lemma}{Lemma}
\newtheorem{proposition}{Proposition}
\newtheorem{theorem}{Theorem}
\newtheorem{corollary}{Corollary}

\theoremstyle{definition}
\newtheorem{remark}{Remark}

\newtheorem{definition}{Definition}

\def\SI{{\bf Supporting Information}\@\xspace}

\begin{document}

\title{Identifying parameter regions for multistationarity}

\author{
Carsten Conradi$^1$, Elisenda Feliu$^{2}$, Maya Mincheva$^3$, Carsten Wiuf$^2$
}

\footnotetext[1]{Life Science Engineering, HTW Berlin, Wilhelminenhoftstr.\ 75, 10459 Berlin}
\footnotetext[2]{Department of Mathematical Sciences, University of Copenhagen, Universitetsparken 5, 2100 Copenhagen, Denmark}
\footnotetext[3]{Department of Mathematical Sciences, Northern Illinois University, 1425 W. Lincoln Hwy.,  DeKalb IL 60115, USA}
\footnotetext[4]{All authors contributed equally to this work} 

 \date{\today}

\maketitle

\begin{abstract}
Mathematical modelling has become an established tool for studying the dynamics of  biological systems. Current applications range from building models that reproduce quantitative data to  identifying  systems with predefined qualitative features, such as switching behaviour, bistability or oscillations.  Mathematically, the latter question amounts to identifying parameter values associated with a given qualitative feature.

We introduce {a procedure} to partition the parameter space of a parameterized system of ordinary differential equations  into regions for which  the system has a unique or multiple  equilibria. The  {procedure} is based {on the computation of the Brouwer degree}, and it creates a multivariate polynomial with parameter depending coefficients. The signs of the coefficients {determine} parameter regions with and without multistationarity. A particular strength of the  {procedure} is the avoidance of numerical analysis and parameter sampling. 

The procedure consists of a number of steps. Each of these steps might  be addressed algorithmically using various {computer programs and available software}, {or manually}. 
We demonstrate our  procedure on {several} models of  gene transcription and cell signalling, and show that in many cases we  obtain a complete partitioning of the parameter space with respect to multistationarity.

\end{abstract}

\section*{Introduction}

Mathematical models in the form of parameterized
systems of ordinary differential equations (ODEs) are valuable tools
in biology.   Often, qualitative properties of the ODEs are
associated with macroscopic biological properties and biological functions \cite{gould,laurent1999,ozbudak2004,Xiong:2003jt}.
{It is therefore important that we are able to analyse mathematical models with respect to their qualitative features and to understand when these properties arise in models.
With the growing adaptation of differential equations in biology,    an automated screening of ODE models} for parameter dependent properties and discrimination of parameter regions with different properties would be a very useful tool
for biology, and perhaps even more for synthetic biology \cite{Marchisio:2009eg}. 
Even though it is currently not conceivable how and if this task can
be efficiently formalized, we view the {procedure} presented here as a first step in  this direction.

\emph{Multistationarity},  that is,  the 
{capacity} of the
system  to rest in different positive equilibria depending on the
initial state of the system, is an important   qualitative  property. Biologically, multistationarity is linked to cellular
decision making    and `memory'-related on/off responses to graded
input \cite{laurent1999,ozbudak2004,Xiong:2003jt}.   
Consequently the existence of multiple equilibria is often a design objective in synthetic biology \cite{Gardner:2000bm,palani}.
Various  mathematical methods, developed in the context of reaction
network theory, can be applied to decide whether
multistationarity exists for some parameter values  or not at all, or
to pinpoint specific values for which it does occur
\cite{Feinbergss,feliu_newinj,wiuf-feliu,PerezMillan,fein-043,conradi-PNAS,mincheva2007,otero1,otero2,craciun2008}.
Some of these methods   are  freely available as software tools  \cite{control,crnttoolbox}.

It is a hard mathematical problem to delimit parameter regions for which multistationarity occurs. Often it is solved by numerical investigations and parameter sampling, guided by  biological intuition or by case-by-case mathematical approaches. {A general approach, in part numerical, is based on a certain bifurcation condition \cite{otero1,otero2,OteroMuras:2012fn,Otero2017}.}  Alternatively, for polynomial   ODEs, a  decomposition of the parameter space into regions with different numbers of equilibria  could be achieved by Cylindrical Algebraic Decomposition
(a version of  quantifier elimination) \cite{rob-024}. This  method,  however, scales very poorly and is thus only of limited help in biology, where models tend to be large in terms of the number of variables and parameters.

 Here we present two new theoretical results   pertaining to multistationarity (Theorem~1 and Corollary~2). The results are in the context of reaction network theory and generalize ideas in \cite{maya-bistab,conradi2014graph}.   We consider a parameterized ODE system  
{  defined by}
a  reaction network and compute a single polynomial in the {species concentrations}  with coefficients depending on the parameters of the system.  The theoretical results relate the capacity for multiple equilibria or a single equilibrium  to  the signs of the polynomial as a function of the parameters and the variables (concentrations).

The theoretical results apply to \emph{dissipative} reaction networks (networks for which all trajectories eventually remain in  a compact set) without \emph{boundary equilibria}  {in stoichiometric compatibility classes with non-empty interior}. These conditions  are met in many reaction network models of molecular systems.
{We show by example that the results allow us to identify regions of the parameter space for which multiple equilibria exist and regions for which only one equilibrium  exists.  }Subsequently this leads to the formulation of a general \emph{procedure} for detecting regions of mono- and multistationarity.
The procedure {verifies}  the conditions of the theoretical results  and further, calculates the before-mentioned polynomial. A key ingredient is the existence of a \emph{positive parameterization} of the set of positive equilibria. Such a parameterization is  known to exist for many classes of reaction networks, for example, systems with toric steady states \cite{PerezMillan} and post-translational modification systems \cite{TG-rational,fwptm}.

The conditions of the procedure might be verified manually or algorithmically according to computational criteria.  The algorithmic criteria are, however, only sufficient for the conditions to hold. For example, a basic condition  is that of dissipativity. To our knowledge there is not a sufficient and necessary computational criterion for dissipativity, but several sufficient ones. If these fail, then the reaction network might still be dissipative, which might be verified by other means.
By collecting the algorithmic criteria, the  procedure can be formulated as a fully automated procedure (an algorithm) that  partitions the parameter space without any manual intervention. The algorithm might however terminate indecisively if some of the criteria are not met.

Table~\ref{tab:cherry} shows two examples of reaction network motifs  that occur frequently in intracellular signalling: a two-site protein
modification  by a kinase--phosphatase pair and  a one-site
modification of two proteins by the same kinase--phosphatase pair. {These reaction networks are in the domain of the automated procedure  and conditions for mono- and multistationarity can be found without any manual intervention. } The
conditions discriminating  between a unique and multiple  equilibria
highlight a  delicate relationship between the catalytic and
Michaelis-Menten constants  of the kinase and the phosphatase with the
modified protein as a substrate (the $k_c$- and
$k_M$-values).  If the condition for multiple equilibria is met,
then multiple equilibria occur provided the total concentrations of
kinase, phosphatase and substrate are in suitable ranges (values
thereof can be computed as part of the procedure).

\begin{table}[!t]

\medskip
\centering
{\footnotesize \begin{tabular}{|c|c|}
\hline
 Motif & Condition \\  \hline
   \begin{minipage}[h]{0.58\textwidth}        
 {\small  $ A + K \cee{<=>} AK \cee{->} A_p + K$ \quad  $ B + K \cee{<=>} BK \cee{->} B_p + K$}   \\[5pt]
    {\small    $ A_p + F \cee{<=>} A_pF \cee{->} A + F$  \quad    $ B_p + F \cee{<=>} B_pF \cee{->} B + F$}
     \end{minipage}      & 
          \begin{minipage}[h]{0.5\textwidth}         
          
          \vspace{0.1cm}
         $ \begin{aligned}
b(\kappa) &= \bigl(k_{c1} k_{c4} - k_{c2} k_{c3}\bigr) \cdot  \left(
          \frac{k_{c1} k_{c4}}{k_{M1} k_{M4}} - 
          \frac{k_{c2} k_{c3}}{k_{M2} k_{M3}}
        \right)
      \end{aligned} $
 \\[8pt]
 Multiple:  $b(\kappa)< 0$, \quad 
      Unique:        $b(\kappa)\geq0$ \\[-3pt]
          \end{minipage} \\[0.3cm]  \hline
             \begin{minipage}[h]{0.58\textwidth}    
             \begin{center}
      {\small  $ A + K \cee{<=>} AK \cee{->} A_p + K \cee{<=>} A_pK \cee{->} A_{pp}+K$} \\[5pt]
        {\small    $ A_{pp} + F \cee{<=>} A_{pp}F \cee{->} A_p + F  \cee{<=>} A_pF \cee{->} A + F $ }
        \end{center}
                 \end{minipage}   &   
      \begin{minipage}[h]{0.5\textwidth}
      
                \vspace{0.1cm}
  $  \begin{aligned}
          b_1(\kappa) &= k_{c1} k_{c4} - k_{c2} k_{c3} \\
     b_2(\kappa) &= k_{c1} k_{c4}(k_{M2} + k_{M3})- k_{c2} k_{c3}(k_{M1} + k_{M4})
      \end{aligned}$
     \\[8pt]
      Multiple:     $b_1(\kappa)< 0$, \quad 
        Unique: 
      $b_1(\kappa) \geq  0$ and $b_2(\kappa) \geq 0$ \\[-3pt]
    \end{minipage}
  \\[0.1cm]  \hline
\end{tabular}}

\medskip
\caption{
{\bf Conditions for unique and multiple equilibria   in  post-translational modification of proteins}. The symbols $k_{ci}$ and $k_{Mi}$ denote {respectively} the catalytic and the Michaelis-Menten constants of the $i$-th modification step       ($i=1$: phosphorylation of $A$, $i=2$:  dephosphorylation of $A_p$,  $i=3$: phosphorylation of $B$ or 
      $A_p$, $i=4$: dephosphorylation of $B_p$ or $A_{pp}$).  All parameter values satisfying the conditions in
     {the second} column  yield multiple (unique) equilibria for some (all) values of the conserved quantities.  For the second motif, we cannot decide on the number of equilibria for $b_1(\kappa)\geq 0$ and $b_2(\kappa)<0$.  See {\S6.1  and \S6.2} in the \SI for details. 
}    \label{tab:cherry}

\end{table}

  The paper has 
    three main sections: a theoretical section, a section about the procedure
    and an application section. We  close the paper with
    two brief sections discussing computational limitations, related
    work and future directions.    In the theoretical    section we first introduce notation and
  mathematical background material. We then give the theorem and the corollary that
  links the number of equilibria to the sign of the determinant of the
  Jacobian of a certain function, which is derived from the ODE system
  associated with a reaction network.
   In the    second section   we  state the procedure,
    derive the algorithm  and comment on the feasibility     and verifiability  of the conditions. 
  Finally, in the application section we apply the procedure to several  examples.
  The \SI has six sections. All proofs are relegated to \S1--4 together with background material. In \S5 we {elaborate further on} how the conditions  
  of the  procedure/algorithm can be verified. In \S6 we provide details of the algorithmic analysis of the examples in Table~\ref{tab:cherry}.  Also we include a further {monostationary} example for illustration of the algorithm.

\section*{Results}

\subsection*{{Theory}}

In this part of the manuscript we present the theoretical
  results. We start by introducing the basic formalism of reaction
  networks.
  Theorem~1, Corollary~1 and 2 below apply to \emph{dissipative networks}
  without \emph{boundary equilibria} and concern the 
  (non)existence
  of \emph{multiple equilibria} in some \emph{stoichiometric
    compatibility class}.
Corollary 2 assumes the existence of a \emph{positive
  parameterization} of the set of positive equilibria.  
  Before stating the results these five concepts are
  formally defined.

\paragraph{Reaction networks. }

A \emph{reaction network}, or simply a \emph{network}, 
consists of a set of species $\{X_1,\dots,X_n\}$ and a set of reactions of the form: 
\begin{equation}\label{eq:reaction}
R_j\colon \sum_{i=1}^n \alpha_{ij} X_i \rightarrow \sum_{i=1}^n\beta_{ij} X_i, \qquad j=1,\dots,\ell
\end{equation} 
where $\alpha_{ij},\beta_{ij}$ are   non-negative integers.  {The left hand side is called the reactant, while the right hand side is called the product.} 
We let $N=(N_{ij})\in \R^{n\times \ell}$ be the \emph{stoichiometric
  matrix} of the network, defined as 
$N_{ij} = \beta_{ij}-\alpha_{ij},$
that is, the $(i,j)$-th entry encodes the net production
of species $X_i$ in reaction $R_j$. We refer to  the \lq running
example\rq  \,\,in  Fig~\ref{fig:setup_running} for an illustration of the
definitions.

\begin{figure}[h]
  \centering
 \includegraphics[scale=1]{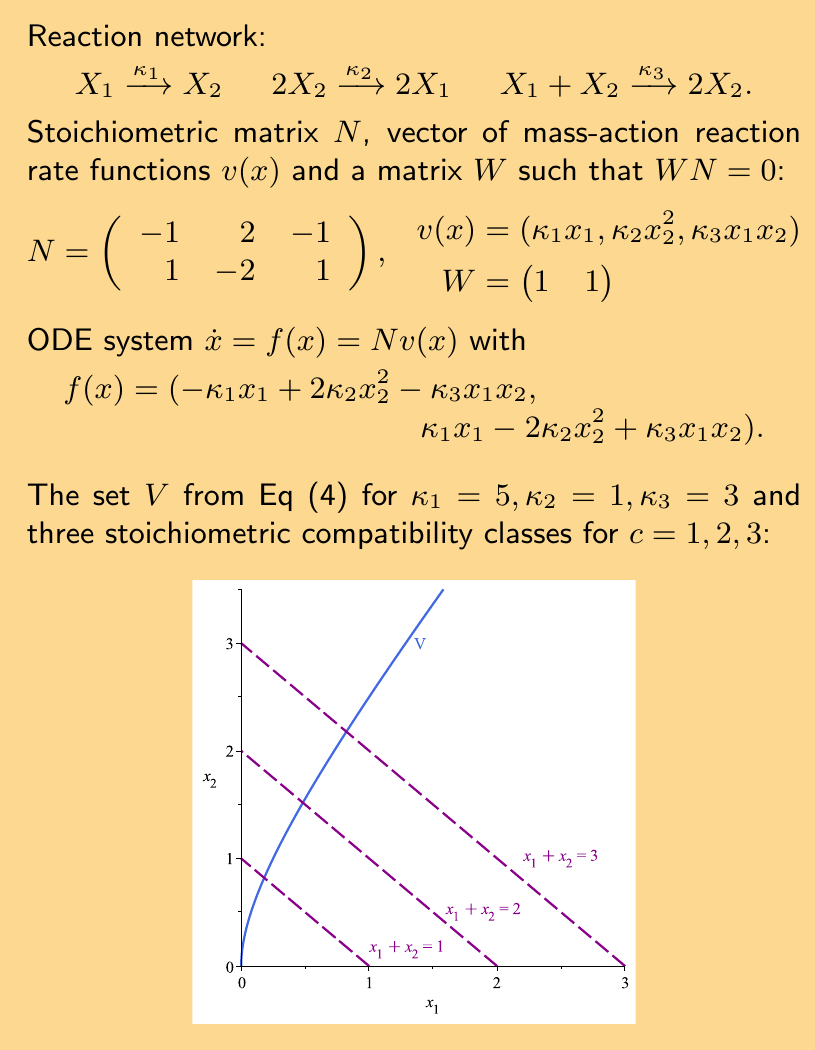}
  \caption{  \label{fig:setup_running}
    {\bf Running example. } 
    Example network with two species, $X_1$ and $X_2$, and three reactions with mass-action kinetics.
  }
\end{figure}
The concentrations of the species $X_1,\dots,X_n$  are denoted by lower-case letters $x_1,\dots,x_n$ {and  we let $x=(x_1,\ldots,x_n)$}. 
We denote by $\R^n_{>0}$ ($\R^n_{\geq 0}$), the positive (non-negative) orthant in $\R^n$. 
The evolution of the concentrations with respect to time is modeled as an ODE system {derived from a set of \emph{reaction rate functions}. A reaction rate function for  reaction $R_j$ is  a $\mC^1$-function $v_j\colon
\R_{\geq 0}^n \rightarrow \R_{\geq 0}$ that models the (non-negative) speed of the reaction. We further assume  that 
\begin{equation}
  \label{eq:invariance}
  v_j(x)=0 \quad \Leftrightarrow \quad x_{i}=0 \quad \textrm{ for some
  }i \textrm{ such that }\alpha_{ij}>0, 
\end{equation}
that is, the {reaction only takes place in the presence of all reactant species.}
We refer to the set of reaction rate functions as the \emph{kinetics}.}

{A particular important example of a kinetics is that of  \emph{mass-action kinetics}. In this case the reaction rate functions are given by
$$ v_j(x) = \kappa_j x_1^{\alpha_{1j}}\cdot \ldots \cdot x_n^{\alpha_{nj}},\qquad j=1,\dots,\ell,$$
where $\kappa_j$ is a positive number called the \emph{reaction rate constant} { and we assume $0^0=1$}.  
Other important examples are Michaelis-Menten kinetics and Hill kinetics. All three types of kinetics fulfil the assumption in Eq~\eqref{eq:invariance}. 
}

{For a choice of reaction rate functions $v=(v_1,\ldots,v_\ell)$, the ODE system modelling the species concentrations over time  with initial condition $x(0) =x_0$, is}
\begin{equation}\label{eq:ode}
  \dot{x} = f(x),\qquad x\in \R^n_{\geq 0},\qquad\textrm{where }\quad f(x)= Nv(x).
\end{equation} 
{Under assumption Eq~\eqref{eq:invariance}}, 
the orthants $\R^n_{>0}$ and $\R^n_{\geq 0}$ are {\em forward-invariant} under $f$ in Eq~\eqref{eq:ode} \cite[Theorem 5.6]{smirnov-2002},  \cite[Section 16]{Amann}.  
{Forward-invariance implies that the solutions to the ODE system stays in $\R^n_{>0}$ (resp. $\R^n_{\geq 0}$) for all positive times if the initial condition is in $\R^n_{>0}$ (resp. $\R^n_{\geq 0}$).}

\medskip
 
The trajectories of the ODEs in Eq~\eqref{eq:ode} are  confined to the
so-called \emph{stoichiometric compatibility classes}, which are defined as follows. 
Let $s=\text{rank}(N)$ { be  the rank of the network} and $d=n-s$ {be the corank. }
{Further, let $W\in \R^{d\times n}$ be any matrix  of full rank $d$ such that $WN=0$,}
see Fig~\ref{fig:setup_running} for an example.  This matrix is zero-dimensional if $N$ has full rank $n$.  
{For each} $c\in   \R^d$,  {there is an associated}  \emph{stoichiometric compatibility class} {defined as  } 
$$ \mP_c := \{ x\in \R^n_{\geq 0} \st Wx=c\}.$$
{This set} is  empty if $c\notin W(\R^n_{\geq 0})$.
The \emph{positive stoichiometric compatibility class} is defined as the relative interior of $\mP_c$, {that is, the intersection of $\mP_c$ with the positive orthant}:
$$ \mP_c^+ := \{ x\in \R^n_{> 0} \st Wx=c\} = \mP_c \cap \R^n_{>0}.$$
The sets $\mP_c^+$ and $\mP_c$  are convex. {Since by construction} $Wx$ is conserved over time and determined by the initial condition, then {$\mP_c^+$ and $\mP_c$ are also} forward-invariant. 

An equation of the form $\omega\cdot x = c'$ for some $\omega\in
\im(N)^\perp$ and $c'\in\R$ is called a \emph{conservation relation}. In
particular, $Wx=c$ forms a system of $d$ conservation relations. 

{For the running example in Fig~\ref{fig:setup_running}, the rank of the network is $s=1$ and the corank is $d=1$. }The matrix
$W$  in the figure leads to the conservation relation
$ x_1+x_2 = c$. 
{Here  the stoichiometric compatibility class $\mP_{c}$ has non-empty interior, that is, }$\mP_{c}^+\neq \emptyset$, if and only if
$c>0$.

In the following,  {to ease the notation, we implicitly assume a reaction network comes with a kinetics (a set of reaction rate functions) and the associated ODE system.} 
 
\paragraph{Dissipative and conservative reaction networks. }
A  reaction network is \emph{dissipative} if, {for all  stoichiometric compatibility classes $ \mP_c$, there exists a compact set where the trajectories of $ \mP_c$  eventually enter}
 (see \S3.2 in the \SI). A  reaction network is \emph{conservative} if 
there exists a conservation relation with only positive coefficients, or, equivalently, {if for all species $X_i$ there is a conservation relation such that the coefficient of $x_i$ is positive and all  other coefficients are  non-negative}.
 This is equivalent to the stoichiometric compatibility classes being compact sets  \cite{benisrael}. 
Hence, in particular, a conservative reaction network is dissipative {because we can choose the attracting compact set to be the stoichiometric compatibility class itself}.
{Because of  the conservation relation $x_1+x_2=c$, } the reaction network of the  running example is conservative.

\paragraph{Equilibria. }
Given the ODE in Eq~\eqref{eq:ode}, the set of non-negative equilibria is {the set of points for which $f(x)$ vanishes: }
\begin{equation}\label{eq:V}
  V  = \{ x \in \R^n_{\geq 0} \st  f(x) = 0\}.
\end{equation}
We are interested in the positive equilibria in each  stoichiometric compatibility class, that is, {in} the set $V  \cap {\mP_{c}^+}$. Generically, this set consists of isolated points obtained as the {simultaneous} positive solutions to the equations 
\begin{equation}\label{eq:ss}
  f(x)=0,\qquad Wx = c.
\end{equation}
{Fig~\ref{fig:setup_running} shows a representation of the set $V$ together with examples of stoichiometric compatibility classes for the running example. }
 {The figure suggests that the set $V$ intersects each stoichiometric compatibility class in exactly one point.}

We {introduce}  some definitions: a network admits \emph{multiple equilibria} (or is \emph{multistationary}) if there exists $c\in \R^d$ such that $V \cap \mP_{c}^+$ contains at least two points, {that is, the system in Eq~\eqref{eq:ss}
has at least two positive solutions.  }
Equilibria belonging to $V \cap \mP_{c}$ but not to $V  \cap\mP_{c}^+$ for some $c$ are \emph{boundary equilibria}. 
A boundary equilibrium has at least one coordinate equal to zero.

\paragraph{The function $\varphi_c(x)$. }
{Some of the $n$ equations in the system $f(x)=0$ might be redundant. Indeed, every vector $\omega\in \im(N)^{\perp}$ fulfils $\omega \cdot f(x)=0$, and hence gives a linear relation among the entries of $f(x)$. As a consequence, there are (at least) as many independent linear relations as rows of $W$, that is, $d$,  and there are  at most $s=n-d$ linearly independent  equations in the system $f(x)=0$.
Thus $d$ of the equations   are redundant. } 
 {By removing these from $f(x)=0$, the system in Eq~\eqref{eq:ss} becomes a system of $n$ equations in $n$ variables.}

 {In order to systematically choose $d$ equations to remove, we proceed as follows. }
We choose the matrix of conservation relations $W\in\R^{d\times n}$ to be row reduced 
and let
$i_1,\dots,i_d$ be the  indices of the first non-zero coordinate of each row. {Then the scalar product of the $j$-th row of $W$  with $f(x)$ can be used to express $f_{i_j}(x)$ as a linear combination of the entries of $f(x)$ with indices different from $i_1,\dots,i_d$. It follows that the equations $f_{i_1}(x)=0,\dots,f_{i_d}(x)=0$ can be removed.}

For $c\in \R^d$,  we define the $\mC^1$-function  $\varphi_c(x) \colon \R^n_{\geq 0} \rightarrow \R^n$ by
\begin{equation}
  \varphi_c(x)_i =  \begin{cases}  f_i(x) & i\notin \{i_1,\dots,i_d\} \\
    (Wx-c)_i  & i\in \{i_1,\dots,i_d\}.
  \end{cases} \label{eq:varphi}
  \end{equation}
  {
    For the running example  in Fig~\ref{fig:setup_running} the
    matrix $W$ is already row reduced with $i_1=1$. Hence $\varphi_c$
    is obtained by replacing $f_1(x)$ with $x_1+x_2-c$:
    \begin{displaymath}
        \varphi_c(x) =
        \begin{pmatrix}
          x_1+x_2- c \\
          \kappa_1 x_1-  2\kappa_{2} x_2^2+  \kappa_{3} x_1x_2
        \end{pmatrix}
        .
    \end{displaymath}
  }
  
{ As the function $\varphi_c(x)$ is obtained by replacing redundant
  equations in $f(x)=0$ with equations defining $\mP_c$, we have }
\begin{displaymath}
  V  \cap \mP_{c}= \{ x\in \R^n_{\geq 0} \st \varphi_c(x)=0\}.
\end{displaymath}
{Consequently, }  a network admits \emph{multiple equilibria}   if the equation $\varphi_c(x)=0$ has at least two positive solutions for some  $c\in \R^d$.

\paragraph{{A theorem for }unique and multiple equilibria. }

  Let  $M(x)\in \R^{n\times n}$ be the Jacobian matrix of $\varphi_c(x)$, {that is, the matrix with  $(i,j)$-th entry equal to the partial derivative of $\varphi_{c,i}(x)$ with respect to $x_j$.}
 The matrix $M(x)$  does not depend on $c$, see Eq~\eqref{eq:varphi}.   
 
 {We say that an} equilibrium $x^*\in V \cap \mP_{c}$  is \emph{non-degenerate} if
  the Jacobian of $\varphi_c$ at $x^*$, $M(x^*)$, is non-singular, that
  is, if  $\det(M(x^*))\neq 0$ \cite{wiuf-feliu}.

\medskip
\noindent
{\bf Theorem 1 (Unique and multiple equilibria). } 
Assume  the reaction rate functions fulfil Eq~\eqref{eq:invariance}, let $s=\text{rank}(N)$ and let $\mP_c$ be a stoichiometric compatibility class such that  $\mP_c^+\neq \emptyset$, where $c\in \R^d$. Further, assume that  

\smallskip
\noindent (i) The  network is dissipative.
 
\noindent (ii) There are no boundary equilibria in  $\mP_c$.  

\smallskip
\noindent Then the following holds.
    
\smallskip
\noindent    (A') {\bf Uniqueness of equilibria.} If
    \begin{displaymath}
      \sign(\det(M(x)))=(-1)^s\quad \textrm{ for  \textbf{all}  positive equilibria} \quad  
                 x  \in V\cap \mP_c^+,     
          \end{displaymath}
  then there is exactly one  positive equilibrium in $\mP_c$.
 Further, this equilibrium is
    non-degenerate.
    
    \smallskip
    \noindent  (B') {\bf Multiple equilibria.} If
    \begin{displaymath}
        \sign(\det(M(x)))=(-1)^{s+1}\quad \text{ for \textbf{some} equilibrium }  \quad  
                 x \in V\cap \mP_c^+,     
    \end{displaymath}
    then there are at least two positive equilibria in  $\mP_{c}$, at
    least one of which is non-degenerate.
    If all positive equilibria in $\mP_{c}$ are non-degenerate, then
    there are at least three and always an odd number.
 
\medskip
{The proof of }Theorem 1 is based on relating $\det(M(x))$ to the  Brouwer degree of  $\varphi_c$ at $0$ (see \S1-\S4 in the \SI).
Note that the only situation that is not covered by Theorem 1 is when $\sign(\det(M(x)))$ takes the value $0$ for some $x$, but never the value $(-1)^{s+1}$.  The determinant of $M(x)$ is the same as the \emph{core determinant} in \cite[Lemma 3.7]{helton:determinant}. See also  \cite[Remark 9.27]{wiuf-feliu}.

{
To check whether the  sign conditions in part (A') or (B') hold requires  information 
  about the equilibria in $\mP_c^+$. As such, these conditions are difficult to check.
 If $\sign(\det(M(x)))$ is constant for all $x$ in a set containing the  positive equilibria, then the condition in (A') is 
 always fulfilled. In particular, this is  
  the case
}
for \emph{injective networks}, where $\sign(\det(M(x)))=(-1)^s$ for all $x\in \R^n_{>0}$ \cite{wiuf-feliu} (see also \cite{feliu2012,craciun2005multiple,fein-044,inj-006,banaji2010,pantea} for related work on injective networks).  The latter might  be verified or falsified without any knowledge about the equilibria of the system ({see the comments to Step 5 and Step 7 in the section ``Procedure for finding parameter regions for mono- and multistationarity''}).

\medskip
\noindent
{{\bf Corollary 1 (Unique  equilibria). } 
Assume that the assumptions of Theorem~1 hold and that $\sign(\det(M(x)))=(-1)^s$ for all $x\in \R^n_{>0}$. Then there is exactly one  positive equilibrium in each stoichiometric compatibility class.  Further, this equilibrium is
    non-degenerate.}
\medskip

{The conclusions of Theorem 1  refer specifically to
  non-degenerate equilibria. Non-degenerate equilibria are always
  isolated from each other within a given stoichiometric compatibility
  class, as $\det(M(x))\not=0$ ensures $M(x)$     is  locally invertible.
In some situations we might be able to ``lift" non-degenerate equilibria of a reaction network to another reaction network that in some sense is larger, thereby proving lower bounds on the number of non-degenerate equilibria of the larger reaction network.
This is for example the case if the smaller network is embedded in the larger  \cite{joshi-shiu-II,joshi-shiu-III}, if the smaller network is without inflows/outflows while the larger has all inflows/outflows  \cite{craciun-feinberg}, or if the smaller is obtained by elimination of intermediate species \cite{feliu:intermediates}.  
 Conditions for the existence of degenerate equilibria, where  $\det(M(x))$ is expected to change sign, are also known \cite{conradi-switch,OteroMuras:2012fn}. }

\paragraph{Positive parameterizations and a corollary. }
{  Verifying condition (A') or (B') is considerably easier if there exists 
  a positive parameterization of the set $V\cap \R^n_{>0}$ of all
  positive  equilibria.  In this subsection we define such a parameterization and restate Theorem~1 as
Corollary~2  in this situation. In the following sections this
  corollary will become the foundation for the procedure to partition the
  parameter space into regions with different equilibrium properties.
}

{By a \emph{positive parameterization} of the set of positive equilibria we mean
a surjective function    }
 \begin{equation} \label{eq:param}
 \begin{array}{rccc}
  \Phi\colon &  \R^{m}_{> 0} & \rightarrow & {V\cap } \R^n_{> 0}   \\ 
  & \hat{x}={(\hat{x}_1,\dots,\hat{x}_m)} & \mapsto & {(\Phi_{1}(\hat{x}),\dots,\Phi_{n}(\hat{x})),  }
\end{array}
\end{equation}
for some $m<n$, such that  $\hat{x}\in \R^{m}_{>0}$ is the vector of free variables.  In other words,  a positive parameterization  implies
 that $x_{1},\dots,x_n$ are expressed at equilibrium as functions of  $\hat x$: 
$$ x_{i}= \Phi_{i}(\hat x),\qquad i={1},\dots,n, $$
such that  $x_{1},\dots,x_n$ are positive provided  $\hat x$ is positive.  Thus
\begin{equation}\label{eq:Vparam}
V\cap \R^n_{>0} = \{ \Phi(\hat{x}) \st \hat{x} \in \R^m_{>0}\}.
\end{equation}
{Typically, the number of free variables equals the corank of the network, that is $m=d=n-s$. }

We say that a parameterization is \emph{algebraic} if the components $\Phi_i(\hat x)$ are  polynomials or rational functions (quotients of polynomials) {and  can be given such that the denominator is positive for all $\hat{x}$}. See  Fig~\ref{fig:running_Exa} (Step~{6}) for an {application to the running example.} 
{Note that the parameterizations considered here do not make use of the conservation relations.}

{
A positive equilibrium $\Phi(\hat x)$, $\hat x\in\R^m_{>0}$, 
 } belongs to the stoichiometric compatibility class $\mP_c$ where 
\begin{equation}\label{eq:c_phi}   
  c:= W \Phi(\hat{x}).
\end{equation}
{Combining Eq~\eqref{eq:Vparam} and Eq~\eqref{eq:c_phi}, it follows that}  the positive solutions to Eq~\eqref{eq:ss} for a  given $c$ are in one-to-one correspondence with the positive solutions to Eq~\eqref{eq:c_phi}, that is,
$$ V\cap \mP_c^+ = \{\Phi(\hat{x}) \st   \hat{x} \in \R^m_{>0} \quad\textrm{and}\quad c= W \Phi(\hat{x}) \}.$$

  {In order to restate Theorem 1 using the parameterization $\Phi$, } 
we consider the determinant of $M(x)$ evaluated at $\Phi(\hat x)$, 
\begin{equation}\label{eq:akp}
  a(\hat{x})  = \det (M(\Phi(\hat{x}))), \quad \hat x\in \R^m_{>0}.
\end{equation}

\medskip
\noindent
{\bf {Corollary 2} (Positive parameterization). } 
Assume  the reaction rate functions fulfil Eq~\eqref{eq:invariance} and let $s=\text{rank}(N)$.  Further, assume that  

\smallskip
\noindent (i) The  network is dissipative.
 
\noindent (ii)  There are no boundary equilibria in  $\mP_c$,  for all $c\in \R^d$ such that $\mP_c^+\neq \emptyset$.
 
\noindent (iii) The set of positive equilibria admits a positive parameterization as in Eq~\eqref{eq:param}.

\smallskip
\noindent Then the following holds.
    
\smallskip
\noindent    (A) {\bf Uniqueness of equilibria.} If 
$$\sign(a(\hat{x}))=(-1)^s\quad \textrm{ for  \textbf{all}  } \hat{x}\in \R^m_{>0},$$
    then there is exactly one  positive equilibrium in each    $\mP_c$ with $\mP_c^+\neq \emptyset$.     Further, this equilibrium is  non-degenerate.
    
    \smallskip
    \noindent  (B) {\bf Multiple equilibria.} If 
    $$\sign(a(\hat{x}))=(-1)^{s+1}\quad \textrm{ for  \textbf{some}  }\hat{x}\in \R^m_{>0},$$ 
    then there are at least two positive equilibria in  {the stoichiometric compatibility class }$\mP_{c}$  where  $c:= W \Phi(\hat{x})$.
    Further,  at    least one of {the equilibria} is non-degenerate.     If all positive equilibria in $\mP_{c}$ are non-degenerate, then
    there are at least three equilibria  and always an odd number.

\medskip
 {Note that, contrary to Theorem 1, the stoichiometric compatibility class $\mP_c$ is not fixed in the corollary.}

{
In the next section we formulate a procedure based on Corollary~1 and Corollary~2 to find regions of mono- and multistationarity.
Before that we end this section with an application   to the running example. The analysis is divided into seven steps which prelude the steps of the procedure.
}

\paragraph{
  {Application of Corollary~1 and Corollary~2 to the running example.
  } 
}

 We start with the setup given in Fig~\ref{fig:setup_running} and first check whether the sign condition of Corollary~1 is fulfilled, in which case there is a single equilibrium in all stoichiometric compatibility classes.  The steps of the analysis are illustrated in Fig~\ref{fig:running_Exa}.  
 
 The assumptions of the corollary are easily verified in this case. As we are assuming mass-action kinetics, Eq~\eqref{eq:invariance} is fulfilled  (Fig~\ref{fig:running_Exa}, Step 1). Further the network is conservative, hence dissipative and (i)  is fulfilled (Fig~\ref{fig:running_Exa}, Step 2). It is easily seen that there are no boundary equilibria in any stoichiometric compatibility class with non-empty interior (Fig~\ref{fig:running_Exa}, Step 3).  Hence (ii) is fulfilled. We then construct $\varphi_c(x)$ and calculate the determinant of  $M(x)$. It  is a polynomial (in fact, a linear function) in $x_1,x_2$ with coefficients containing both positive and negative terms (Fig~\ref{fig:running_Exa}, Step 4). By choosing $(x_1,x_2)\in\R^2_{>0}$ with $x_1$ large enough,  the determinant of  $M(x)$ is positive 
 (Fig~\ref{fig:running_Exa}, Step 5). Therefore  Corollary~1 cannot be applied as $s=1$. We note that this conclusion is independent of the specific  choice of the parameter vector  $\kappa=(\kappa_1,\kappa_2,\kappa_3)$, so in fact it holds for all parameter values.

 Corollary~2 has the same assumptions as Corollary~1.
 We  find a positive parameterization by solving the equilibrium equation for $x_1$. That is, we treat it as an equation in $x_1$, while $x_2$ $(= \hat{x})$ is treated as a parameter. The  function $a(x_2 )$ obtained by substituting $x$ by $\Phi(x_2)$ in the determinant of $M(x)$ is given in Fig 2, Step 6. It is clear from the expression of $a(x_2)$ that it takes the sign $-1$ for all   $x_2>0$. Also this conclusion does not depend on the specific  value of $\kappa$.  By application of Corollary~2(A) with $s = 1$, we conclude that there exists a unique positive non-degenerate equilibrium in each stoichiometric compatibility class with $c > 0$, for all values of the reaction rate constants (Fig~\ref{fig:running_Exa}, Step 7). The possibility of multiple equilibria is therefore excluded. 
 In this particular example the existence of a positive parameterization is essential to draw the conclusion.

  \medskip
  
 To illustrate how Corollary 2 can be used to find parameter regions for multistationarity, we consider the polynomial $a(\widehat{x})$  given in the first row of Fig~\ref{fig:examples}, where $n=6$ and $s=4$ (the example is worked out in detail below):
         \begin{align*}
          a(\widehat{x})  &=  \frac{1}{\k_3} \big(  \k_{2}\k_{4}\k_{5}^2 \mathbf{\left(\k_{1}-\k_3
          \right) }x_{4}x_{5}^{2}    +  ( \k_{1}+\k_{2})\k_3\k_{4}\k_{5}\k_{6} {x_{5}^2}  
          + 2\k_{1}\k_{2}\k_3\k_{4}\k_{5} {x_{4}x_{5}} \\ & \quad + \k_{1}( \k_{2} +\k_{3})\k_3\k_{5}\k_{6}  {x_{5}}
+\k_{1}\k_{2}\k_{3}^2\k_{5} {x_{4}} + \k_{1}\k_{2}\k_{3}^2\k_{6}  \big).
        \end{align*}
 Only one of the coefficients of the polynomial $a(\hat{x})$ in $x_4,x_5$ can be negative. If $\k_3\leq \k_1$, then  $\sign( a(\hat{x}) )=(-1)^4=1$ for all positive $\hat x$. 
Corollary 2(A)  implies  that there is a unique positive non-degenerate equilibrium in each stoichiometric compatibility class with non-empty positive part.

 Oppositely, we show that for $\k_3> \k_1$,  Corollary 2(B) applies. For that,  let $x_4=T$ and $x_5=T$.  Then $a(\hat{x})$ becomes a polynomial in $T$ with negative leading coefficient of degree $3$.
For $T$ large enough,  $  a(\widehat{x})$ is negative, and we conclude that there exists $\hat x$ such that  $\sign(a(\hat{x}))= (-1)^{5}=-1$.  Corollary 2(B)  implies that there exists a stoichiometric compatibility class $\mP_{c}$ that contains at least two positive equilibria.  
In summary, the region of the parameter space for which multistationarity exists is completely characterized by the inequality $\k_3>\k_1$.

   \medskip
 Step 5 and 7 are sign analyses of $\det(M(x))$ and $a(\hat x)$, respectively. These are  crucial steps and  essential for determining parameter regions with mono- and multistationarity.   
   In general, the sign of a polynomial might be studied by studying the signs of the coefficients of the monomials in the polynomial. If all {coefficients have the same sign,} 
   then the polynomial {is either} positive or negative for all $x\in\R^n_{>0}$, respectively,  $\hat x\in\R^m_{>0}$, depending on the sign, and Corollary~1, respectively, Corollary~2(A) applies. If this is not the case, then Corollary~2(B) might be applicable {if we can show that the polynomial has the sign $(-1)^{s+1}$ for some $\hat x\in\R^m_{>0}$}. 
      In {the two examples discussed here, the signs of $\det(M(x))$ and $a(\hat x)$
  are straightforward to analyse. However, this is not always the case, see the section ``Checking the steps of the procedure".}
 
\begin{figure*}[!t]
  \centering
  \includegraphics[scale=0.8]{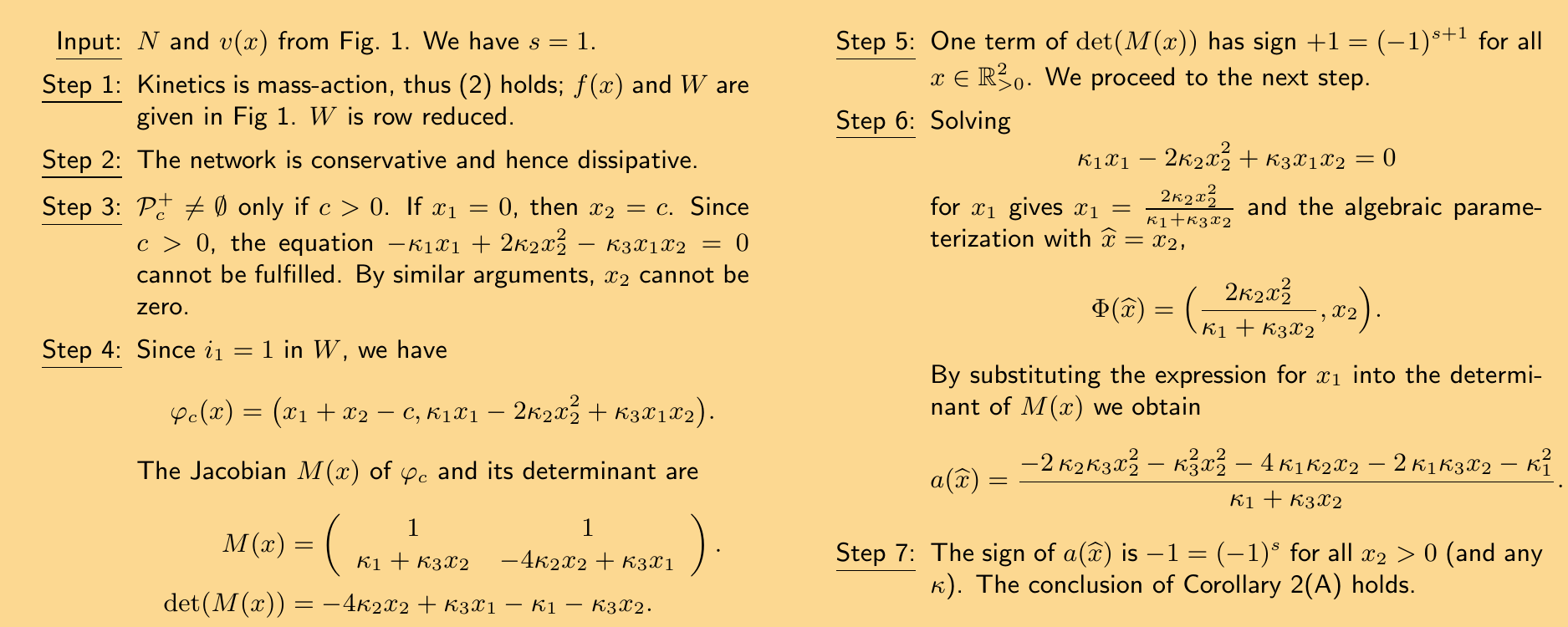}
  \caption{\label{fig:running_Exa}
      {    Step 1-3 check the assumptions of Corollary~1 and Corollary~2. In step 4 the function $\varphi_c(x)$ is constructed and the determinant of $M(x)$ is found. Step 5 is the sign analysis of the polynomial $\det(M(x))$ for $x\in\R^2_{>0}$.
      Step 6 establishes a positive parameterization and finds the polynomial $a(\hat x)$. Step 7 is similar to step 5, but for $a(\hat x)$.
   }}
\end{figure*}

\subsection*{Procedure for finding parameter regions for multistationarity}

 In the {previous subsection} we applied {Corollary~1 and Corollary~2} {to the running example} by going through a number of steps corresponding to the conditions of the statements and the calculation of the determinant. In this section we outline the steps formally.  Afterwards we discuss the steps and how they can be verified either manually or algorithmically, that is, without user intervention. Finally we  devise an algorithm to conclude  uniqueness of equilibria or to find regions in the parameter space where multistationarity occurs. We conclude this section with some extra examples that follow the steps of the procedure.
 
We assume  the reaction rate functions $v(x)$  depend on some parameters $\k$. The reaction rate functions are further assumed to be polynomials (as for mass-action 
kinetics) or quotients of polynomials (as for Michaelis-Menten and Hill kinetics with integer exponents).

{
 The input to  the procedure   is $v(x)$ and $N$ (the stoichiometric matrix) and the output is  parameter regions for which the  network admits multistationarity or uniqueness of equilibria. }  
 
\medskip
\noindent
{\bf Procedure (Identification of parameter regions for multistationarity) } 

\medskip
\noindent  {\bf Input:}  $N$ and $v(x)$ {depending on $\kappa$}.

\noindent  {\bf 1. } Find $f(x)$, a row reduced matrix $W$ {of size $d\times n$} such that {$WN=0$}, 
and check that $v(x)$ {vanishes in the absence of one of the reactant species,} 
that is, check that it satisfies Eq~\eqref{eq:invariance}.

\noindent {\bf 2. } Check that the network is dissipative.
 
\noindent {\bf 3. } Check for boundary equilibria in  $\mP_c$ for $\mP_c^+\neq \emptyset$ and $c\in\R^d$.

\noindent {\bf 4. } Construct $\varphi_c(x)$, $M(x)$ and compute $\det(M(x))$.

\noindent {\bf 5. } {Analyze the sign of  $\det(M(x))$.} {Find conditions on the parameters $\k$ such that 
 $\sign(\det(M(x)))=(-1)^s$ for all $x\in\R^n_{>0}$, in which case  Corollary~1 holds.
 }

{If Corollary~1 does not hold for all  $\k$, continue to the next step.}

\noindent {\bf 6. } Obtain an \textbf{algebraic} parameterization $\Phi(\hat{x})$ of the set of positive equilibria {for all $\kappa$}, as in Eq~\eqref{eq:param}, {  such that the coefficients of the numerator and the denominator of each $\Phi_i(\hat{x})$ possibly depend on  $\kappa$.} Compute $a(\hat{x})=\det(M(\Phi(\hat x)))$.   By hypothesis, $a(\hat{x})$ can be written as the quotient of two polynomials in $\hat{x}$ with coefficients depending on $\kappa$, whose denominator takes positive values.

\noindent {\bf 7. } {Analyze} the sign of the numerator  of $a(\hat{x})$.

{\bf 7a.}  { Identify coefficients with sign $(-1)^{s+1}$ and coefficients  that can have different signs depending on the parameters.}

{\bf 7b.} Use the terms corresponding to identified coefficients to construct parameter inequalities such that, whenever these inequalities hold, one has either $\sign(a(\hat{x}))=(-1)^s$ for all $\hat{x}\in\R^m_{>0}$  or  $\sign(a(\hat{x}))=(-1)^{s+1}$ for at least one  $\hat{x}\in\R^m_{>0}$, {in which case either {Corollary~2(A) or (B)}  holds.}

\medskip

There is no guarantee that all steps of the procedure can be carried
out successfully, let alone automatically.  While step 1 and 4
usually are straightforward ({only computational issues might
  arise for large networks}),  step 2, 3, 5, 6 and 7 might in
particular require case specific approaches.  
  However, there exist computationally feasible sufficient criteria that  guarantee the conditions in each step can be checked  efficiently.

\paragraph{{Checking the steps} of the procedure. }

\subparagraph{\textit{Step 2: establishing  dissipativity. }} 
If the network is not dissipative, then at least one concentration  grows to infinity over time. This is typically not the case for realistic networks, but it needs to be ruled out in order to apply  the procedure. 

We start by checking whether the network is conservative. This implies solving the linear system $\omega^t\, N = 0$ with the constraint $\omega >0$. 
Alternatively, conservation relations are often easily established by inspection of the reactions.
For example, in many signalling networks, the total concentration of enzyme (free and bounded) and of substrate (phosphoforms) are conserved.

{If the network is not conservative, then we check whether it is \emph{strongly endotactic} \cite{endotactic1,endotactic2}. Strongly endotactic reaction networks are in particular  \emph{permanent}, that is,  dissipative and the compact set can be chosen such that it does not intersect the boundary of $\R_{>0}^n$, see  \cite{endotactic1,endotactic2,johnston} for details. }

If the network is neither conservative nor strongly endotactic, then we can use the 
following proposition to decide on dissipativity  (see the \S3.2 in the \SI).

\medskip\noindent
{\bf Proposition 1 (Dissipative network). }
Let  $||\cdot ||$ be a norm in $\mathbb{R}^n$.
Assume that for each $c$ with $\mP_c^+\neq \emptyset$, there exists a
vector $\omega_c\in \R^n_{>0}$ and a number $R>0$ such that
$\omega_c \cdot f(x)< 0$ for all $x\in \mP_c$ with $||x||>R$. Then the
network is dissipative.

\medskip
Thus, we look for {vectors $\omega_c$ with all coordinates
  positive and} such that $\omega_c  \cdot f(x) <0$ for large $x$.   To avoid restricting the parameter values, this computation should be done
symbolically.

\subparagraph{\textit{Step 3: absence of  boundary equilibria. }}
For systems of moderate size  it is often possible to establish
nonexistence of boundary equilibria by arguments similar to those
employed in the analysis of the running  example: for each $i$, assume
$x_i=0$, and show that it leads to a contradiction.

A systematic procedure to check for the existence of boundary equilibria relies on computing the so-called {\em minimal siphons} of the network \cite{angelisontag}. 
 A \emph{siphon} is a set of species $Z\subseteq \{X_1,\dots,X_n\}$ fulfilling the following closure property: if $X_i\in Z$ and 
 $X_i$ is produced in reaction $R_j$ (that is, $\beta_{ij}>0$), then there exists  $X_k\in Z$ such that $X_k$ is consumed in the same reaction (that is, $\alpha_{kj}>0$).
 A \emph{minimal siphon} is a siphon that does not properly contain any other siphon. 

\medskip\noindent
{\bf Proposition 2 (Siphons)} (\cite{Shiu-siphons,marcondes:persistence})
If for every minimal siphon $Z$ there exists a subset $\{X_{i_1},\dots,X_{i_k}\}\subseteq Z$,  and a conservation relation   $\lambda_1 x_{i_1}+\dots+\lambda_k x_{i_k} = c$ for some  positive $\lambda_1,\dots,\lambda_k$, then the network has no boundary equilibria in any stoichiometric compatibility class $\mP_c$ with $\mP_c^+\neq \emptyset$.
 
\medskip
{
The hypothesis of the proposition can be summarised by saying that each minimal siphon contains the support of a positive conservation relation.}

{For example, the running example has only one minimal siphon, namely $\{X_1,X_2\}$. The conservation relation $x_1+x_2 = c $ fulfils the requirement of Proposition 2, and hence the network has no boundary equilibria in any $\mP_{c}$ with $c>0$.}

{
  More information about using siphons to preclude boundary equilibria is given in the section ``Computational issues'' below and in     \S5.1 of the \SI.
}

 \subparagraph{\textit{Step 5: determining the sign of $\det(M(x))$.}}
 {If the kinetics is mass-action, then $\det(M(x))$ is a polynomial in $x$. In general,  if the reaction rate functions are rational functions in $x$, then so is $\det(M(x))$.  In the latter case, if the $j$th reaction rate function fulfils $v_j(x)=p_j(x)/q_j(x)$ with $p_j(x)\ge0$ and $q_j(x)>0$ for all $x\in\R^n_{>0}$, then $\det(M(x))=p(x)/q(x)$, where $q(x)= \prod_{j=1}^\ell q_j(x)^2>0$. It follows from the definition of $M(x)$ and by differentiation of $v_j(x)$, $j=1,\ldots,\ell$.
 }
 
 {We determine conditions on the parameters such that all coefficients of $p(x)$ have sign $(-1)^s$. Then 
  the sign of $\det(M(x))$ is also $(-1)^s$ for all $x\in \R^n_{>0}$ and Corollary 1 holds. }

\subparagraph{\textit{Step 6: finding an algebraic positive parameterization. }}

Computer algebra systems like Maple or Mathematica can be used
to find a parameterization.  {One  strategy} is to  solve the equations $f_i(x) = 0$, $i\notin \{i_1,\dots,i_d\}$, for some subset of (at most) $s$ variables, treating the remaining (at least) $d$ variables  as coefficients of the system. 
If a  parameterization {found in this way} exists but is not positive,  another set of variables should be tried out. This can be systematically addressed by trying out  all possible subsets of variables.  It requires computation and analysis of at most
$\left(
  \begin{smallmatrix}
    n \\ d
  \end{smallmatrix}
\right)
$ 
parameterizations.  Alternatively, one can compute the circuits of degree one of the matroid associated with the equilibrium equations \cite{Gross:2015vo}.
 
In some cases, the network structure implies that a positive parameterization of the set of equilibria exists. 
A set, say $\{X_{k+1},\dots,X_{n}\}$ with $n-k$ elements for some $k$, is \emph{non-interacting} if  two species never appear on the same side of a reaction and they have coefficient at most one in all reactions. 
{In this case the equilibrium equations $f_{k+1}(x)=\dots=f_n(x)=0$ form a linear system in the variables $\{x_{k+1},\dots,x_{n}\}$. Provided that the determinant of the coefficient matrix of the linear system is not identically zero, this system can be solved } and we obtain a  positive parameterization of the non-interacting variables $x_{k+1},\dots,x_{n}$ at equilibrium   in terms of the remaining variables $x_{1}, \dots,x_{k}$ \cite{Fel_elim,saez_graph}.  {A necessary condition for the determinant of the coefficient matrix not being identically zero}  is that there is no conservation relation of the form $x_{i_1}+\dots+x_{i_l}$  with $i_1,\dots,i_l\in \{k+1,\dots,n\}$.
If {a non-interacting set with $k=d$ exists, that is, with $s=n-d$ elements}, then this guarantees the existence of the desired parameterization.   
{In the running example there is not a non-interacting set because both species have coefficient $2$ in the reaction $2X_1\rightarrow 2X_2$.}

The non-interacting condition can be relaxed in some cases by requiring that  none of the  species in $\{X_{k+1},\dots,X_{n}\}$ appear together in a reactant  (these sets are called \emph{reactant-non-interacting} \cite{meritxell-noninter}). Proceeding as above, provided that the determinant of the coefficient matrix is not identically zero, $x_{k+1},\dots,x_{n}$ can be expressed at equilibrium   in terms of  $x_{1}, \dots,x_{k}$. Conditions that ensure this is a positive parameterization 
are given in \cite{meritxell-noninter}. {In the running example, species $X_1$ is a reactant-non-interacting set and we can obtain a positive parameterization of $x_1$ in terms of $x_2$, see Fig \ref{fig:setup_running} and Fig \ref{fig:running_Exa}.} 
 
If the network admits so-called \emph{toric steady states}, then a positive parameterization also exists \cite{PerezMillan}.

\subparagraph{\textit{Step 7: the  sign of $a(\hat{x})$ and the Newton polytope.} }
This is perhaps the hardest step of all.  
{We write $a(\hat{x})=p(\hat x)/q(\hat x)$ with $q(\hat x)$ positive for all $\hat x$ and would like to determine the sign of $p(\hat{x})$.}
We first look for conditions that ensure  uniqueness of positive equilibria  by imposing that all coefficients of $p(\hat{x})$ {as a polynomial in $\hat x$} have sign $(-1)^s$. 

We next identify the monomials of $p(\hat x)$, 
{where the sign of the coefficient, say $\beta$,  is $(-1)^{s+1}$ for some parameter values. For each of these monomials we check whether the monomial can  ``dominate'' the sign of $p(\hat{x})$.
That is to say,  if $\sign(\beta)=(-1)^{s+1}$, then we determine whether there is an $\hat x$ such that $p(\hat{x})$ also has the sign $(-1)^{s+1}$.}
If it is the case, then the condition $\sign(\beta)=(-1)^{s+1}$ is a sufficient condition for multiple equilibria { according to Corollary~2(B)}.  

Given a coefficient of a monomial with sign $(-1)^{s+1}$, it might not be straightforward to decide if the polynomial $p(\hat x)$ has the same sign for some value of $\hat{x}$.  
(For example, the polynomial $x^2-2xy+y^2=(x-y)^2$ has one monomial with negative sign, but the polynomial itself can never be negative.) 
{When the number of variables is small, one can attempt to decide the sign as we did in the examples above. Otherwise, } 
 our strategy is to determine whether the monomial of interest corresponds to  a vertex of the \emph{Newton polytope}. {If that is the case, then the monomial can dominate the sign of $p(\hat x)$ (see $\S$5.2 in the \SI).}
The Newton polytope of $p(\hat{x})$ is defined as the convex hull of the exponent vectors $\alpha=(\alpha_1,\dots,\alpha_m)\in \R^m$ corresponding to the monomials $\hat x_1^{\alpha_1}\cdot \ldots \cdot \hat x_m^{\alpha_m}$ of $p(\hat{x})$. If $\alpha$ is a vertex of the Newton polytope, then there exists $\hat{x}\in \R^m_{>0}$ such that the sign of $p(\hat{x})$  agrees with the sign of the coefficient of the monomial (see $\S$5.2 in the \SI).

Fig~\ref{fig:examples} shows the Newton polytopes associated with the {polynomials} 
in the figure.
The vertex corresponding to the monomial of interest is shown  in red.

\paragraph{An algorithm. }

In the previous subsection we have outlined computational criteria that might be used to verify the conditions of  the steps in the procedure.
These computational criteria are only sufficient, that is, even if they fail  the   procedure might still work on the given network.  For example, a sufficient computational criterion for the absence of boundary equilibria is based on Proposition 2. However, it might happen that Proposition 2 cannot be applied, but that the network nonetheless has no boundary equilibria in stoichiometric compatibility classes with non-empty interior.

{
We  have collected   sufficient computational criteria that guarantee   the conditions  of the procedure are fulfilled. In this way  the procedure is formulated as an algorithm with decision diagram  shown in Fig \ref{fig:decision_diagram}.  }
{If one step of the algorithm fails, then we say that the algorithm ends indecisively.  In that case  we might check whether the step  can be verified by other means.}

 For simplicity, we have restricted to mass-action kinetics. Under this assumption, $\det(M(x))$ is a polynomial in   $x$ and the parameters $\k$, and $a(\hat x)$ is a  rational function in $\hat x$ and  $\k$  because the parameterization is assumed to be algebraic.

\begin{figure*}[!h]
\centering
\includegraphics[scale=0.7]{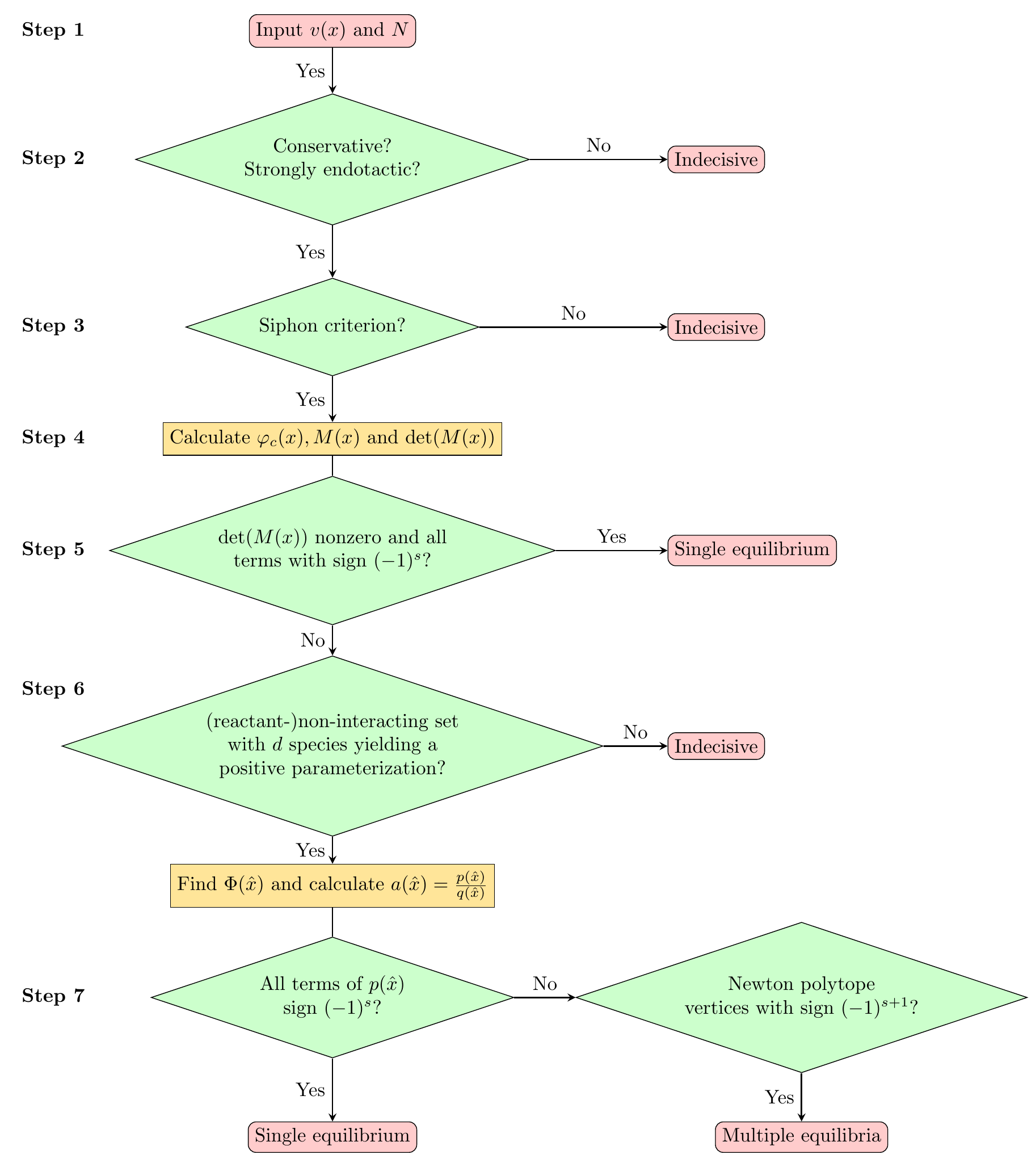}
  \caption{  \label{fig:decision_diagram}
  {  {\bf Decision diagram of the algorithm.} At each step either the condition is fulfilled or the algorithm terminates indecisively. If that is the case, the corresponding condition might still be verified manually  and  the algorithm resumed from the next following step.    }}
\end{figure*}

\section*{Applications to selected examples}

To illustrate several aspects of the algorithm we provide a detailed
step-by-step analysis of a collection  of examples.

\subsection*{Two-component system} 

We have chosen this example to illustrate the situation where 
an algebraic parameterization is not required, as already
$\det(M(x))$ is of constant sign. The algorithm therefore stops
successfully at Step~5 (and consequently skips Step~6 and~7).

We consider a simple version of a two-component system consisting of a
histidine kinase HK that autophosphorylates and transfers the
phosphate group to a response regulator RR, which undergoes
autodephosphorylation.  
The reactions of the network are 
\begin{align*}
    {\rm HK}   & \cee{->[\k_1]}    {\rm HK}_p &
       {\rm HK}_p+ {\rm RR} & \cee{->[\k_2]}{\rm HK} +{\rm RR}_p &
       {\rm RR}_p  &  \cee{->[\k_3]}{\rm RR}.
\end{align*}
We let  $X_1={\rm HK}$, $X_2={\rm HK}_{p}$, $X_3={\rm RR}$ and $X_4={\rm RR}_p$.
The stoichiometric matrix $N$ and a row reduced matrix
$W$ such that $WN=0$
are
\begin{displaymath}
  N=\left(
    \begin{array}{rrrr}  
      -1 & 1  & 0  \\
      1 & -1  & 0  \\ 
      0 & -1  & 1 \\
      0 & 1 & -1
    \end{array}
  \right),\qquad 
  W = \begin{pmatrix} 
    1 & 1 & 0 & 0 \\ 
    0 & 0 & 1 & 1 
  \end{pmatrix}.
\end{displaymath}
The matrix $W$ gives rise to the conservation relations
$ x_1+x_2 = c_1$ and $x_3+x_4 = c_2$. With mass-action kinetics, the vector of reaction rates is
$v(x) = (\kappa_1 x_1,\kappa_{2} x_2x_3,\kappa_{3} x_4),$
and the function $f(x)=Nv(x)$ is
\begin{align*}
  f(x) &= (
         -\kappa_1 x_1  + \kappa_{2} x_2x_3,
         \kappa_1 x_1  - \kappa_{2} x_2x_3   -\kappa_{2} x_2x_3    +\kappa_{3} x_4,
         \kappa_{2} x_2x_3     -\kappa_{3} x_4
         ).
\end{align*}
We apply the algorithm to this network.

\smallskip\noindent
{\bf Step 1. } Mass-action kinetics fulfills assumption in Eq~\eqref{eq:invariance} on the vanishing of reaction rate functions. The function $f(x)$ and $W$ are given above. The matrix $W$ is row reduced.

\smallskip\noindent
{\bf Step 2. }  The network is conservative since $(1,1,1,1)\in \im(N)^\perp$.  Therefore the 
network is dissipative.

\smallskip\noindent
{\bf Step 3.  } 
The minimal siphons of the network are $\{X_1,X_2\}$ and
$\{X_3,X_4\}$. These two sets are the supports of the conservation
relations.
By Proposition~2, 
there are no boundary equilibria in any $\mP_c$ as long as $\mP_c^+\neq \emptyset$.

\smallskip\noindent
{\bf Step 4.  } 
 With our choice of $W$, we have $i_1=1,i_2=3$.  
Hence  $\varphi_c$ is obtained by replacing the components $f_1(x),f_3(x)$ of
$f(x)$ by the expressions derived from the two conservation relations: 
\begin{align*}
  \varphi_c(x) &= \bigl( x_1+x_2- c_1, 
                 \k_1 x_1  - \k_{2} x_2x_3, x_3+x_4-c_2, \k_{2} x_2x_3 -\k_{3} x_4 \bigr).
\end{align*}
The Jacobian $M(x)$ of $\varphi_c$ and its determinant are 
      \begin{align*}
        M(x) & =  \left(
          \begin{array}{cccc}  
            1 & 1 & 0  & 0  \\
            \k_1 & -\k_2x_3 & -\k_2x_2 & 0 \\ 0 & 0 & 1 & 1  \\ 
            0 &  \k_2x_3 & \k_2x_2 & -\k_3
          \end{array}
        \right). \\
        \det(M(x)) & = \k_1\k_2x_2+ \k_2\k_3x_3+\k_1\k_3.
      \end{align*}

\smallskip\noindent
{\bf Step 5.  }      All terms of $\det(M(x))$ have sign $+1=(-1)^s$, since $s=2$, and thus the conclusion of Corollary~1 holds. The network admits exactly one non-degenerate equilibrium point in every stoichiometric compatibility class with non-empty positive part.

\subsection*{Hybrid histidine kinase}

This example has been analysed in \cite{feliu:unlimited}. Taken with
mass-action kinetics the network is known to be  multistationary for
specific choices of reaction rate constants. 
We have chosen this example to illustrate how the algorithm can be used
to sharpen known results: not only does it establish multistationarity
for some parameter values, it provides
precise conditions for when it occurs and allows a complete partition
of the parameter space into regions with and without
multistationarity. It also illustrates the use of an algebraic
parameterization,  which can be obtained by identifying sets of
reactant-non-interacting species, and the use of the Newton polytope in Step 7.

This reaction network is an extension of the two-component
system discussed above and it is given in the first row of
Fig~\ref{fig:examples}. Specifically,  the histidine kinase is
assumed to be \emph{hybrid}, that is, it has two ordered
phosphorylation sites \cite{feliu:unlimited}. Whenever the second
phosphorylation site is occupied,  the phosphate group can be
transferred to a response protein.

\begin{figure*}[!ht]
\includegraphics[scale=0.85]{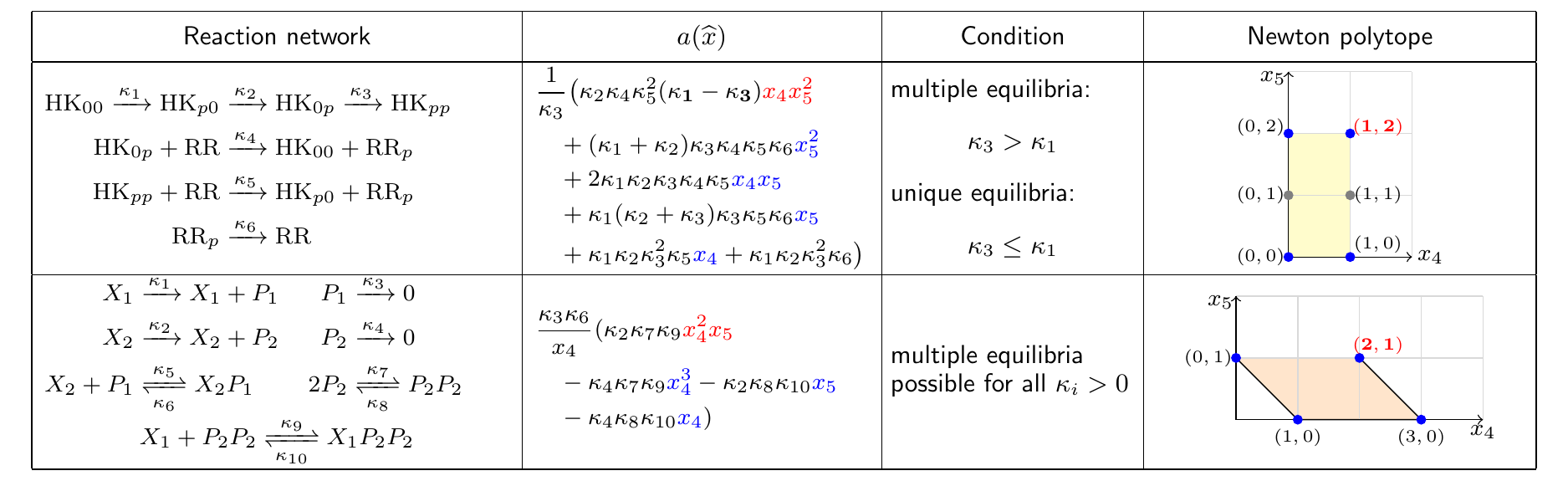}
  \caption{\label{fig:examples} 
    {
      Two examples describing a hybrid histidine kinase (row 1) and a
      gene transcription network (row 2). \emph{Column 1}: the reaction network;
      \emph{Column 2}: the function $a(\hat{x})$ where monomials with
      coefficients of  constant sign $(-1)^s$ are in  black, and those
      that can have sign $(-1)^{s+1}$ are in red; \emph{Column 3}:
      parameter conditions for multistationarity; \emph{Column 4}: Newton
      polytope where each point corresponds to the exponent vector  of
      a monomial of the numerator of $a(\hat{x})$ (e.g. $(1,2)$ is the
      exponent vector of the monomial $x_4x_5^2$), black
      points are the vertices of the Newton polytope and red numbers
      indicate the exponents of the red monomials in column 2.
    }
  }
  \end{figure*}

Using the notation $X_1={\rm HK}_{00}$, $X_2={\rm HK}_{p0}$, $X_3={\rm
  HK}_{0p}$, $X_4={\rm HK}_{pp}$,  $X_5={\rm RR}$ and $X_6={\rm
  RR}_p$, the stoichiometric matrix $N$ and a row reduced matrix
$W$ such that $WN=0$ are
\begin{displaymath} 
  N = 
  \left(
    \begin{array}{rrrrrr}
      -1 & 0 & 0 & 1 & 0 & 0 \\ 
      1 & -1 & 0 & 0 & 1 & 0 \\
      0 & 1 & -1 & -1 & 0 & 0 \\
      0 & 0 & 1 & 0 & -1 & 0 \\
      0 & 0 & 0 & -1 & -1 & 1 \\
      0 & 0 & 0 & 1 & 1 & -1
    \end{array}
  \right),\qquad 
  W = 
  \left(
    \begin{array}{rrrrrr}
      1 & 1 & 1 & 1 & 0 & 0 \\ 
      0 & 0 & 0 & 0 & 1 & 1
    \end{array}
  \right).
\end{displaymath}
The matrix $W$ gives rise to the conservation relations
  $x_1+x_2+x_3+x_4 = c_1$ and $x_5+x_6 = c_2.$ 
We assume mass-action kinetics 
$$ v(x)= (\k_1 x_1, \k_2 x_2, \k_3x_3,\k_4 x_3x_5, \k_5x_4x_5, \k_6 x_6),$$
and  the function is 
\begin{align*}
  f(x) &= 
         (-\kappa_1 x_1  + \kappa_{4} x_3x_5 ,  \kappa_1 x_1  - \kappa_{2} x_2 +\kappa_5 x_4 x_5 , -\kappa_{3} x_3    +\kappa_{2} x_2 -\kappa_4 x_3 x_5 ,  \\ & \qquad \kappa_{3} x_3     -\kappa_{5} x_4 x_5, -\kappa_{4} x_3x_5    -\kappa_{5} x_4x_5  +\kappa_{6} x_6, \kappa_{4} x_3x_5     -\kappa_{6} x_6 +\kappa_5 x_4 x_5).
\end{align*}

We apply the algorithm to this network.

\smallskip\noindent
{\bf Step 1. } Mass-action kinetics fulfills assumption in Eq~\eqref{eq:invariance} on the vanishing of reaction rate functions. The function $f(x)$ and $W$ are given
above. The matrix $W$ 
is row reduced.

\smallskip\noindent
{\bf Step 2. }  Since $(1,1,1,1,1,1)\in \im(N)^\perp$ the network is conservative and hence dissipative.

\smallskip\noindent
{\bf Step 3.  } 
The network has two minimal siphons $\{X_1,X_2,X_3,X_4\}$ and
$\{X_5,X_6\}$, which are  respectively the supports of the two conservation
relations. 
We apply  Proposition~2 to conclude that 
there are no boundary equilibria in any $\mP_c$ as long as $\mP_c^+\neq \emptyset$.
 
\smallskip\noindent
{\bf Step 4.  } 
Since $i_1=1,i_2=5$, the function $\varphi_c$ is obtained by replacing the components $f_1(x),f_5(x)$ of
$f(x)$ by the expressions derived from the two conservation relations: 
 \begin{align*}
 \varphi_c(x) &= 
(x_1+x_2+x_3+x_4-c_1 ,  \kappa_1 x_1  - \kappa_{2} x_2 +\kappa_5 x_4 x_5 , -\kappa_{3} x_3    +\kappa_{2} x_2 -\kappa_4 x_3 x_5 ,  \\ & \qquad \kappa_{3} x_3     -\kappa_{5} x_4 x_5, x_5+x_6-c_2  , \kappa_{4} x_3x_5     -\kappa_{6} x_6 +\kappa_5 x_4 x_5).
 \end{align*}
The Jacobian $M(x)$ of $\varphi_c(x)$ and its determinant are 
\begin{align*} M(x) & = \left(
\begin{array}{cccccc}
1 & 1 & 1 & 1 & 0 & 0 \\
\k_1 & -\k_2 & 0 & \k_5x_5 & \k_5x_4 & 0 \\
0 & \k_2 & -\k_3- \k_4x_5 & 0 & -\k_4x_3 & 0 \\
0 & 0 & \k_3 & -\k_5x_5 & -\k_5x_4 & 0 \\ 
  0 & 0 & 0 & 0 & 1 & 1 \\
  0 & 0 & \k_4x_5 & \k_5x_5 & \k_4x_3+\k_5x_4 & -\k_6
\end{array}
\right),\\
\det (M(x)) &= 
\kappa_2\kappa_4\kappa_5(\kappa_1- \kappa_3)x_3x_5  + \kappa_1\kappa_2\kappa_4\kappa_5x_4x_5 +\kappa_4\kappa_5\kappa_6(\kappa_1+\kappa_2)x_5^2 \\ &\quad  + \kappa_1\kappa_2\kappa_3\kappa_4x_3 + \kappa_1\kappa_2\kappa_3\kappa_5x_4  + \kappa_1\kappa_5\kappa_6(\kappa_3+\kappa_2)x_5 + \kappa_1\kappa_2\kappa_3\kappa_6.
 \end{align*}
 
\smallskip\noindent
{\bf Step 5.  }   The sign of the first coefficient of $\det (M(x))$   depends on the parameters. If $\kappa_1\ge \kappa_3$, then the sign is positive and $\det (M(x))$ has sign $+1=(-1)^4$ ($s=4$) as the remaining  terms are  positive. According to Corollary~1, there is a single non-degenerate equilibrium in each stoichiometric compatibility class with non-empty positive part. If  $\kappa_1< \kappa_3$, then Corollary~1 cannot be applied.
We proceed to the next step to investigate the parameter space further.

\smallskip\noindent
{\bf Step 6.  }  {The set $\{X_1,X_2,X_3,X_6\}$ is reactant-non-interacting and consists of $s=4$ elements. }
We solve the equilibrium equations {$f_1=f_2=f_3=f_6=0$} for
$x_1,x_2,x_3,x_6$. This gives the following algebraic parameterization $\Phi \colon
\R^{2}_{> 0} \rightarrow  {V\cap } \R^6_{> 0} $ of the set of equilibria in terms
of $\widehat{x}=(x_4,x_5)$:
 \begin{align*}
  \Phi(x_4,x_5) & = \Big(\frac {\k_{4}\k_{5}x_4x_5^2}{\k_1\k_3},
  \frac{\k_{5}( \k_{4}x_{5}+\k_{3}) x_4x_5}{\k_{2}\k_{3}}, \frac{\k_{5}x_{4}x_{5}}{\k_{3}}, x_4,x_5,\frac {\k_{5}( \k_{4}x_{5}+\k_{3})x_{4}x_{5} }{\k_{3}\k_{6}}
 \Big).
 \end{align*}
The function $a(\hat{x})$, which is     $\det(M(x)) $ evaluated at 
$\Phi(x_4,x_5)$, is the polynomial given in the first row of Fig~\ref{fig:examples}.

 \smallskip\noindent
 {\bf Step 7.  } We assume $\kappa_1< \kappa_3$, as the case $\kappa_1\ge \kappa_3$ is analysed in Step 5.
Only one coefficient of  $a(\hat{x})$ has sign $-1=(-1)^{s+1}=(-1)^5$.
The monomial associated with this term is  $x_4x_5^2$. As the point $(1,2)$ (the degrees of the monomial) is a vertex of the Newton polytope (see Fig~\ref{fig:examples}), then there exists $\hat{x}\in\R^m_{>0}$ such that the sign of $a(\hat{x})$ is $-1$. 
Corollary~2(B)  implies that there exists  $c=(c_1,c_2)$ such that $\mP_{c}$ contains at least two positive equilibria.

Multistationarity is thus completely characterized by the inequality $\k_3>\k_1$.
This condition states
that the reaction rate constant for   phosphorylation of the first
site of the hybrid kinase is larger if the second site is phosphorylated than if it is not.

\subsection*{Gene transcription network}

We  consider the gene transcription network given in row 2 of
Fig~\ref{fig:examples}. This example has been studied in
\cite{gaskins:royal}. The particularities of this example are that the
network is dissipative but not conservative,  and  that it displays multistationarity for \emph{all}  parameters $\kappa$. Further, this network illustrates  the situation where the algorithm stops inconclusively at some step, but
can be resumed after successful manual verification. 

The network  represents a
gene transcription motif with two proteins $P_1,P_2$, 
produced by their respective genes $X_1,X_2$, and such that $P_2$
dimerises  \cite{gaskins:royal}. Further, the proteins cross regulate
each other as depicted in Fig~\ref{fig:examples}. 
Using the notation  $X_1=X_1$, $X_2=X_2$, $X_3=P_1$, $X_4=P_2$,
$X_5=X_2P_1$, $X_6=P_2P_2$, and $X_7=X_1P_2P_2$
the stoichiometric matrix $N$ 
and a row reduced matrix
$W$  such that $WN=0$
are
\begin{displaymath}
  \begin{split}
      N & =
      \left(
        \begin{array}{rrrrrrrrrr}
          0 & 0 & 0 & 0 & 0 & 0 & 0 & 0  & -1 & 1 \\ 
          0 & 0 & 0 & 0 & -1 & 1 & 0 & 0 & 0 & 0  \\
          1 & 0 & -1 & 0 & -1 & 1 & 0 & 0 & 0 & 0  \\
          0 & 1 & 0 & -1 & 0 & 0 & -2 & 2 & 0 & 0  \\
          0 & 0 & 0 & 0 & 1 & -1  & 0 & 0 & 0 & 0 \\
          0 & 0 & 0 & 0 & 0 & 0 & 1 & -1 & -1 & 1 \\
          0 & 0 & 0 & 0 & 0 & 0 & 0 & 0 & 1 & -1
        \end{array}
      \right), \\
      W &= 
      \left(
        \begin{array}{rrrrrrr}
          1 & 0 & 0 & 0 & 0 & 0 & 1 \\ 
          0 & 1 & 0 & 0 & 1 & 0 & 0
        \end{array}
      \right).
  \end{split}
\end{displaymath}
{From $W$ we find the 
conservation relations
  $x_1+x_7=c_1$ and $x_2+x_5=c_2$. Here $s=5$.}
{We consider mass-action kinetics such that
$$ v(x)= (\k_1 x_1, \k_2 x_2, \k_3x_3,\k_4 x_4, \k_5x_2x_3, \k_6 x_5,\k_7x_4^2, \k_8x_6,\k_9x_1x_6,\k_{10}x_7)$$
and $f(x)=Nv(x)$ is the function}
\begin{align*}
  &( -\k_{9}x_{1}x_{6}+\k_{10}x_{7}, -\k_{5}x_{2}x_{3}+\k_{6}x_{5}, 
         \k_{1}x_{1}-\k_{3}x_{3}  -\k_{5}x_{2}x_{3}+\k_{6}x_{5},\k_{2}x_{2}-\k_{4}x_{4} \\
       & \quad 
         -2\k_{7}x_4^{2}+2\k_{8}x_{6},     \k_{5}x_{2}x_{3}-\k_{6}x_{5},
         \k_{7}x_4^{2}-\k_{8}x_{6}-\k_{9}x_{1}x_{6}+\k_{10}x_{7},
         \k_{9}x_{1}x_{6}-\k_{10}x_{7}).
\end{align*}
We apply the algorithm to this network:

\smallskip\noindent
{\bf Step 1. } Mass-action kinetics fulfills assumption  in Eq~\eqref{eq:invariance} on the vanishing of reaction rate functions. The function $f(x)$ and $W$ are given
above. The matrix $W$ 
is row reduced.

\smallskip\noindent
{\bf Step 2. }  The network is neither conservative nor strongly
endotactic. Thus the algorithm terminates inconclusive. We take a
manual approach: we pick $\omega_c = (1,1,1,1,2,2,3)\in \R^7_{>0}$ and
observe 
\begin{displaymath}
  \omega_c\cdot f(x)= \k_1 x_1 + \k_2 x_2 - \k_3 x_3 - \k_4 x_4.
\end{displaymath}
Note that $x_1,x_2$ are bounded (due to the conservation relations) while  
 $x_3,x_4$  can be arbitrarily large. Then, for $x_3,x_4$ large
 enough, $\omega_c\cdot  f(x)<0$ and the network is dissipative by
 Proposition~1 (as has been shown in \cite{gaskins:royal} by other means).

\smallskip\noindent
{\bf Step 3.  }
{This network has two minimal siphons: $\{X_1,X_7\}$ and $\{X_2,X_5\}$, which are the supports of the two conservation relations. Therefore, by Proposition 2, there are no boundary steady states in stoichiometric compatibility class with non-empty positive part.}

{In section \S5.1 in the \SI we illustrate how to apply a simplification technique, based on the removal of so-called \emph{intermediates} and \emph{catalysts}, to check whether Proposition~2 holds for this network. }

 \smallskip\noindent
{\bf Step 4.  } 
Using that $i_1=1,i_2=2$ for our choice of $W$, the function $\varphi_c(x)$ is:
 \begin{align*}
 &
 (x_1+x_7-c_1,x_2+x_5-c_2, 
          \k_{1}x_{1}-\k_{3}x_{3},
          -\k_{5}x_{2}x_{3}+\k_{6}x_{5}\k_{2}x_{2}-\k_{4}x_{4} \\
          & \quad 
          -2\k_{7}x_4^{2}+2\k_{8}x_{6},     \k_{5}x_{2}x_{3}-\k_{6}x_{5},
          \k_{7}x_4^{2}-\k_{8}x_{6}-\k_{9}x_{1}x_{6}+\k_{10}x_{7},
          \k_{9}x_{1}x_{6}-\k_{10}x_{7}).
 \end{align*}
 The matrix $M(x)$ and its determinant are:
\begin{align*}
M(x) & = {\small \left( \begin {array}{ccccccc} 1&0&0&0&0&0&1\\ 0&1&0
&0&1&0&0\\ \k_{{1}}&-\k_{{5}}x_{{3}}&-\k_{{5}}x_{{2}}-\k_{{3}}&0&\k_{{6}}&0&0\\ 0&\k_{{2}}&0&-4\k_{{7}}x_{{4}}-\k_{{4}}&0&2\k_{{8}}&0\\ 0&\k_{{5}}x_
{{3}}&\k_{{5}}x_{{2}}&0&-\k_{{6}}&0&0\\ -\k_{{9}}x_{{6}}&0&0&2\,\k_{{7}}x_{{4}}&0&-\k_{{9}}x_{{1}}-\k_{{8}}&\k_{{10}}
\\ \k_{{9}}x_{{6}}&0&0&0&0&\k_{{9}}x_{{1}}&-\k_{{10}}\end {array} \right), }
\\
  \det(M(x)) &=
2\k_1\k_2\k_5\k_7\k_9x_1x_2x_4
    -\k_{{3}}\k_{{4}}\k_{{5}}\k_{{8}}\k_{{9}}x_{{3}}x_{{6}} 
    -\k_{{3}}\k_{{4}}\k_{{5}}\k_{{8}}\k_{{10}}x_{{3}} \\
    &\qquad 
    -\k_{{3}}\k_{{4}}\k_{{6}}\k_{{8}}\k_{{9}}x_{{6}}
    -\k_{{3}}\k_{{4}}\k_{{6}}\k_{{8}}\k_{{10}}.
\end{align*}

 \smallskip\noindent
{\bf Step 5.} One coefficient of $\det (M(x))$ has sign $(-1)^{s+1}=
1$ {for all values of $\k$. Thus we proceed to the next step}.

 \smallskip\noindent
{\bf Step 6.}   There is not a  set of non-interacting species {nor reactant-non-interacting} with $s=5$ elements. Thus the algorithm terminates inconclusively.

We take a manual approach and solve the equilibrium equations $f_3=f_4=f_5=f_6=f_7=0$ for
$x_1,x_2,x_3,x_6,x_7$. This gives the following algebraic parameterization $\Phi \colon
\R^{2}_{> 0} \rightarrow  \R^7_{> 0} $ of the set of equilibria   in terms
of $\widehat{x}=(x_4,x_5)$:
 \begin{align*}
  \Phi(x_4,x_5) = \Big(\frac{\k_{2}\k_{3}\k_{6}x_{5}}{\k_{1}\k_{4}\k_{5}x_{4}},
        \frac {\k_{4}x_{4}}{\k_{2}}, \frac {\k_{2}\k_{6}x_{5}}{\k_{4}\k_{5}x_{4}},
        x_4,x_5,
        \frac {\k_{7}x_4^2}{\k_{8}},  \frac
        {\k_{2}\k_{3}\k_{6}\k_{7}\k_{9}x_{4}x_{5}}{\k_{1}\k_{4}\k_{5}\k_{8}\k_{10}} 
        \Big).
 \end{align*}
Evaluating $\det (M(x))$ at $\Phi(x_4,x_5)$ we   obtain the polynomial 
 \begin{align*}
a(x_4,x_5) &= \frac{\k_{3}\k_{6}}{x_4}( \k_{2}\k_{7}\k_{9}x_4^2x_5
        -\k_{4}\k_{7}\k_{9}x_4^3 -\k_{2}\k_{8}\k_{10}x_{5}
        -\k_{4}\k_{8}\k_{10} x_{4}).
 \end{align*}

\smallskip\noindent
{\bf Step 7.} 
{The coefficient of the monomial  $x_4^2x_5$ of the numerator of $a(x_4,x_5) $ has sign $(-1)^{s+1}=(-1)^6=1$. 
Since the monomial $x_4^2x_5$  is a vertex of the associated Newton polytope (see Fig~\ref{fig:examples}), there exists $(x_4,x_5)\in\R^2_{>0}$ such that the sign of $a(x_4,x_5)$ is $1$. }
We conclude from Corollary~2(B) that for all $\k_i>0$ there exists $c$ such that $\mP_{c}$ contains at least two positive equilibria.

\subsection*{Special classes of networks}

There are several classes of networks for which some of the steps of the procedure  are automatically  fulfilled. 
We review some of them here. 

Post-Translational Modification (PTM) networks consist of enzymes ($E_i$), substrates ($S_i$) and intermediate  species  ($Y_i$) \cite{TG-rational}. Allowed reactions are of the form
$$ E_i+S_j \rightarrow Y_k, \qquad Y_k \rightarrow E_i+S_j,\qquad Y_j\rightarrow Y_i.$$
All intermediates are assumed to be the reactant, respectively,  the product   of some reaction.
These networks are conservative (hence dissipative) and boundary equilibria are precluded provided the underlying substrate network obtained by ignoring enzymes and intermediates is strongly connected  \cite{marcondes:persistence}, see also \S5.1 in the \SI.  When equipped with mass-action kinetics, these networks have a non-interacting set with $d$ elements consisting of all enzymes and some of the substrates, namely one per (minimal) conservation relation involving the substrates  \cite{TG-rational}. Thus, a positive parameterization can always be found under the conditions stated above in Step 6.  The class of PTM networks is contained in the class of cascades of PTM networks. Also this class admits  a positive parameterization in terms of the concentrations of the enzymes and some of the substrate forms  \cite{fwptm}.
 
 Cascades of PTM networks might  further be generalized to so-called MESSI networks  \cite{PerezMESSI}.
These networks are all conservative. 
Easy-to-check conditions for the absence of boundary equilibria and to decide whether the network admits toric steady states (and hence a positive parameterization) are given in  \cite{PerezMESSI}.

A   class  of networks that cannot have boundary equilibria in any stoichiometric compatibility class with non-empty interior is given
in \cite{Gnacadja:2011jw}.

{The two examples in Table~\ref{tab:cherry} are both PTM networks. Hence  they are conservative and positive parameterizations exist. {The underlying substrate network is strongly connected  (they pass the criterion based on minimal siphons)}. For both networks the conditions shown in Table 1 are obtained by the algorithm. See \S 6.1 and \S 6.2 in the \SI}.
{For illustration purposes, we apply the algorithm  in  \S6.3 of the \SI to an additional network and show that it is monostationary.}

\section*{Computational issues }

The computational complexity of some of the steps in the procedure are demanding. Some conditions can be checked using linear algebra and do not depend on parameter values, others depend on parameter values and require symbolic manipulations.  In some situations, the calculation can be done for even  large networks at the cost of time, while in other situations symbolic software (like Mathematica and Maple) have inherent limits to what it can process. We offer here  a few remarks about computational strategies and time complexity.

\begin{enumerate}[(1)]
\item Dissipativity. There are efficient algorithms to  check whether the network is conservative and strongly endotactic, using linear algebra or mixed-integer linear programming \cite{johnston,control}. We are not aware of a systematic way to check if Proposition 1 is fulfilled or not.

\item Finding the minimal siphons of a network requires in general exponential time and there might be exponentially many of these \cite{Nabli2016}. Different algorithms developed in Petri Net theory can be applied to find the minimal siphons; see for example \cite{Nabli2016,Shiu-siphons,angelisontag} and references therein.  The complexity of this computation can often be substantially reduced  by   removing so-called intermediates and catalysts from the network \cite{marcondes:persistence} (see  $\S$5.1 in the \SI for details).
 
\item Finding  all non-interacting and reactant-non-interacting sets requires in general exponential time. 
One strategy is the following.
We first remove all species $S_i$ for which $\alpha_{ij}>1$ or $\beta_{ij}>1$ for some reaction $R_j$ (the latter constraint is omitted if we are looking for reactant-non-interacting sets only). Then we build non-interacting (reactant-non-interacting) sets  by adding  new species  recursively until no more species can be added without having an interacting pair of species in the set. 
\item Calculation of the symbolic determinant of the matrix $M(x)$, and hence also of $a(\hat x)$, often fails in our experience for networks with more than 20 variables on common laptops\cite{feliu-bioinfo}. However, this clearly depends on the sparsity of the matrix $M(x)$, that is, on the number and order of the reactions. Strategies to reduce the complexity of the computation by expanding the determinant along the non-symbolic rows (conservation relations) were inspected in \cite{feliu-bioinfo}. 
\item 
Positive parameterizations: The worst case scenario involves checking $\sum_{i=1}^d \binom{n}{i}$ different sets of variables, each with at most $d$ variables.
\item Finding the vertices of the Newton polytope can be done with existing symbolic software, for example  Polymake \cite{polymake} or Maple, as we demonstrate in the \SI\!\!. 
\end{enumerate}

We stress that it is always  beneficial to guide the procedure/algorithm whenever possible in the sense that, if something is known for the network, there is no reason to go through many possibilities.

\section*{Discussion}

The main result of this paper, the {procedure} to identify parameter
regions for unique and multiple equilibria, combines Brouwer degree
theory and algebraic geometry.  In particular,  under the assumptions of {Corollary~2}, we show  that there exist stoichiometric compatibility classes with at least two equilibria if, and only if, a certain multivariate polynomial can attain a specific sign. 

Discriminating regions of the parameter space where multistationarity occurs is a hard mathematical problem, theoretically addressable by   computationally expensive means \cite{rob-024}. Our {approach} beautifully overcomes these difficulties by building on a simple idea,  the computation of the Brouwer degree of a function related to a dissipative network. Additionally, not only closed-form expressions in the parameters are obtained, but, as illustrated in examples, these expressions are often interpretable in biochemical terms, providing an  explanation of \emph{why} multistationarity occurs.

 {The procedure applies theoretically to any choice of algebraic reaction rate functions. However, in practice, the procedure works well with mass-action kinetics. For example, we have considered the two-site phosphorylation cycle depicted in the second row of Table~\ref{tab:cherry}, but now modelled with Michaelis-Menten kinetics instead of mass-action kinetics. This network is known to be multistationary \cite{Hell:2015iq}, and the conditions to apply {Corollary~1 and Corollary~2} are  valid. However, a positive algebraic parameterization does not exist, and hence our approach cannot be used to find parameter conditions for multistationarity. }
 
 However,  {Corollary~1} might  be used with rational reaction rate functions for monostationary networks. This is the case for example for {the one-site phosphorylation cycle $S\cee{<=>} S_p$} with Michaelis-Menten kinetics \cite{Hell:2015iq}. This network has two species and rank one. The sign of $\det(M(x))$ is $-1$ for all parameter values and all $x\in \R^2_{>0}$.  By {Corollary~1}, the network admits exactly one positive equilibrium in every stoichiometric compatibility class $\mP_c$ with $\mP_c^+\neq \emptyset$ for all parameter values.
 
 If a reaction network does not have any conservation relation, then the set of equilibria consists typically of a finite number of points. In this case an algebraic parameterization is an algebraic expression of the equilibria in terms of the parameters of the system. {Since $m=0$, then $\R^m$ consists of a single point and it follows directly that there is a unique equilibrium.  } Such an expression rarely exists. Therefore the procedure applies mainly to reaction networks with conservation relations. In particular, this rules out reaction networks where each species is produced and degraded.

Several natural questions remain outside the reach of our {procedure}. Firstly one would like to determine the particular stoichiometric compatibility classes  for which there are multiple equilibria. As stated in {Corollary~2}, if $\sign(a(\hat x)) =(-1)^{s+1}$, then $c:=W\Phi(\hat x)$ defines a stoichiometric compatibility class  with multiple equilibria. However, this only establishes $c$ indirectly through $\hat x$. {In some situations, it might be possible to find a positive parametrization that uses some of the conservation relations (ideally, all but one) and the stoichiometric compatibility classes with multiple/single equilibria  would be determined up to a single parameter. }

Secondly, one could ask for parameter regions that differentiate between the precise number of equilibria (that is, $0,1,2,\ldots$). This question should be seen in conjunction with the previous question: in typical examples,  when there are two equilibria in a particular stoichiometric compatibility class,  then there exists another class for which there are three. Hence the  number of equilibria cannot be separated from the  stoichiometric compatibility classes.

{A third question concerns the stability of the equilibria, which cannot be obtained from our procedure. It is, however, known that if the sign of the Jacobian evaluated at an equilibrium is $(-1)^{s+1}$, then it is unstable \cite{feliu2012}.  This is in particular the case for an equilibrium fulfilling the condition in {Corollary~2(B).}}

{We have shown that for some reaction networks our procedure can be formulated as an algorithm.}
We consider therefore our research a step in the direction of providing `black box tools' to analyse complex dynamical systems. Such tools would easily find their use in systems and synthetic biology, where it is commonplace to consider (many) competing models. A particular problem is to exclude  models that cannot explain observed qualitative features, such as multistationarity.

\paragraph*{Supporting Information:}
{\bf Proof of mathematical statements and examples.} In this document we first prove the claims  of the main text. 
Next, we provide   details on how to check the steps of the
procedure. 
{
  Finally, we give details of the examples  in
  Table~\ref{tab:cherry} and include an extra example which is a PTM network.
}

\paragraph*{Acknowledgments}
EF and CW acknowledge funding from the Danish Research Council of Independent Research. Alicia Dickenstein, Timo de Wolff and Bernd Sturmfels are thanked for discussions and the idea to use the vertices of the Newton polytope to study the sign of polynomials. {Meritxell Sáez, Amirhossein Sadeghi Manesh, Anne Shiu and Angélica Torres are thanked  for their comments on preliminary versions of the manuscript.}

\newpage

\begin{center}
{\Large \bf Identifying parameter regions for multistationarity

\medskip

}

\bigskip

\bigskip

{\large \bf Supplementary Information }

\bigskip

Carsten Conradi, Elisenda Feliu, Maya Mincheva, Carsten Wiuf
{\large }

 \thispagestyle{plain}

\medskip
\today

\end{center}
\bigskip
In this document we prove the claims  of the main text. 

Sections \ref{sec:notation} to \ref{sec:multi} focus on the proofs of the {theorem and its corollaries}  in the main text.
We start by  introducing some preliminaries before recapitulating  the main facts about Brouwer degree theory. Then we compute the Brouwer degree for a special class of functions (Theorem~\ref{theo:deg_g_forward}). We proceed to introduce the necessary background on  reaction networks and to state and prove a key result regarding the Brouwer degree of a reaction network with a dissipative semiflow (Theorem~\ref{thm:dissipative}). In Section~\ref{sec:multi} we use Theorem~\ref{thm:dissipative} to prove Theorem 1 of the main text.   The {first four sections of the document are self-contained and  do} not require parallel reading of the main text. For this reason some parts of the main text are repeated here for convenience.

Subsequently in Section~\ref{sec:steps}, we provide details on how to check the steps in the procedure of the main text. 
In Section~\ref{sec:mainex}  we give details of the examples  in the main text {and apply the algorithm to an extra network that is monostationary.}

\tableofcontents

\section{Preliminaries}
\label{sec:notation}

\subsection{Convex sets}
We let $\R^n_{\geq 0}$ denote the non-negative orthant of $\R^n$ and $\R^n_{> 0}$ denote the positive orthant of $\R^n$.

For a subset $B$  of $\R^n$, we let $\bd(B)$ denote the boundary of $B$ and $\cl(B)$ the closure of $B$, such that $\cl(B) =\bd(B) \cup B$.  If $B$ is open, then $\bd(B)\cap B=\emptyset$. If $B$ is bounded, then $\cl(B)$ is compact.

A set $B$ is \emph{convex} if the following holds:
$$ \mbox{if}  \quad x_1, x_2 \in B  \quad\mbox{then}  \quad \lambda x_1 + (1-\lambda) x_2 \in  B \quad\mbox{for all}  \quad 0\leq \lambda \leq 1.$$ 

Let $B\subseteq \R^n$ be a convex set. We say that $v\in \R^n$ \emph{points inwards} $B$ at $x\in \bd(B)$ if    $x+\epsilon v\in \cl(B)$ for all $\epsilon>0$ small enough. In particular, $v=0$ points inwards $B$ at all $x\in \bd(B)$. If $v$ points inwards $B$ at $x\in \bd(B)$, then it also points inwards $\cl(B)$ at  $x\in \bd(B)$.
The vector $v$ \emph{points outwards} $B$ at $x\in \bd(B)$, if it does not point inwards $B$ at $x\in \bd(B)$.

We will use the following  facts about convex sets. 

\begin{lemma}\label{lem:inwards}
Let $B\subseteq\R^n$ be a convex set. Then the following holds:
\begin{itemize}
\item[(i)]  The closure $\cl(B)$ of $B$ is convex.
\item[(ii)]   \label{fact2}
Assume $B$ is  open and consider $x_1\in B$, $x_2\in\bd(B)$. Let 
$$[x_1,x_2)=\{tx_1+(1-t)x_2 \st 0< t\le 1\}$$
be the half-closed line segment between $x_1$ and $x_2$. Then 
$[x_1,x_2)\in B$. 
\item[(iii)] \label{fact3}
Let $x_1\in B$ and $x_2\in\bd(B)$. Then the vector $x_1-x_2$ points inwards $B$ at $x_2$. If $B$ is open, then the vector $x_2-x_1$ points outwards  $B$ at $x_2$.
\end{itemize} 
\end{lemma}
\begin{proof} (i) See Theorem 6.2 in \cite{rockafellar}. (ii) See Theorems 6.1 in \cite{rockafellar}. (iii) Consider $x=x_2+\epsilon(x_1-x_2)=(1-\epsilon)x_2+\epsilon x_1$ with $0<\epsilon<1$. By convexity, $x$ belongs to $\cl(B)$, hence $x_1-x_2$ points inwards $B$ at $x_2\in\bd(B)$. Assume that $x_2-x_1$ also points inwards $B$ at $x_2$ and that $B$ is open. Then, for small $\epsilon$ we have $x=x_2+\epsilon(x_1-x_2) \in B$ by (ii) (which is stronger than $x\in \cl(B)$), and  $x'=x_2+\epsilon(x_2-x_1)\in \cl(B)$ by definition of pointing inwards. Again by (ii), 
$\frac{1}{2}x+\frac{1}{2}x'=x_2\in B$, contradicting that $x_2\in \bd(B)$ ($B$ is open). Hence $x_2-x_1$ points outwards $B$ at $x_2$.  \end{proof}

\subsection{Functions} 
Given an open set $B\subseteq \R^n$,  we let $\mC^k(B,\R^m)$ denote the set of $\mC^k$-functions from $B$ to $\R^m$. If $B$ is open and bounded, then we let $\mC^k(\cl(B),\R^m)$ denote the subset of $\mC^k(B,\R^m)$-functions $f$ whose $j$-th derivative $d^jf$, $j=0,\dots,k$, extends continuously to the boundary of $B$. Equivalently,  $d^jf$ is uniformly continuous in $B$ for $j=0,\dots,k$, since $\cl(B)$ is compact.

For $f\in\mC^1(B,\R^n)$ and $x^*\in B$, we let
$J_f(x^*)\in \R^{n\times n}$ be the \emph{Jacobian} of $f$ evaluated at $x^*$, that is, $J_f(x^*)$ is the matrix with  $(i,j)$-entry  $\partial f_i (x^*)/ \partial x_j$. We say that $y\in \R^n$ is a \emph{regular value} for $f$ if $J_f(x)$ is non-singular for all $x\in B$ such that
$y=f(x)$. If this is not the case, then we say that $y$ is a \emph{critical value} for $f$.  

If $B\subseteq \R^n$ is open and bounded, $f\in \mC^1(\cl(B),\R^n)$ and $y$ is a regular value for $f$ such that $y\notin f(\bd(B))$, then the set 
 $$\{x\in B | f(x)=y\}.$$ is finite \cite[Lemma 1.4]{vandervorst}.

\section{Brouwer degree and a theorem}

\subsection{Brouwer degree}
We first recall basic facts about the {\em Brouwer degree}. We refer to Section 14.2 in 
\cite{Teschl} for background and fundamental properties of the Brouwer degree.   
See also the lecture notes by Vandervorst \cite{vandervorst}.

 In this section we let $B\subseteq\R^n$ be an open bounded set.  We use the symbol $\deg(f,B,y)$ to denote  the Brouwer degree (which is an integer number) of a function $f\in\mathcal{C}^0(\cl(B),\R^n)$ with respect to $(B,y)$, $y\in \R^n\setminus f(\bd(B))$. 

A main property of the Brouwer degree is that if $y\notin f(\cl(B))$, then
$\deg(f,B,y)=0$ (but not vice versa) and if $\deg(f,B,y)\neq 0$, then
there exists at least one $x\in B$ such that $y=f(x)$. In particular, the Brouwer degree can be used to study the number of solutions to the equation
\begin{displaymath}
  f(x) = y,\qquad  x \in B,
\end{displaymath}
provided $y\notin f(\bd(B))$ and $f\in \mC^0(\cl(B),\R^n)$. 

The Brouwer degree 
$\deg(f,B,y)$ is characterized by the following properties:
\begin{enumerate}[{(A}1{)}]
\item {\bf Normalization.} Let $\id_B$ denote the identity map from $B$ to itself.  If $y\in B$, then 
    \begin{displaymath}
    \deg(\id_B,B,y)=1.
  \end{displaymath}
  \item {\bf Additivity.} If $B_1$ and $B_2$ are disjoint open subsets
  of $B$ such that $y\notin f\big(\cl(B) \setminus (B_1 \cup B_2)\big)$, then 
  \begin{displaymath}
    \deg(f,B,y) = \deg(f,B_1,y) + \deg(f,B_2,y).
  \end{displaymath}
\item {\bf Homotopy invariance.} Let $f,g\colon \cl(B)\rightarrow \R^n$ be two homotopy equivalent $\mC^0$-functions via a continuous homotopy $H\colon \cl(B)\times [0,1]\rightarrow \R^n$ such that $H(x,0)=f(x)$ and $H(x,1)=g(x)$. If  $y\notin H(\bd(B) \times [0,1])$, then
  \begin{displaymath}
    \deg(f,B,y)=\deg(g,B,y).
  \end{displaymath}
\item {\bf Translation invariance.} $\deg(f,B, y) = \deg(f-y,B,0)$.
\end{enumerate}

To prove our main result (Theorem~\ref{theo:main} below) we need the
following well-known property of the Brouwer degree, see e.g.~\cite[Theorem 14.4]{Teschl}:
 \begin{theorem}
  \label{theo:degree_sum}
  Let $f\in \mC^1(\cl(B),\R^n)$ with $B\subseteq \R^n$ an open  bounded set.
  If $y$ is a regular value for $f$ and $y\notin f(\bd(B))$, then
  \begin{equation}
    \label{eq:sum_formula}
    \deg(f,B,y) = \sum_{\{x\in B \st f(x)=y\} }
    \sign(\det(J_f(x))),
  \end{equation}
 where the sum over an empty set is defined to be zero.
\end{theorem}

\begin{corollary}\label{cor:odd}
Under the assumptions of Theorem \ref{theo:degree_sum}, assume $\deg(f,B,y) =\pm 1$.
Then the equation {$f(x)=y$}
has at least one solution  $x \in B$ and the number of solutions  in  $B$ is  odd.
\end{corollary}

\subsection{The Brouwer degree for a special class of functions}

In this section we  use Theorem~\ref{theo:degree_sum}
and the homotopy invariance of the Brouwer degree (A3) to compute the Brouwer degree
of certain functions. 
  Specifically, we are concerned with  $\mC^1$-functions 
 \begin{equation}\label{eq:f}
f\colon \R^n_{\ge 0}\to\R^n,
 \end{equation} 
and matrices $W\in \R^{d\times n}$  of maximal rank $d$. A priori there is no relationship between $f$ and  $W$.

Assume that $W$ is row reduced and let $i_1,\dots,i_d$ be the indices of the first non-zero coordinate of each row, $i_1<\ldots <i_d$.
Let $c\in  \R^d$   
and define the $\mC^1$-function 
$$ {\varphi_c} \colon \R^n_{\geq 0} \rightarrow \R^n$$
by
\begin{equation}\label{eq:phic}
\varphi_c(x)_i =  \begin{cases}  f_i(x) & i\notin \{i_1,\dots,i_d\} \\
(Wx-c)_i  & i\in \{i_1,\dots,i_d\}.
\end{cases}
\end{equation}
We say that $\varphi_c$ is constructed from $f$ and $W$.  The dependence of  $\varphi_c$  on $f$ and $W$ is omitted in the notation. We will make use of this construction with different choices of $f$ and $W$.

 Define the positive closed and open level sets of $W$ by
\begin{align}
\mP_c &= \{ x\in \R^n_{\geq 0} \st Wx=c\}\label{eq:stoich},  &
\mP_c^+ &= \{ x\in \R^n_{>0} \st Wx=c\}.
\end{align}
It follows readily that the two set are convex. By reordering the columns of $W$, the vector $(x_1,\dots,x_n)$ and the coordinates of $f$ simultaneously, if necessary, we can assume without loss of generality that $\{i_1,\dots,i_d\} = \{1,\dots,d\}$. In this case, $W$ has the block form
\begin{equation}\label{eq:Wblock} 
W = ( I_d \quad \widehat{W} ),
\end{equation}
where $\widehat{W} \in \R^{d\times s}$, $s:=n-d$, and $I_d$ is the identity matrix of size $d$. The last $s$ coordinates of the function $\varphi_c$ come from $f$.

Assuming this reordering,
 let $\pi\colon \R^n \rightarrow \R^s$ be the projection onto the last $s$ coordinates. Using \eqref{eq:Wblock}, it follows that 
\begin{equation}\label{eq:Wc}
Wx=c \quad\text{if and only if }\quad (x_{1},\dots,x_d)^T  = c- \widehat{W}(\pi(x)). 
\end{equation}
In particular, for $x,y\in \R^n$ fulfilling $Wx=Wy$, we have that
\begin{equation}\label{eq:Wc2}
x=y \quad\text{if and only if }\quad  \pi(x)=\pi(y). 
\end{equation}
If $Wf(x)=0$, then it follows from \eqref{eq:Wc2} that  $f(x)=0$ if and only if $\pi(f(x))=0$.
 
 Our first result concerns the Brouwer degree of $\varphi_c$. 
The proof of the  theorem is adapted from the proof of
Lemma~2 in \cite{mono-018} in order to account for the reduction in dimension introduced by $\mP_c$.

\begin{theorem}
  \label{theo:deg_g_forward}
 Let  $f\colon \R^n_{\ge 0}\to\R^m$  be a $\mC^1$-function and $W\in \R^{d\times n}$  a matrix of rank $d$.  Let $s:=n-d$,  $c\in \R^d$, $\mP_c$ as in \eqref{eq:stoich}  and $\varphi_c$ as in \eqref{eq:phic}.
 Let $B_c$ be an open, bounded  and convex subset of $\R^n_{> 0}$ such that 
 \begin{itemize}
\item[(i)] $B_c\cap \mP_c\neq \emptyset$.
\item[(ii)] $f(x)\not= 0$ and $Wf(x)=0$  for  $x\in\bd(B_c)\cap \mP_c$.
\item[(iii)] for every $x\in \bd(B_c) \cap \mP_c$, the  vector $f(x)$ points inwards 
$B_c$ 
at $x$.
  \end{itemize}
 Then
  \begin{displaymath}
    \deg(\varphi_c,B_c,0)=(-1)^s.
  \end{displaymath}
\end{theorem}

\begin{proof}
Without loss of generality, we might assume that $W$ has the block form in \eqref{eq:Wblock}.
Choose an arbitrary point $\bar{x}\in B_c\cap \mP_c$, which exists by assumption (i), and consider the continuous function $G\colon
  \cl(B_c) \rightarrow \R^n$ defined by
  $$G(x)=(Wx - c, \pi(\bar{x}-x)) \in \R^d\times \R^s \cong \R^n,$$
  where  $\pi$ is  the projection map onto the last $s$ coordinates of $\R^n$.
By \eqref{eq:Wblock}, the Jacobian of $G$ has the block form
    $$ J_G(x) = \begin{pmatrix}  I_d &  \widehat{W}  \\ 0 & -I_s  \end{pmatrix}.  $$
 Therefore, $\det(J_G(x))= (-1)^s$ for all $x$. In particular, $0$ is a regular value for $G$. 
 Furthermore, if $G(x)=0$, then $x\in \mP_c$ since $Wx=c$ and $\pi(\bar{x})=\pi(x)$. Using \eqref{eq:Wc2}, we conclude that $\bar{x}=x$. Since $\bar{x} \notin \bd(B_c)$, it follows that $G$ does not vanish on the boundary. 
   We apply Theorem~\ref{theo:degree_sum} to compute the degree of $G$ for $0$:
  $$\deg(G,B_c,0)= \sign(\det(J_G(\bar{x})))= (-1)^s.$$

Consider now the  following homotopy between the functions $\varphi_c$ and $G$:
  \begin{eqnarray*}
    H \colon \cl(B_c) \times [0,1] & \rightarrow   & \R^n \\ 
    (x,t) & \mapsto &  t  \varphi_c(x) + (1-t) G(x). 
  \end{eqnarray*}
  Clearly, $H$ is continuous.  To apply (A3) to find the degree of $\varphi_c$, we need to show that
  $H(\bd(B_c) \times [0,1])\neq 0$ for all $t\in [0,1]$. Since
  $$ H(x,t)= (Wx-c,t\pi(f(x))+(1-t) \pi(\bar{x}-x)),$$
  $H(x,t)=0$ implies that $Wx=c$ and hence $x\in \mP_c$. Thus, we need to show that 
\begin{equation}\label{eq:tpi}
t\pi(f(x))+(1-t) \pi(\bar{x}-x)\neq 0\quad \textrm{for all }\quad x\in \bd(B_c)\cap \mP_c.
\end{equation}
For $t=1$, \eqref{eq:tpi} follows from \eqref{eq:Wc2} using that $f(x)\neq 0$ and $Wf(x)=0$ for $x \in \bd(B_c)\cap \mP_c$ by assumption (ii). 
For $t=0$, we have already shown that $G$ does not vanish on the boundary of $B_c$.

Assume now that for $t\in (0,1)$, \eqref{eq:tpi} does not hold. That is, there exists $x' \in \bd(B_c)\cap \mP_c$ such that    \begin{equation*}
 \pi( f(x')) = \frac{t-1}{t} \pi(\bar{x}-x').
  \end{equation*}
 Since $x' \in \bd(B_c)\cap \mP_c$, we have that  $Wf(x')=0$ and $W(\bar{x}-x')=0$. We conclude using  \eqref{eq:Wc2} that 
  \begin{equation}\label{eq:H} 
  f(x')=  \frac{t-1}{t} ( \bar{x} - x'). 
    \end{equation}
 Since $\frac{t-1}{t}<0$,
  $\bar{x}\in B_c$ and $x' \in \bd(B_c)$, it follows from Lemma~\ref{lem:inwards}(iii) that $f(x')$ points outwards  $B_c$ at  $x'$,  contradicting assumption (iii).

 Therefore,  $H(x,t)\neq 0$ for all  $x\in \bd(B_c)$ and $t\in [0,1]$.  As a consequence, the homotopy invariance of the Brouwer degree (A3), gives the desired result
  \begin{displaymath}
    \deg (\varphi_c,B_c,0)=\deg(G,B_c,0)=(-1)^s.
\vspace{-0.5cm}
  \end{displaymath}
\end{proof}

\section{Chemical reaction networks}

\subsection{Setting}
Consider a chemical reaction network with species set   $\{X_1,\dots,X_n\}$ and reactions: 
\begin{equation}\label{eq:reaction}
R_j\colon \sum_{i=1}^n \alpha_{ij} X_i \rightarrow \sum_{i=1}^n\beta_{ij} X_i, \qquad j=1,\dots,\ell,
\end{equation} 
where $\alpha_{ij},\beta_{ij}$ are non-negative integers.  {The left hand side is called the reactant complex and the right hand side the product complex.}

The ODE system associated with the chemical reaction network $G$  (as described in the main text) takes the form
\begin{equation}\label{eq:ode}
\dot{x}=f(x)=Nv(x),\quad f\colon \R^n_{\ge 0}\to\R^n,
\end{equation}
where $N\in \R^{n\times \ell}$ is the stoichiometric matrix and $v(x)$ is the vector of rate functions, which are assumed to be $\mC^1$-functions (e.g. mass-action monomials).

We say that the network has rank $s$ if  the rank of the stoichiometric matrix is $s$ {and define $d=n-s$ to be the corank of the network}. 
The stoichiometric compatibility classes are the convex sets $\mP_c$ defined in \eqref{eq:stoich}, where $W$ is a matrix such that the rows form a basis of $\im(N)^\perp$.  By construction,  a trajectory of \eqref{eq:ode} is confined to the stoichiometric compatibility class where its initial condition belongs to.
The positive stoichiometric compatibility classes $\mP_c^+$ are defined accordingly.

The positive solutions to the system of equations $\varphi_c(x)=0$ with $\varphi_c$ as in \eqref{eq:phic}, are precisely the positive equilibria of the network in the stoichiometric compatibility class $\mP_c$.

Let $\phi(x,t)$ denote the flow of the ODE system and let the semiflow of the ODE system be the restriction of the flow to  $t\ge 0$.
It is assumed that the choice of rate functions $v(x)$ is such that 
\begin{equation}\label{eq:assumption}
v_j(x)=0 \quad  \textrm{if} \quad x_{i}=0 \quad \textrm{ for some }i \textrm{ with }\alpha_{ij}>0.
\end{equation}
In particular,   mass-action kinetics fulfil this condition.
Under this assumption, the non-negative and the positive  orthants, $\R^n_{\geq 0}$ and $\R^n_{> 0}$, are forward invariant under the ODE system \eqref{eq:ode}, cf.\, \cite[Section 16]{Amann}. That is, if $x_0\in \R^n_{\geq 0}$ (resp. $\R^n_{> 0}$), then the solution to the ODE system \eqref{eq:ode} with initial condition $x_0$ is confined to $\R^n_{\geq 0}$ (resp. $\R^n_{> 0}$):
\begin{equation}
  \label{eq:forward_invariance}
  x_0 \in  \R^n_{\geq 0}  \Rightarrow \phi(x_0,t)\in \R^n_{\geq 0} ,\qquad \forall t \geq  0 \quad \textrm{ in the interval of definition}.
\end{equation}
Forward invariance implies that the semiflow $\phi(x,t)$ 
maps $\R^n_{\geq 0} $ to itself for any fixed $t\geq 0$ for which the solution is defined.

Since the dynamics is confined to the stoichiometric compatibility classes, this implies that for a point $x'$ at the relative boundary of $ \mP_c$, the vector $f(x')$ points inwards $\mP_c$. Further,  both $\mP_c$ and $\mP_c^+$ are also forward invariant sets. Recall that these are convex sets.

\subsection{Conservative and dissipative networks}
 
\begin{definition}A chemical reaction network is conservative if $\im(N)^\perp$ contains a positive vector, that is, if $\R^n_{>0} \cap \im(N)^\perp\neq \emptyset$.
\end{definition}
A network is \emph{conservative} if and only if  
 the stoichiometric compatibility classes $\mP_c$ are compact  subsets of $\R^n_{\geq 0}$  \cite{benisrael}.

 \begin{definition}\label{def:Kc}
 Consider a network with associated ODE system $\dot{x}=Nv(x)$. 
The semiflow of the network  is  \emph{dissipative} if  for all $c\in  \R^d$ such that $\mP_c^+\neq \emptyset$, there exists a compact set $K_c\subseteq \mP_c$  such that $\phi(x,t)\in K_c$ for all $x\in \mP_c$ and $t\ge t(x)$, for some $t(x)\geq 0$. That is, all trajectories in $\mP_c$ enter $K_c$ at some point.
  \end{definition}

The set $K_c$ is called \emph{attracting} (and sometimes \emph{absorbing})  \cite{hofbauer}.
  Equivalently, the semiflow of a network is dissipative if all trajectories are \emph{eventually uniformly bounded}, that is, there exists a constant $k>0$ such that 
 $$ \limsup_{t\rightarrow +\infty}  x_i(t) \leq k$$
 for all $i=1,\dots,n$ and all initial conditions in $\mP_c$, provided that  $\mP_c^+\neq \emptyset$ for $c$ arbitrary.

 If the semiflow of the network is  dissipative, then the unique solution to the ODE system \eqref{eq:ode} for a given initial condition is defined for all $t\geq 0$, in which case the semiflow is said to be forward complete.

\begin{lemma}\label{lem:dis}
Consider a network with a dissipative semiflow and let $c\in\R^d$ such that $\mP_c^+\neq\emptyset$. 
   Then the following holds:
\begin{enumerate}[(i)]
\item  {An attracting set $K_c$ can be chosen such that $K_c\cap \R^n_{>0}$ is non-empty, that is,  $K_c\not\subseteq \bd(\R_{\ge 0}^n)$.}
\item All $\omega$-limit points in $\mP_c$ of the system are contained in $K_c$.  In particular, all positive equilibria in  $\mP_c$ belong to $K_c$. 
\item There exists an attracting set $K_c'$ such that $K_c\subseteq K_c'$,  $K_c'$ is forward invariant and all $\omega$-limit points outside the boundary of $\R^n_{\geq 0}$ are in the interior of $K_c'$ (relatively to $\mP_c$).
\end{enumerate}
\end{lemma}
\begin{proof}
(i) {Consider an attracting set  $K''_c\subseteq \mP_c$ and assume that $K''_c\subseteq \bd(\R_{\ge 0}^n)$. Since $\mP_c^+\neq \emptyset$, there exists a compact set $K_c\subseteq \mP_c$ that includes $K''_c$ and such that $K_c\cap \R^n_{>0}$ is non-empty. This set is also an attracting set. }

(ii) If it  were not the case, there would exist an $\omega$-limit point $x'\in \mP_c\setminus K_c$, a trajectory $\phi(x,t)$ and a sequence of time points $t_i$ such that $\lim_{i\rightarrow \infty} t_i=\infty$     and $\lim_{i\rightarrow \infty} \phi(x,t_i) = x'$.   As $K_c$ is closed, there exists an open ball $B_\epsilon(x')$ in $\R^n$ such that $B_\epsilon(x')\cap K_c=\emptyset$ and $\phi(x,t)\in B_\epsilon(x')$ for arbitrary many time points. However, this contradicts that $K_c$ is an attracting set.

(iii) 
By (ii) and choosing $K_c$ potentially larger, all $\omega$-limit points outside the boundary of $\R^n_{\geq 0}$ are in the interior of $K_c$ (relatively to $\mP_c$).
The existence of an attracting set $K_c'$ that includes $K_c$ and is forward-invariant is proven in the first part of the proof of Lemma 2 in \cite{hofbauer}. In the notation of  \cite{hofbauer}, $K_c'=K^+$. 
\end{proof}

 The semiflow of a conservative network is   dissipative.  Indeed, it is sufficient to take $K_c =\mP_c$, since $\mP_c$ is compact.  If the network is not conservative,  then the  semiflow associated with the network might still be dissipative (see Example ``Gene transcription network'' in the main text).  However, in general, it is not straightforward to show that. In some cases it is possible to prove dissipativity by constructing 
  a suitable Lyapunov function. It is the idea underlying the proof of the next proposition.

  \begin{proposition}\label{prop:diss_flow}
Assume that for all $c\in\R^d$ such that $\mP_c^+\neq \emptyset$, there exists a vector $\omega \in \R^{n}_{>0}$
and a real number $r>0$ such that $\omega \cdot f(x)< 0$ for all $x\in \mP_c$ with $||x||\geq r$, where $||\cdot ||$ is any norm. (Note that   $\omega$ and $r$ might depend on $c$.)
Then the semiflow of the network is dissipative.
 \end{proposition}
 
\begin{proof}
Let $c\in\R^d$ with $\mP_c^+\neq \emptyset$ and let $\omega$ be as given in the statement. Define 
       $$ V(x) =\sum_{i=1}^n \omega_i x_i \quad \text{for}\quad x\in \R^n_{\geq 0}.$$
The function $V(x)$ satisfies  $V(0)=0$ and $V(x)> 0$ for all  $x\in \R^n_{\geq 0}$, different from $0$. 
Further, for $||x||\geq r$ and $x\in \mP_c$, $\dot{V}(x) = \nabla V \cdot f(x) = \omega \cdot f(x)  <0$  by assumption.
Thus, $V(x) $  is a strict Lyapunov function and $V(\phi(x,t))$ is strictly decreasing along trajectories $\phi(x,t)$ in $\mP_c$ as long as $||\phi(x,t)||\geq r$.
Choose  $R>0$ such that
$$\{x\in \R^n_{\geq 0}\st ||x||\leq r\} \subseteq  \{ x\in \R^n_{\geq 0} \st V(x)\le R  \},$$
and define $K_c=\{ x\in \R^n_{\geq 0} \st V(x)\le R  \}\cap \mP_c$.
The set $K_c$ is compact by construction  and forward invariant since $\dot{V}(x)<0$ for all $||x||\ge r$. Further, all trajectories eventually enter $K_c$ within finite time, that is, $K_c$ is attracting. 
Indeed, if this were  not the case, then there would exist $x\in \mP_c$, $x\notin K_c$ (hence $||x||>r$) such that $V(\phi(x,t))$ is decreasing  for all $t\geq 0$ in the interval of definition and bounded {below}
by $R$. As a consequence, the trajectory is defined for all $t\geq 0$ and  ($*$) $\lim_{t\to\infty}V(\phi(x,t))=R'\ge R$ for some $R'$.  Hence   $\phi(x,t)$ is in $B_\epsilon:=\{x\mid V(x)\le R'+\epsilon\}$  for large $t$ (and any $\epsilon>0$). Since $B_\epsilon$ is compact it follows that the semiflow $\phi(x,t)$ has at least one $\omega$-limit point in $B_\epsilon$. By virtue of ($*$), all $\omega$-limit points $x'$ of $\phi(x,t)$ must fulfil $V(x')=R'$. Further, the set of $\omega$-limit points is forward invariant and since $V(x')=R'$ it must be that $\dot{V}'(x')=0$. This  contradicts the assumption that $\dot{V}'(x) <0$ for all $x$ with $||x||\ge r$.
We conclude that there exists $t(x)\ge 0$ such that $\phi(x,t)\in K_c$ for all $x\in\mP_c$ and $t\ge t(x)$. Hence, the semiflow is dissipative.
\end{proof}

 \medskip\noindent
\subsection{Degree for dissipative semiflows}
The main results to establish a characterization of regions of multistationarity (Theorem~\ref{theo:main}) are  Theorem \ref{theo:degree_sum} and the theorem below.
  The proof of the theorem relies on Theorem~\ref{theo:deg_g_forward} and ideas developed in \cite{hofbauer}.
  
      \begin{theorem}\label{thm:dissipative}
   Consider a network of rank $s$ with an associated ODE system $\dot{x}=f(x)$ where $f(x)=Nv(x)$ as in \eqref{eq:ode}. Assume \eqref{eq:assumption}  holds on the rate functions and let $W\in \R^{d\times n}$, $d=n-s$, be a row reduced matrix such that the rows of $W$ form a basis of $\im(N)^\perp$. Let $c\in \R^d$ such that $\mP_c^+\neq \emptyset$. Further, assume that:
   \begin{itemize}
   \item  The semiflow of the network is dissipative, and that
   \item  $f(x)\neq 0$ for all $x\in \bd(\R^n_{\geq 0})\cap \mP_c$.  That is, there are no boundary equilibria in $\mP_c$.
\end{itemize}
 Then there exists an open bounded and convex set $B_c\subseteq \R^n_{>0}$ that contains all positive  equilibria of the network in the stoichiometric compatibility class $\mP_c$, and such that
   $$ \deg(\varphi_c,B_c,0)=(-1)^s,$$
where $\varphi_c$ is defined in \eqref{eq:phic} from $f$ and $W$.
    \end{theorem}
   \begin{proof}
The idea of the proof is to construct a function $g$ defined on $\R^n_{\ge 0}$ and a set $B_c\subseteq \R^n_{> 0}$
such that the conditions of Theorem \ref{theo:deg_g_forward} are fulfilled for $g, W$ and $B_c$. 
If we let  $\varphi_c^g$ be the function $\varphi_c$ in \eqref{eq:phic}  constructed from the function $g$ and $W$,
this will imply that $ \deg(\varphi_c^g,B_c,0)=(-1)^s.$
 Subsequently, we will use homotopy invariance  to conclude that also $ \deg(\varphi_c,B_c,0)=(-1)^s$.

The function $g$ will be  defined as
   $$g(x)=\frac{1}{T}(\phi(x,T)-x)+T\rho(x),$$ 
  where $\phi(x,t)$ is the semiflow of $\dot{x} =f(x)$, $K_c$ is a suitably chosen attracting set, $T$ is the maximum entrance time into $K_c$ from a specific set, and $\rho(x)$ is an auxiliary function with certain  useful properties (see below).
   
The proof is divided into four steps. In step (A) we define the set $B_c$, choose $K_c$ and find basic properties of $B_c$ and $K_c$.  In step (B), we construct the function $\rho$. In step (C), we properly define $g$ and show that $g$, $B_c$ and $W$ have the required properties to apply Theorem~\ref{theo:deg_g_forward}.  In step (D) we show that $\varphi_c^g$ and $\varphi_c$ are homotopy equivalent and conclude the proof of the theorem using the homotopy invariance of the Brouwer degree.
     
     \medskip 
{\bf (A)} 
Let $K_c\subseteq \mP_c$ be  as in Definition \ref{def:Kc}, that is, a compact attracting set of all trajectories with initial condition in $\mP_c$. According to Lemma~\ref{lem:dis}, 
{$K_c$ can be chosen such that  $K_c$ is forward invariant, $K_c\cap \R^n_{>0}\neq \emptyset$, and all $\omega$-limit points in   $\mP_c^+$ are interior points of $K_c$ (relatively to $\mP_c$).}

\begin{figure}
\begin{center}
\begin{tikzpicture}[scale=0.6]
\draw[dashed,fill=gray!2!white] (0,6) .. controls (-1.3,5.5) and (-1.7,4.5) .. (-2,3.9) .. controls (-2.6,2) and (-2,0) .. (-0.5,-0.5) .. controls (2,-1.2) and (5,-1.5) .. (5.5,0);
\draw[fill=yellow!10!white,draw=orange!80!red,dashed] (0,0) -- (5.5,0) .. controls (6,2) and (4,7) ..  (0,6) -- (0,0);
\draw[->] (5.5,0) -- (7,0);
\draw[->] (0,6) -- (0,7);
\draw[-,thick] (4.5,1) -- (6,0);
\draw[-,blue,line width=1.5pt] (0,4) -- (4.5,1);
\draw[blue,fill] (0,4) circle (1.5pt);
\draw[blue,fill] (4.5,1) circle (1.5pt);
\draw[fill=green,very nearly transparent,line width=2pt,dashed] (2.25,2.5) circle [x radius=100pt, y radius=40pt,rotate=-34];
\draw (-1,5.8) node {$B$};
\draw[orange] (2,5.3) node {$B_c$};
\draw[blue] (2.2,3) node {$K_c$};
\draw (-0.7,3) node {$U_1$};
\draw[] (6.1,0.5) node {$\mP_c$};
\end{tikzpicture}
\caption{\footnotesize Step (A). The set $\mP_c$ is the straight line connecting the two axis. The compact attracting set $K_c$ is depicted in blue. The set $B\subseteq \R^n$ is an open set containing $K_c$ and $B_c=B\cap \R^n_{>0}$ is the restriction of $B$ to the positive orthant (shown in orange), such that $B_c$ is open. Hence $K_c$ is contained in $B_c$, except for  points on the boundary  $K_c \cap \bd(\R^n_{\geq 0} )$, hence also $B_c\cap \mP_c\not=\emptyset$. Step (B). The open set $U_1\subseteq\R^n$ (in green) is chosen such that $K_c\subseteq U_1\subseteq B$. In the $\mC^1$-partition of unit, the support of $\psi_1$ is in $U_1$ and that of $\psi_2$ is in $\R^n\setminus K_c$. }\label{fig:sets}
\end{center}
\end{figure}
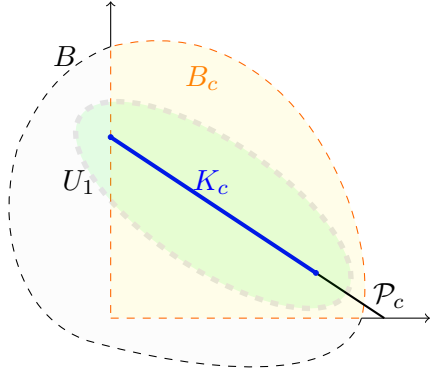

 Let $B\subseteq \R^n$ be an open, bounded and convex set   containing $K_c$, that is, $K_c\subseteq B$.       Let $B_c=\R^n_{> 0} \cap B$. Then $B_c$ is also open, bounded and convex.  
 Since $K_c\subseteq \R^n_{\geq 0}\cap B$, then $B_c$ contains all points in $K_c$ except  those on the boundary $K_c \cap \bd(\R^n_{\geq 0} )$. Further, 
 \begin{equation*}\label{eq:E1}
K_c\subseteq \cl(B_c)\subseteq \R^n_{\geq 0}, \quad \text{and}\quad   K_c\cap \bd(B_c) = K_c \cap \bd(\R^n_{\geq 0}),
 \end{equation*}
see Figure \ref{fig:sets}.
Since $ \emptyset \neq K_c\cap \R^n_{>0} =  B \cap K_c \cap\R^n_{> 0} =B_c \cap K_c \subseteq B_c \cap \mP_c$, then
 \begin{equation*}\label{eq:E2}
 B_c\cap \mP_c\neq \emptyset.
 \end{equation*} 
 Since $f(x)\neq 0$ for all $x\in \bd(\R^n_{\geq 0})\cap \mP_c$ by assumption and $K_c$ contains all zeros of $f$ in $\mP_c$,  then  $B_c$ contains all zeros of $f$ in $\mP_c$, that is
\begin{equation}\label{eq:E3}
\{x\in \mP_c\mid  f(x)=0\}\subseteq B_c.
 \end{equation}

  \medskip
{\bf (B)} The function $\rho\colon \R_{\ge 0}^n\to \R^n$ in the definition of $g$  is defined such that it has the following properties:
\begin{enumerate}[(i)] 
 \item $\rho(x)$  points inwards $B_c$
  for all $x\in \bd(B_c)\cap \mP_c$. 
  \item $\rho(x)=0$ for $x\in \R^n_{>0}\cap\bd(B_c)\cap \mP_c$.
 \item $\rho(x)\neq 0$  for $x\in K_c\cap \bd(B_c)$. 
 \item $W\rho(x)=0$ for all $x\in \bd(B_c)\cap \mP_c$. 
  \end{enumerate}
We first construct two other functions $\widetilde\rho$ and $\psi_1$, and subsequently define $\rho\colon \R_{\ge 0}^n\to \R^n$ as the product $\rho=\widetilde\rho\,\psi_1$.
Let $\widetilde{x}\in K_c\cap \R^n_{>0}$ and  define $\widetilde{\rho}\colon \R^n\rightarrow \R^n$ as  $\widetilde{\rho}(x):= \widetilde{x}-x$.
Let $U_1\subsetneq B$ be an open set containing $K_c$ (which exists since $B$ is open), see  Figure \ref{fig:sets}. Consider the open cover of $\R^n$ given by $U_1$ and $U_2=\R^n\setminus K_c$, such that $U_1\cap U_2\not=\emptyset$ and $U_1\cup U_2=\R^n$. 
Choose a $\mC^1$-partition of unit $\psi_1,\psi_2\colon \R^n\rightarrow [0,1]$ associated with this open cover.  This implies in particular that the support of $\psi_i$ is included in $U_i$ and $\psi_1(x)+\psi_2(x)=1$ for all $x$.

 Define  $\rho\colon \R_{\ge 0}^n\to \R^n$ by $\rho(x)=\psi_1(x)\widetilde{\rho}(x)$, $x\in \R^n_{\geq 0}$ (note the restriction to $\R^n_{\geq 0}$). This function fulfils  properties (i)-(iv) above. 
Property (i): Follows by definition of $\rho(x)=\psi_1(x)(\widetilde x-x)$, $\psi_1(x)\ge 0$ and Lemma~\ref{lem:inwards}(iii), using that $\widetilde{x}\in B_c$ and $x\in \bd(B_c)$.
 Property (ii):  Since the support of $\psi_1$ is contained in $U_1$,  $\psi_1(x)=0$ for all $x\notin U_1$, in particular for all $x\in \R^n_{>0}\cap\bd(B_c)\cap \mP_c$, since $\R^n_{>0}\cap\bd(B_c)\subseteq \bd(B)$ and $\bd(B)\cap U_1=\emptyset$. Property (iii): Similarly,  $\psi_1(x)=1$ (since $\psi_2(x)=0$) for all $x\notin  U_2=\R^n\setminus K_c$, that is, for all $x\in K_c$; hence $\rho(x)\not=0$ for $x\in K_c\cap \bd(B_c)$ since $\widetilde x\not\in\bd(B_c)$. Property (iv):  $W\rho(x)=\psi_1(x)W(\widetilde x-x)=0$ as $x,\widetilde{x}\in \mP_c$.

\medskip
 {\bf (C)}  
Let $T$ be defined as the maximum of the entry times to $K_c$ from any $x\in \cl(B_c)\cap \mP_c$.  
 The time $T$ is finite because $\cl(B_c)\cap \mP_c$ is compact and the semiflow is dissipative with respect to $K_c$.  
Note that once a trajectory is in $K_c$, it stays there since $K_c$ is forward invariant
 Redefine $T$ to be any positive number if $T=0$. 

We define 
$$g\colon \R^n_{\geq 0}\to \R^n, \quad g(x):= \frac{1}{T}(\phi(x,T) -x) + T\rho(x),$$

 Observe that $Wg(x)=0$  for all $x\in \bd(B_c)\cap \mP_c$, using property (iv) in step (B)  and that $\phi(x,T),  x\in \mP_c$. By definition of $T$, $\phi(x,T)\in \cl(B_c)\cap\mP_c $ if $x\in \cl(B_c)\cap\mP_c$ and hence $\frac{1}{T}(\phi(x,T) -x)$   points inwards  $B_c$ 
 at $x\in \bd(B_c)\cap \mP_c$ by convexity of $\cl(B_c)$.  Also $ T\rho(x)$ points inwards  $B_c$ 
  at $x$ by property (i) in step (B). 
Hence,  $g(x)$ points inwards $B_c$ 
 at $x\in \bd(B_c)\cap \mP_c$ by convexity again.

Therefore, the function $g$, together with $B_c$ and $W$, fulfil the conditions  of Theorem~\ref{theo:deg_g_forward}.
By letting $\varphi_c^g$ be the function $\varphi_c$ in \eqref{eq:phic}  constructed from  $g$ and $W$,
we conclude that $$ \deg(\varphi_c^g,B_c,0)=(-1)^s.$$

\medskip
 {\bf (D)}
 We  define a homotopy between $\varphi_c$ and $\varphi_c^g$  on $\cl(B_c)\times [0,T]$ by
 $$H(x,t)= \left\{ \begin{array}{cl} \varphi_c(x) & \text{if }\quad t=0 \\ 
 \big(\, Wx-c,\frac{1}{t}\pi(\phi(x,t)-x)+ t \pi( \rho(x)) \, \big) & \text{if }\quad 0< t\le T. \end{array}\right.$$
The function $H(x,t)$ is continuous since $\phi(x,t)$ is differentiable and is the semiflow of $\dot{x}=f(x)$. Note that 
$H(x,0)=\varphi_c(x)$ and $H(x,T) = (Wx-c,\pi(g(x))) = \varphi_c^g(x)$. Thus $H(x,t)$ is a homotopy between $\varphi_c(x)$ and $\varphi_c^g(x)$.
We need to show that $H(x,t)$ does not vanish on the boundary $\bd(B_c)$.

 If   $H(x,0)=\varphi_c(x)=0$, then  $x\in \mP_c$ is an equilibrium of the ODE system. Hence $H(x,0)$ does not vanish on $\bd(B_c)$ since $B_c$ contains all zeros of $f$ in $\mP_c$, see \eqref{eq:E3}.  
 Now let $x'\in \bd(B_c)$ and assume that $H(x',t)=0$ for some $t\in (0,T]$. It follows that $x'\in \mP_c$, hence   
   \begin{equation}\label{eq:phirho2}
   x'\in \bd(B_c)\cap \mP_c,\quad\text{ and}\quad \pi(\phi(x',t)-x') = -t^2\pi(\rho(x')).
   \end{equation}
Using  \eqref{eq:Wc2} and property (iv) in step (B) we have that  
   \begin{equation}\label{eq:phirho}
      \phi(x',t) =x' -t^2\rho(x').
   \end{equation}
By construction of $K_c$, all fixed points and periodic orbits are contained in $K_c$. If $\rho(x')=0$, then \eqref{eq:phirho} implies 
$x'\in K_c\cap \bd(B_c)$ as $x'\in\bd(B_c)$ by assumption. However, this contradicts property (iii) in step (B). 
Hence, it must be the case that $\rho(x')\neq 0$. 

Using that $x'\in \bd(B_c)\cap\mP_c$ from \eqref{eq:phirho2} and $\rho(x')\neq 0$, we conclude that $x'\in \bd(\R^{n}_{\geq 0})$ by property (ii) in step (B), since $x'\not\in \R^n_{>0}\cap\bd(B_c)\cap \mP_c$.
It follows that  there exists $i$  such that $x'_i=0$ and we have
{$$ \phi(x',t)_i = x_i'- t^2\rho(x')_i =   x_i' -   t^2\psi_1(x') \widetilde{\rho}(x')_i=  - t^2\psi_1(x') \widetilde{x}_i <0.$$}
{Here we have used that $\psi_1(x') \neq 0$, since $\rho(x')=\psi_1(x') \widetilde{\rho}(x')\not =0$ 
and  $\psi_1(x')$ is a scalar. }
Now, by the inequality above, $\phi(x',t)$  does not belong to $\R^n_{\geq 0}$. However, this contradicts the forward invariance of $\R^n_{\geq 0}$ with respect to the flow.
Therefore, $H$ does not vanish on $\bd(B_c)\times [0,T]$.

With this in place, homotopy invariance of the Brouwer degree  implies that
$$\deg(\varphi_c,B_c,0)= \deg(H(x,0),B_c,0) = \deg(H(x,T),B_c,0) = \deg(\varphi_c^g,B_c,0)=(-1)^s.$$
  \end{proof} 
   
 \begin{remark}
 The statement and proof of the theorem focus exclusively on one stoichiometric compatibility class, that is, on a fixed value $c\in \R^d$. 
 Therefore, if a semiflow admits an attracting set in one specific stoichiometric compatibility class (and not necessarily in all), then the theorem and computation of the Brouwer degree holds for this specific class.
 \end{remark}

\section{Multistationarity in dissipative networks}\label{sec:multi}
In this section we prove the  theorem and corollaries stated in the main text, which are  consequences of Theorem~\ref{thm:dissipative} from the previous section.

Consider the Jacobian of the map $\varphi_c(x)$. Because $\varphi_c(x)$ is independent of $c$, we denote the Jacobian by $M(x)$. The $i$-th row of this matrix is given as
$$M(x)_i :=  J_{\varphi_{c}}(x)_i =  
 \begin{cases}  J_{f_i}(x) & i\notin \{i_1,\dots,i_d\} \\
W_i  & i\in \{i_1,\dots,i_d\},
\end{cases}
$$
where $W_i$ is the $i$-th row of $W$.
That is, one can think of $M(x)$ as being the matrix obtained from the Jacobian of $f(x)$, with  the $i_j$-th row, $j=1,\dots,d$, replaced by the $j$-th row of $W$.

An equilibrium $x^*$  is said to be \emph{non-degenerate} if $M(x^*)$ has rank $n$, that is, if $\det(M(x^*))\neq 0$.

\begin{theorem}  \label{theo:main} 
Assume  the reaction rate functions fulfil \eqref{eq:assumption}, let $s=\text{rank}(N)$ and let $\mP_c$ be a stoichiometric compatibility class such that  $\mP_c^+\neq \emptyset$ (where $c\in \R^d$). Further, assume that  

\smallskip
\noindent (i) The semiflow of the network is dissipative.
 
\noindent (ii) There are no boundary equilibria in  $\mP_c$.  

\smallskip
\noindent Then the following holds:
    
\smallskip
\noindent    (A') {\bf Uniqueness of equilibria.} If
    \begin{displaymath}
      \sign(\det(M(x)))=(-1)^s\quad \textrm{ for  all  equilibria} \quad  
                 x  \in V\cap \mP_c^+,     
          \end{displaymath}
  then there is exactly one  positive equilibrium in $\mP_c$.
 Further, this equilibrium is
    non-degenerate.
    
    \smallskip
    \noindent  (B') {\bf Multiple equilibria.} If
    \begin{displaymath}
        \sign(\det(M(x)))=(-1)^{s+1}\quad \text{ for some  equilibrium }  \quad  
                 x \in V\cap \mP_c^+,     
    \end{displaymath}
    then there are at least two positive equilibria in  $\mP_{c}$, at
    least one of which is non-degenerate.
    If all positive equilibria in $\mP_{c}$ are non-degenerate, then
    there are at least three and always an odd number.
\end{theorem}
\begin{proof}
The hypotheses ensure that we can apply Theorem~\ref{thm:dissipative}. Therefore choose  an open bounded convex set  $B_c\subset \R^n_{>0}$ 
  that contains all positive equilibria of the network in the stoichiometric compatibility class $\mP_c$ and such that
   $$ \deg(\varphi_c,B_c,0)=(-1)^s.$$
Let   $V_{c}$ be the set of positive equilibria in the stoichiometric compatibility class $\mP_c$. Note that
 $$V_c=\{x\in B_c \st \varphi_c(x)=0\}.$$

\medskip
(A')  
Since $ \sign(\det(M(x)))=(-1)^s\neq 0$ for all equilibria in  $\mP_c$,  $0$ is a regular value for $\varphi_c$. We can therefore apply Theorem~\ref{theo:degree_sum} and obtain
$$ (-1)^s =  \sum_{x\in V_c } \sign(\det(M(x))) = (-1)^s (\#V_{c}),$$
where $\#V_{c}$ is the cardinality of $V_c$.
We conclude that  $\#V_{c}=1$ and therefore  that there exists a unique positive equilibrium in the stoichiometric compatibility class. Furthermore, since $\sign(\det(M(x))) \not=0$ for all equilibria, the equilibrium is non-degenerate.

\medskip
(B') Let $x^*\in V_c$  be such that
    $\sign(\det(M(x^*))) =(-1)^{s+1}$. 
If $0$ is a regular value for $\varphi_{c}(\cdot)$,
then the equality
$$ (-1)^s =   \sum_{x\in V_c } \sign(\det(M(x)))  =   (-1)^{s+1}  +   \sum_{x\in V_c, \ x\neq x^* } \sign(\det(M(x)))   $$
implies that there must exist at least two other points $x',x''\in V_{c}$, such that $$ \sign(\det(M(x')))  =  \sign(\det(M(x''))  =   (-1)^{s},$$  that is, there are at least three positive equilibria in $\mP_{c}$, all of which are non-degenerate.
 In this case by Corollary~\ref{cor:odd}, there is an odd number of  equilibria and they are all non-degenerate.

Assume now that  $0$ is not a regular value for $\varphi_{c}$. Then there must exist another positive equilibrium $x'$ in $\mP_{c}$ for which the Jacobian of $\varphi_{c}(x')$ is singular.  This implies that there are at least two positive equilibria in $\mP_{c}$, $x^*$ and $x'$, one of which is non-degenerate.
\end{proof}

In typical applications we find an odd number of equilibria ($\geq 3$), all of which are non-degenerate.  
Observe that the hypothesis for Part (A) holds if the sign of $\det(M(x))$ is $(-1)^s$ for all $x$ in a set containing the positive equilibria. In particular, this is the case if $\det (M(x)) =(-1)^s$ for all $x\in \R^n_{>0}$.

\medskip

{We assume now that the positive solutions to the system 
$f(x)=0$ (with $f(x)$ as in \eqref{eq:ode}) admit a parameterization
\begin{eqnarray}
\Phi \colon \R^{m}_{> 0} & \rightarrow & \R^n_{> 0} \label{eq:param} \\
\hat x = (\hat x_1,\dots, \hat x_m) & \mapsto & (\Phi_{1}(\hat x),\dots,\Phi_{n}(\hat x)), \nonumber
\end{eqnarray}
for some $m<n$.  
That is, we assume that we can express $x_{1},\dots,x_n$ at equilibrium as    functions of  $\hat x_1,\dots, \hat x_m$:
$$ x_{i}= \Phi_{i}(\hat x_1,\dots, \hat x_m),\qquad i=1,\dots,n, $$
 such that 
$x_{1},\dots,x_n$ are positive if  $\hat x_1,\dots, \hat x_m$ are positive.}

For mass-action kinetics, the equation $f(x)=0$  results in $s=n-d$ polynomial equations in $n$ unknowns, which generically would lead to a $d$-dimensional parameterization and $m=d$ (if such a parameterization exists).
 
When such a parameterization exists, then positive values of $\hat x_1,\dots, \hat x_m$ determine uniquely a positive equilibrium. This equilibrium then belongs to the stoichiometric compatibility class given by 
$$   c:= W \Phi(\hat x).$$
Reciprocally, given $c$, the positive solutions to  $\varphi_c(x)=0$ are in one-to-one correspondence with the positive solutions to 
the equation $ c= W \Phi(\hat x)$.

As before, we let   $W\in \R^{d\times n}$ be a row-reduced matrix whose rows form a basis of $\im(N)^{\perp}$.
Let $i_1,\dots,i_d$ be the indices of the first non-zero coordinate of each row. Let $\pi\colon \R^n\rightarrow \R^s$ denote the projection onto the coordinates with indices different from $i_1,\dots,i_d$. We do not reorder the coordinates now to ensure that  $\{i_1,\dots,i_d\}=\{1,\dots,d\}$, because we have already chosen a convenient order of the free variables of the parameterization.

 We next consider the determinant of $M(x)$ and use  the parameterization \eqref{eq:param} to substitute the values of 
 $x_{1},\dots,x_n$ by their expressions as functions of $\hat x_1,\dots, \hat x_m$. 
We define
\begin{equation}\label{eq:akp}
 a(\widehat{x})  = \det (M(\Phi(\widehat{x}))).
 \end{equation}

\begin{corollary}  \label{cor:main} 
Assume  the reaction rate functions fulfil \eqref{eq:assumption} and let $s=\text{rank}(N)$.  Further, assume that  

\smallskip
\noindent (i) The semiflow of the network is dissipative.
 
\noindent (ii) The set of positive equilibria admits a positive parameterization as in \eqref{eq:param}.

\noindent (iii)  There are no boundary equilibria in  $\mP_c$, for all $c\in \R^d$ such that  $\mP_c^+\neq \emptyset$.

\smallskip
\noindent Then the following holds.
    
\smallskip
\noindent    (A) {\bf Uniqueness of equilibria.} If $\sign(a(\hat{x}))=(-1)^s$ for all  $\hat{x}\in \R^m_{>0}$,
     then there is exactly one  positive equilibrium in each    $\mP_c$ with $\mP_c^+\neq \emptyset$.     Further, this equilibrium is  non-degenerate.
    
    \smallskip
    \noindent  (B) {\bf Multiple equilibria.} If $\sign(a(\hat{x}))=(-1)^{s+1}$ for some $\hat{x}\in \R^m_{>0}$, 
    then there are at least two positive equilibria in  $\mP_{c}$, at
    least one of which is non-degenerate,
    where  $c:= W \Phi(\hat{x})$.
    If all positive equilibria in $\mP_{c}$ are non-degenerate, then
    there are at least three equilibria  and always an odd number.
\end{corollary}
\begin{proof}
Given $c$, note that $\Phi$ induces a bijection between the sets
$$V_c \quad \text{and}\quad \mathcal{S}_{c}:= \{ \widehat{x}\in \R^m_{>0} \st c= W \Phi(\widehat{x})\}.$$
An element of $\mathcal{S}_{c}$ corresponds to a positive equilibrium in the stoichiometric compatibility class $\mP_c$. 

\medskip
(A)  
Consider a  stoichiometric compatibility class $\mP_c$ defined by $c$ such that $\mP_c^+\neq \emptyset$. 
Let $x\in V_c$ and $\hat{x}$ such that $x=\Phi(\hat{x})$. Then
$$\det(M(x))= \det(M(\Phi(\hat{x}))) = a(\widehat{x}).  $$
By hypothesis $\sign(\det(M(x)))=\sign( a(\widehat{x}))=(-1)^s$. Since this holds for all equilibria in $V_c$, Theorem~\ref{theo:main}(A') gives that there is exactly one  positive equilibrium in $\mP_c$, which is non-degenerate.

\medskip
(B) Let $\widehat{x}$  be such that
    $\sign(a(\widehat{x}))=(-1)^{s+1}$ and let $c$ be defined as in the statement of the theorem. 
    Then $x=\Phi(\widehat{x})$ is a positive equilibrium in $V_{c}$ for which the sign of $\det(M(x))$ is $(-1)^{s+1}$.   Theorem~\ref{theo:main}(B') gives the desired conclusion.
\end{proof}

\section{Details on the steps of the procedure}\label{sec:steps}

 In this section we expand further on how to check step 3 and {7} of the algorithm. 

\subsection{On siphons and boundary equilibria}\label{sec:siphons}
{ A proof of Proposition 2 in the main text} for mass-action kinetics can be found in \cite{Shiu-siphons}, where strategies to find siphons are also detailed.
 The proof in  \cite{Shiu-siphons} is however valid for general kinetics fulfilling  assumption \eqref{eq:assumption} (see  \cite[Prop. 2]{marcondes:persistence}). Different algorithms developed in Petri Net theory can be applied to find the siphons of a reaction network.

For large networks, the task of finding the siphons can be daunting. A way to reduce the complexity of the computation is by the removal of intermediate species and catalysts  \cite{marcondes:persistence}. We explain the key aspects of this reduction method here. The method is used in the examples below.

The first reduction concerns \textbf{removal of intermediates. }   Intermediates are species in the network that do not appear interacting with any other species, are produced in at least one reaction, and consumed in at least one reaction. For example the species $ES_0$ in the reaction network
\begin{equation}\label{eq:int}
S_{0} + E  \cee{<=>}   ES_0  \cee{<=>} S_{1} + E
\end{equation}
is an intermediate.

Given a network,  we  obtain a reduced network by ``removing'' some intermediates, one at a time. This is done in the following way.  Say we want to remove an intermediate $Y$ from the network. We remove all reactions of the original network that involve $Y$ and add a reaction
$$ y\rightarrow y'\qquad \textrm{whenever} \qquad y \rightarrow Y \rightarrow y'\quad \textrm{with } \ y\neq y'$$
belongs to the original network. Here $y$ and $y'$ are the reactant complex of a reaction $ y \rightarrow Y$ and product complex of a reaction $Y \rightarrow y'$,
respectively.

To illustrate this, we consider the removal of the intermediate $ES_0$ in the network \eqref{eq:int}.
The reactions of the reduced network are obtained by considering all length 2 paths of the original network that go through $ES_0$. 
We have two such paths:
$$S_{0} + E  \cee{->}   ES_0  \cee{->} S_{1} + E\qquad\textrm{and}\qquad S_{1} + E  \cee{->}   ES_0  \cee{->} S_{0} + E .$$
By ``collapsing'' these paths we obtain the reactions
\begin{equation}\label{eq:cat}
S_{0} + E \cee{->} S_{1} + E\qquad\textrm{and}\qquad S_{1} + E \cee{->} S_{0} + E.
\end{equation}

Clearly the process could be repeated now by choosing other intermediates of the network (if any). In this way we can obtain reduced networks by removing several intermediates.

\medskip
The second reduction concerns \textbf{removal of catalysts. } 
Catalysts are species that whenever they appear in a reaction, then they appear at both sides and with the same stoichiometric coefficient. 
For example, $E$ in the reaction network \eqref{eq:cat}
is a catalyst. 
Catalysts  are actually defined in more generality in  \cite{marcondes:persistence}, but we restrict to this scenario to keep the discussion simple. Catalysts are removed from a network  by literally removing them from the reactions where they appear.
Removal of $E$ in the reaction network \eqref{eq:cat} yields the reaction  network
\begin{equation}\label{eq:mon}
S_0 \cee{<=>} S_1. 
\end{equation}

This network has one minimal siphon, namely $\{S_0,S_1\}$, and $s_0+s_1=c$ is a conservation relation. By Proposition~{2 in the main text}
 it does not admit boundary equilibria in stoichiometric compatibility classes with non-empty positive part.
The next proposition allows us to conclude that the original network in \eqref{eq:int} neither admits boundary equilibria in stoichiometric compatibility classes with non-empty positive part.
 
\begin{proposition}[Theorems 1 and 2 in \cite{marcondes:persistence}] \label{prop:red}
Let $G$ be a network and $G'$ be a network obtained after iterative removal of intermediates or catalysts from $G$ as described above. 
Each minimal siphon of $G$ contains the support of a positive conservation relation if and only if this is the case for $G'$.
\end{proposition}

In several cases, removal of intermediates and catalysts yields a so-called \emph{monomolecular network}. That is, a network whose complexes agree with some species or the complex zero. For example, the network in \eqref{eq:mon} is monomolecular.
In this case, checking the hypothesis of {Proposition~2 in the main text} is straightforward, in view of the next lemma.

\begin{lemma}[Proposition 3 in \cite{marcondes:persistence}] \label{lem:red}
Let $G$ be a monomolecular network.  Each minimal siphon of $G$ contains the support of a positive conservation relation if and only if all connected components of $G$ are strongly connected.
\end{lemma}

The network  in \eqref{eq:mon} is clearly strongly connected. Thus, we do not need to find the siphons of the network to conclude that 
each of its minimal siphons contains the support of a positive conservation relation and thereby conclude that \eqref{eq:int} does not admit boundary equilibria in stoichiometric compatibility classes with non-empty positive part.

\begin{corollary}\label{cor:boundary}
Let $G$ be a network and $G'$ be a network obtained after iterative removal of intermediates or catalysts from $G$ as described above. If $G'$ is a monomolecular network with all connected components  strongly connected, 
then $G$ has no boundary equilibria in any stoichiometric compatibility class $\mP_c$ such that $\mP_c^+\neq \emptyset$.
\end{corollary}

 {
 For example, consider the ``gene transcription network''  of the main text.  The species $X_5$ and $X_7$ are intermediates. The reaction network obtained upon their removal is:
  \begin{align*}
         X_1 & \cee{->} X_1+X_3 & X_3 & \cee{->} 0  & 
         X_2 & \cee{->} X_2+X_4 & X_4 & \cee{->} 0 & 2X_4 \cee{<=>} X_6.
 \end{align*}
For this network, $X_6$ is an intermediate and $X_1,X_2$ are catalysts. Removal of these three species yields the reaction network
 $$ X_3  \cee{<=>} 0  \cee{<=>} X_4.$$
 This is a strongly connected monomolecular network. By Corollary~\ref{cor:boundary}, there are no boundary equilibria in any $\mP_c$ as long as $\mP_c^+\neq \emptyset$. We have reached the same conclusion as in the main text without the need of finding the minimal siphons of the network. 
 }

\subsection{Newton polytope}  
We write a multivariate polynomial $f(x)\in \R[x_1,\dots,x_n]$ as a sum of monomials:
$$ f(x) = \sum_{\alpha\in \N^n} c_\alpha x^\alpha,$$
where $x^\alpha= x_1^{\alpha_1}\dots x_n^{\alpha_n}$ and $c_\alpha\in \R$, for which only a finite number  are non-zero.

The \emph{Newton polytope of $f(x)$}, denoted by $\mathcal{N}(f)$, is a closed convex set in $\R^n$, defined as the convex hull of the exponents $\alpha\in \N^n$ for which $c_\alpha\neq 0$ (See \cite[Section 2]{rockafellar} for a definition of convex hull).  The set of vertices of $\mathcal{N}(f)$ is a subset of the set of points $\alpha$ for which $c_\alpha\neq 0$.

The following is a well-known fact about the Newton polytope of a polynomial. The proof of the fact is constructive and provides an explicit way to find $\widehat{x}$ in Corollary~\ref{cor:main}(B). Thus it offers  a way to find stoichiometric compatibility classes (i.e.\,values of $c$) for which multistationarity exists.

\begin{proposition}\label{prop:newton}
Let $f(x) = \sum_{\alpha\in \N^n} c_\alpha x^\alpha$ and let $\alpha'$ be a vertex of $\mathcal{N}(f)$. 
Then there exists $x'\in \R^n_{>0}$ such that
$$ \sign(f(x'))= \sign(c_{\alpha'}).$$ 
\end{proposition}
\begin{proof}
By hypothesis $c_{\alpha'}\neq 0$. 
Since $\alpha'$ is a vertex in a bounded  convex polytope, 
there exists a separating hyperplane $\omega\cdot x = T$ that intersects the polytope only in $\alpha'$ and such that $\omega\cdot x' < T$ for any other point $x'$ of the polytope (see e.g.\,Definition 3.5 and Theorem 3.8 in \cite{joswig-theobald}).
In particular, $\omega\cdot \alpha<\omega\cdot \alpha'$ for all vertices $\alpha\neq \alpha'$. 

For {$ y = t^{\omega}= \prod_{i=1}^n t^{\omega_i}$}, we have
$$ f(y)=  \sum_{\alpha\in \N^n} c_\alpha (t^{\omega})^\alpha =
 \sum_{\alpha\in \N^n} c_\alpha t^{\omega\cdot \alpha} 
 = c_{\alpha'} t^{\omega\cdot \alpha'}  +  \sum_{\alpha\in \N^n, \alpha\neq \alpha'} c_\alpha t^{\omega\cdot \alpha}.  $$
 {Now $f(y)$ is a well defined function for $t\in \R_{>0}$, which tends to $+\infty$  for $t$ tending to infinity and $c_{\alpha'}>0$ and to $-\infty$ for $c_{\alpha'}<0$ (by assumption $c_{\alpha'}\not=0$).  Hence, by letting $t$ be large enough, the sign of $f(y)$ agrees with the sign of $c_{\alpha'}$.}
\end{proof}

\paragraph{Finding the vertices in practice. }
In the examples below, we find the vertices of the Newton polytope of the polynomial of interest as follows. 
We use Maple (version 2015). We construct first the polytope using the command \emph{PolyhedralSet} and subsequently use the command \emph{VerticesAndRays}, from the package \emph{PolyhedralSets}, to find the vertices.

\section{Details on the examples in the main text}\label{sec:mainex}

\subsection{Phosphorylation of two substrates}

In this subsection we consider the network in the first row of Table 1 in the main text. 

 We consider a system in which two substrates  can be either unphosphorylated, $A, B$ or phosphorylated $A_p,B_p$. Phosphorylation of both substrates is catalyzed by the same kinase $K$  and dephosphorylation of $A_p,B_p$ is catalyzed by the same phosphatase $F$. 
That is, the system consists of two futile cycles sharing  kinase and phosphatase.

The reactions of the system are:
\begin{align*}
A + K & \cee{<=>[\k_1][\k_2]}   AK  \cee{->[\k_3]} A_p + K  &  B + K & \cee{<=>[\k_7][\k_8]}   BK  \cee{->[\k_9]}  B_p + K \\
A_p + F & \cee{<=>[\k_4][\k_5]}   A_pF  \cee{->[\k_6]}  A + F & B_p + F  & \cee{<=>[\k_{10}][\k_{11}]}   B_pF \cee{->[\k_{12}]}  B + F.
\end{align*}
{This network is a PTM network with substrates $A,B,A_p,B_p$, enzymes $K,F$ and intermediates $AK, BK, A_pF, B_pF$.}
It was shown in \cite{feliu2011} that this network with mass-action kinetics is multistationary. Here we find the necessary and sufficient condition on the reaction rate constants for having multistationarity in some stoichiometric compatibility class.
We let 
\begin{align*}
 X_1 & =K, &   X_3 & =A, & X_5& =B, & X_7 &=AK, &   X_9& =A_pF,  \\
 X_2  & =F,  & X_4 & =A_p,  & X_6& =B_p,  & X_8 & =BK, & X_{10} & =B_pF.
\end{align*}

 The stoichiometric matrix $N$ of the network and a row reduced matrix $W$ whose rows from  a basis of $\im(N)^\perp$
are
{\small \begin{align*}
N & =\left(
\begin {array}{rrrrrrrrrrrr} -1&1&1&0&0&0&-1&1&1&0&0&0
\\ 0&0&0&-1&1&1&0&0&0&-1&1&1\\ -1&
1&0&0&0&1&0&0&0&0&0&0\\ 0&0&1&-1&1&0&0&0&0&0&0&0
\\ 0&0&0&0&0&0&-1&1&0&0&0&1\\ 0&0&0
&0&0&0&0&0&1&-1&1&0\\ 1&-1&-1&0&0&0&0&0&0&0&0&0
\\ 0&0&0&0&0&0&1&-1&-1&0&0&0\\ 0&0
&0&1&-1&-1&0&0&0&0&0&0\\ 0&0&0&0&0&0&0&0&0&1&-1&-1
  \end{array}\right), \\[10pt]
    W & = \begin{pmatrix}
1&0&0&0&0&0&1&1&0&0
\\  0&1&0&0&0&0&0&0&1&1\\  0&0&1&1&0
&0&1&0&1&0\\0&0&0&0&1&1&0&1&0&1
\end{pmatrix}.
  \end{align*}}
 The rank of $N$ is $s=6$.     The matrix $W$ gives rise to the conservation relations
\begin{align*}
c_1 &= x_1+x_7+x_8, &  c_3  &   =   x_3+x_4+x_7+x_9, \\
c_2 &= x_2+x_9+x_{10}, &   c_4  &=   x_5+x_6+x_8+x_{10},
\end{align*}
where $c_1,c_2,c_3,c_4$ correspond to the total amounts of kinase, phosphatase, substrate $A$ and substrate $B$, respectively.

 With mass-action kinetics, the vector of reaction rates is
 $$ v(x)= (\k_{{1}}x_{{1}}x_{{3}},  \k_{{2}}x_{{7}},  \k_{{3}}x_{{7}},  \k_{{4}}x_{{2}}x_{{4}},  \k_{{5}}x_{{9}},  \k_{{6}}x_{{9}},  \k_{{7}}x_{{1}}x_{{5}},  \k_{{8}}x_{{8}},  \k_{{9}}x_{{8}},  \k_{{10}}x_{{2}}x_{{6}} , \k_{{11}}x_{{10}},  \k_{{12}}x_{{10}}).$$
The function $f(x)=Nv(x)$ is thus 
 \begin{align*}
          f(x) &= ( -\k_{1}x_{1}x_{3}-\k_{7}x_{1}x_{5}+\k_{2}x_{7}+\k_{3}x_{7}+\k_{8}x_{8}+\k_{9}x_{8}, \\
 &\qquad   -\k_{4}x_{2}x_{4}-\k_{10}x_{2}x_{6}+\k_{5}x_{9}+\k_{6}x_{9}+\k_{11}x_{10}+\k_{12}x_{10},  -\k_{1}x_{1}x_{3}+\k_{2}x_{7}+\k_{6}x_{9}, \\
 & \qquad -\k_{7}x_{1}x_{5}+\k_{8}x_{8}+\k_{12}x_{10}, -\k_{10}x_{2}x_{6}+\k_{9}x_{8}+\k_{11}x_{10}, \k_{1}x_{1}x_{3}-\k_{2}x_{7}-\k_{3}x_{7}, \\ & \qquad  \k_{7}x_{1}x_{5}-\k_{8}x_{8}-\k_{9}x_{8}, \k_{4}x_{2}x_{4}-\k_{5}x_{9}-\k_{6}x_{9}, 
 \k_{10}x_{2}x_{6}-\k_{11}x_{10}-\k_{12}x_{10}).
 \end{align*}

We apply the algorithm to this network with the matrix $N$ and the vector $v(x)$.

\medskip\noindent
{\bf Step 1. } Mass-action kinetics fulfils assumption  \eqref{eq:assumption}. The function $f(x)$ and $W$ are given above and the matrix $W$  is row reduced.

\medskip\noindent
{\bf Step 2. }  {The network is a PTM network, hence it is conservative and thus dissipative.}

\medskip\noindent
{\bf Step 3.  }  {We apply the reduction technique from Section~\ref{sec:siphons}. }
The network has four intermediates $AK, A_pF, BK, B_pF$. After their elimination, we are left with the reaction network
\begin{align*}
A + K &  \cee{->} A_p + K  &  B + K  & \cee{->}  B_p + K  & 
A_p + F &   \cee{->}  A + F & B_p + F  &   \cee{->}  B + F.
\end{align*}
This network has two catalysts: $K,F$. Their elimination yields the reaction network {(the so-called \emph{underlying substrate network} in the main text)}
$$ A \cee{<=>} A_p \qquad B \cee{<=>} B_p.$$
This is a monomolecular network with two strongly connected components. By Corollary~\ref{cor:boundary}, there are no boundary equilibria in any $\mP_c$ for which $\mP_c^+\neq \emptyset$.

\medskip\noindent
{\bf Step 4.  } 
For our choice of $W$, 
we have $i_1=1,i_2=2,i_3=3,i_4=5$.
The function $\varphi_c(x)$ is thus
\begin{align*}
\varphi_c(x) & = \big( x_1+x_7+x_8-c_1, x_2+x_9+x_{10}-c_2,   x_3+x_4+x_7+x_9-c_3,   \\ &  \qquad 
-\k_{4}x_{2}x_{4}+\k_{3}x_{7}+\k_{5}x_{9},  x_5+x_6+x_8+x_{10}-c_4, 
-\k_{10}x_{2}x_{6}+\k_{9}x_{8}+\k_{11}x_{10},  \\ &  \qquad
 \k_{1}x_{1}x_{3}-\k_{2}x_{7}-\k_{3}x_{7}, \k_{7}x_{1}x_{5}-\k_{8}x_{8}-\k_{9}x_{8}, 
\k_{4}x_{2}x_{4}-\k_{5}x_{9}-\k_{6}x_{9}, \\ &  \qquad
\k_{10}x_{2}x_{6}-\k_{11}x_{10}-\k_{12}x_{10} \big).
\end{align*}
The Jacobian matrix $M(x)=J_{\varphi_c}(x)$ is
{\small $$  
\begin{pmatrix}
1&0&0&0&0&0&1&1&0&0
\\ 0&1&0&0&0&0&0&0&1&1\\ 0&0&1&1&0 
&0&1&0&1&0\\ 0&-\k_{4}x_{4}&0&-\k_{4}x_{2}
&0&0&\k_{3}&0&\k_{5}&0\\ 0&0&0&0&1&1&0&1&0&1
\\ 0& -\k_{10} x_6 & 0&0&0& -\k_{10} x_2 &0& \k_9 &0& \k_{11}\\ 
\k_{1}
x_{3}&0&\k_{1}x_{1}&0&0&0&-\k_{2}-\k_{3}&0&0&0
\\ \k_{7}x_{5}&0&0&0&\k_{7}x_{1}&0&0&-\k_{8}-\k_{9}&0&0\\ 0&\k_{4}x_{4}&0&\k_{4}x_{2}&0&0&0&0&-\k_{5}-\k_{6}&0\\ 0&\k_{10}x_{6}&0&0&0&\k_{10}x_{2}&0&0&0&-\k_{11}-\k_{12}
\end{pmatrix}.$$}
The determinant of $M(x)$ is a large polynomial. We omit it here. 

\medskip\noindent
{\bf Step 5.  }  The determinant of $M(x)$ has terms of sign $(-1)^{s+1}=-1$.  {We postpone the discussion of the conditions on the reaction rate constants for which all terms have sign $(-1)^s$ to Step 7.}
We proceed to the next step.

\medskip\noindent
{\bf Step 6.  }  There is a non-interacting set  with $s=6$ species: $$\{X_4, X_6, X_7,X_8,X_9,X_{10} \} = \{
A_p,   B_p, AK, BK, A_pF,   B_pF\}.$$
By solving the equilibrium equations $f_4=f_6=f_7=f_8=f_9=f_{10}=0$
in the variables $x_4,x_6,x_7,x_8,x_9,x_{10}$, we find the following positive parameterization of the set of equilibria in terms of $\widehat{x}=(x_1,x_2,x_3,x_5)$:
\begin{align*}
x_{4} & =\frac{ \left( \k_{6}+\k_{5} \right) \k_{3}x_{3}x_{1}\k_{1}}{\k_{6} \left( \k_{3}+\k_{2} \right) x_{2}\k_{4}},
& x_{7}& =\frac{\k_{1}x_{1}x_{3}}{\k_{3}+\k_{2}},  & x_{9}& =\frac{\k_{3}x_{3}x_{1}\k_{1}}{\k_{6} \left( \k_{3}+\k_{2} \right) },\\
x_{6}& =\frac{ \left( \k_{12}+\k_{11} \right) \k_{9}x_{5}x_{1}\k_{7}}{\k_{12} \left( \k_{9}+\k_{8} \right) \k_{10}x_{2}},
& x_{8}& =\frac {\k_{7}x_{1}x_{5}}{\k_{9}+\k_{8}},
& x_{10} & =\frac {\k_{9}x_{5}x_{1}\k_{7}}{\k_{12} \left( \k_{9}+\k_{8} \right) }.
\end{align*}
{The free variables of the  parameterization are the concentrations of the two enzymes and one substrate per conservation relation involving substrates.}

 \medskip\noindent
{\bf Step 7.  } 
The function $a(\widehat{x})$ is a large rational function with positive denominator. Therefore, the numerator of this function, a polynomial $p(\widehat{x})$,  determines  the sign of $a(\widehat{x})$. 
The coefficients are polynomials in $\k_1,\dots,\k_{10}$. All  but one of the  coefficients
are polynomials in $\k_1,\dots,\k_{10}$ with positive coefficients. Therefore, all coefficients but one are always positive, independently of the values of the reaction rate constants $\k_1,\dots,\k_{10}$.

The only coefficient with sign depending on the specific values of  $\k_1,\dots,\k_{10}$ is
\begin{align*}
\alpha(\k) & =\k_{1}\k_{7} \left( \k_{3}\k_{12}-\k_{6}\k_{9} \right)  ( \k_{1}\k_{3}\k_{5}\k_{8}\k_{10}\k_{12}+\k_{1}\k_{3}\k_{5}\k_{9}\k_{10}\k_{12}+\k_{1}\k_{3}\k_{6}\k_{8}\k_{10}\k_{12} \\ & 
+\k_{1}\k_{3}\k_{6}\k_{9}\k_{10}\k_{12}-\k_{2}\k_{4}\k_{6}\k_{7}\k_{9}\k_{11}-\k_{2}\k_{4}\k_{6}\k_{7}\k_{9}\k_{12}-\k_{3}\k_{4}\k_{6}\k_{7}\k_{9}\k_{11}-\k_{3}\k_{4}\k_{6}\k_{7}\k_{9}\k_{12} ). 
\end{align*}
If $\alpha(\k)\geq 0$, then all coefficients of $p(\widehat{x})$ are positive, and hence
$a(\widehat{x})$ is positive for all positive $\widehat{x}$.
Using $(-1)^s=(-1)^6=1$, Corollary \ref{cor:main}(A) {(Corollary 2(A) in the main text)} gives that there is a unique positive equilibrium in each stoichiometric compatibility class with non-empty positive relative interior. 

When this coefficient is negative, then we need to check whether $p(\widehat{x})$ is negative for some $\widehat{x}$. We analyse this by finding the Newton polytope and using Proposition~\ref{prop:newton}.

The coefficient $\alpha(\k) $ corresponds to the monomial $x_1^2x_2x_3x_5$.
The exponent vectors of the monomials of $p(\widehat{x})$ are:
$$
\begin{array}{lllllll}
(0, 3, 0, 0), &  (0, 3, 0, 1),&  (0, 3, 1, 0), & (1, 2, 0, 0),&  (1, 2, 0, 1),&  (1, 2, 1, 0),&  (1, 3, 0, 0),\\  
(1, 3, 0, 1) &
 (1, 3, 1, 0), & (2, 1, 0, 0), & (2, 1, 0, 1), & (2, 1, 1, 0), & (2, 1, 1, 1), & (2, 2, 0, 0), \\ (2, 3, 0, 0), & (3, 0, 0, 1), & (3, 0, 1, 0), & (3, 1, 0, 1), & (3, 1, 1, 0).
\end{array}
$$
We find the vertices of the convex hull of the exponent vectors, and find that they are
$$
\begin{array}{lllllll}
(0, 3, 0, 1), & (0, 3, 1, 0), & \mathbf{(2, 1, 1, 1)}, &  (0, 3, 0, 0), &  (2, 1, 0, 0), & (3, 0, 0, 1),  \\ (3, 0, 1, 0), & (1, 3, 0, 1), & (1, 3, 1, 0), & (2, 3, 0, 0), & (3, 1, 0, 1),&  (3, 1, 1, 0).
\end{array}
$$
Thus the exponent vector of the monomial of interest, $(2, 1, 1, 1)$ (highlighted in bold), is a vertex of the Newton polytope. Therefore, by Proposition~\ref{prop:newton}, there exists
$\widehat{x}$ such that $p(\widehat{x})$ is negative. Corollary \ref{cor:main}(B) {(Corollary 2(B) in the main text)}  gives that there is a 
stoichiometric compatibility class with multiple positive equilibria.

The condition $\alpha(\k)<0$ can be rewritten as:
\begin{align*} 
( \k_{3}\k_{12}-\k_{6}\k_{9} )  
(\k_{3} \k_{12}  \k_{1}\k_{10} (\k_{5}+\k_6)(\k_{8}+ \k_{9})  
 - \k_{6} \k_9    \k_{4}\k_{7}(\k_{2}+\k_{3})(\k_{11} +  \k_{12}) )<0, 
  \end{align*}
which in turn can be written as
\begin{align*} 
( \k_{3}\k_{12}-\k_{6}\k_{9} )  
\left(\k_{3} \k_{12} \cdot \frac{\k_{1}}{\k_{2}+\k_{3}} \cdot \frac{\k_{10} }{\k_{11} +  \k_{12} }
 - \k_{6} \k_9 \cdot  \frac{\k_{4}}{\k_{5}+\k_6}\cdot  \frac{\k_{7}}{\k_{8}+ \k_{9}  }   \right)<0, 
  \end{align*}
Note that $\k_3,\k_6,\k_9,\k_{12}$ are the catalytic constants of phosphorylation/dephosphorylation of $A$ and $B$ ($k_{c1}, k_{c2}, k_{c3}, k_{c4}$ in the main text), and 
$$k_{M1}^{-1}=\frac{\k_{1}}{\k_{2}+\k_{3}}, \qquad k_{M2}^{-1}=\frac{\k_{4}}{\k_{5}+\k_6}, \qquad k_{M3}^{-1}=\frac{\k_{7}}{\k_{8}+ \k_{9}},
\qquad k_{M4}^{-1}=\frac{\k_{10} }{\k_{11} +  \k_{12}}$$ 
are the inverses of the Michaelis-Menten constants of $K$ and $F$ for each substrate.
 Therefore, the necessary and sufficient condition for multistationarity can be written in terms of the catalytic constants and the Michaelis-Menten constants,
 \begin{align*} 
( \k_{3}\k_{12}-\k_{6}\k_{9} )  
\left( \frac{\k_{3} \k_{12}}{k_{M1} k_{M4} }  - \frac{\k_{6} \k_9}{ k_{M2} k_{M3}}   \right)<0.
  \end{align*}
This proves the condition for multiple and unique equilibria given in the first row of Table 1 in the main text, by letting 
$$k_{c1}=\k_3, \qquad k_{c2}=\k_6, \qquad k_{c3}=\k_9, \qquad k_{c4}=k_{12}.$$
    In particular, we have that
 \begin{itemize}
 \item If $ \k_{3}\k_{12}>\k_{6}\k_{9}$, then we need $\frac{\k_{3} \k_{12}}{k_{M1} k_{M4} }   < \frac{\k_{6} \k_9}{ k_{M2} k_{M3}} $ for multiple equilibria to occur.
  \item If $ \k_{3}\k_{12}<\k_{6}\k_{9}$, then we need $\frac{\k_{3} \k_{12}}{k_{M1} k_{M4} }   > \frac{\k_{6} \k_9}{ k_{M2} k_{M3}} $ for multiple equilibria to occur.
 \end{itemize}

\subsection{Two-site phosphorylation system}

In this subsection we consider the network in the second row of Table 1 in the main text.
The conditions given here were also found in \cite{maya-bistab}, the paper that lay the foundations of this algorithm. In this work we consider a direct route using the function $\varphi_c$ and avoiding changes of variables. 
We explain here how to find the conditions  using the algorithm in the main text.  

 We consider a system in which one substrate undergoes sequential and distributive phosphorylation by a kinase $K$ and sequential and distributive dephosphorylation by a phosphatase $F$. 
 The three phosphoforms of the substrate are $A, A_p, A_{pp}$.
 The reactions of the system are:
\begin{align*}
A + K  \cee{<=>[\k_1][\k_2]}   AK  \cee{->[\k_3]} A_p + K  & \cee{<=>[\k_7][\k_8]}   A_pK  \cee{->[\k_9]}  A_{pp} + K \\
A_{pp} + F   \cee{<=>[\k_{10}][\k_{11}]}   A_{pp}F \cee{->[\k_{12}]}  A_p + F &  \cee{<=>[\k_4][\k_5]}   A_pF  \cee{->[\k_6]}  A + F 
\end{align*}
{This network is a PTM network with substrates $A,A_p,A_{pp}$, enzymes $K,F$ and intermediates $AK, A_pK, A_pF,A_{pp}F$.}
We let 
\begin{align*}
 X_1 & =K, &   X_3 & =A, & X_5& =A_{pp}, & X_6 &=AK, &   X_7& =A_pF,  \\
 X_2  & =F,  & X_4 & =A_p,  & &  & X_8 & =A_pK, & X_{9} & =A_{pp}F.
\end{align*}

 The stoichiometric matrix $N$ of the network and a row reduced matrix $W$ whose rows from  a basis of $\im(N)^\perp$
are
{\small \begin{align*}
N & =\left(
\begin {array}{rrrrrrrrrrrr} -1&1&1&0&0&0&-1&1&1&0&0&0
\\ 0&0&0&-1&1&1&0&0&0&-1&1&1\\ -1&
1&0&0&0&1&0&0&0&0&0&0\\ 0&0&1&-1&1&0&-1&1&0&0&0&1
\\ 0&0&0&0&0&0&0&0&1&-1&1&0\\ 1&-1
&-1&0&0&0&0&0&0&0&0&0\\ 0&0&0&1&-1&-1&0&0&0&0&0&0
\\ 0&0&0&0&0&0&1&-1&-1&0&0&0\\ 0&0
&0&0&0&0&0&0&0&1&-1&-1\end {array}
\right), \\[10pt]
    W & = \begin{pmatrix}
 1&0&0&0&0&1&0&1&0
\\ 0&1&0&0&0&0&1&0&1\\  0&0&1&1&1&1
&1&1&1 
\end{pmatrix}.
  \end{align*}}
  The rank of $N$ is $s=6$.
     The matrix $W$ gives rise to the conservation relations
\begin{align*}
c_1 &= x_1+x_6+x_8, & c_2 &= x_2+x_7+x_9,  & c_3  &   =   x_3+x_4+x_5+x_6+x_7+x_8+x_9,
\end{align*}
where $c_1,c_2,c_3$ correspond to the total amounts of kinase, phosphatase and  substrate $A$,  respectively.

 With mass-action kinetics, the vector of reaction rates is
 $$ v(x)=(\k_{{1}}x_{{1}}x_{{3}},\k_{{2}}x_{{6}},\k_{{3}}x_{{6}},\k_{{4}}x_{{2}}x_{
{4}},\k_{{5}}x_{{7}},\k_{{6}}x_{{7}},\k_{{7}}x_{{1}}x_{{4}},\k_{{8}}x_{{8}
},\k_{{9}}x_{{8}},\k_{{10}}x_{{2}}x_{{5}},\k_{{11}}x_{{9}},\k_{{12}}x_{{9}
})
.$$
The function $f(x)=Nv(x)$ is thus 
 \begin{align*}
 f(x) &= ( -\k_{{1}}x_{{1}}x_{{3}}-\k_{{7}}x_{{1}}x_{{4}}+\k_{{2}}x_{{6}}+\k_{{3}}x_{{6}}+\k_{{8}}x_{{8}}+\k_{{9}}x_{{8}}, \\ & \qquad -\k_{{4}}x_{{2}}x_{{4}}-\k_{{10}}x_{{2}}x_{{5}}+\k_{{5}}x_{{7}}+\k_{{6}}x_{{7}}+\k_{{11}}x_{{9}}+\k_{{12}}x_{{9}},  -\k_{{1}}x_{{1}}x_{{3}}+\k_{{2}}x_{{6}}+\k_{{6}}x_{{7}}, 
 \\ & \qquad  -\k_{{4}}x_{{2}}x_{{4}}-\k_{{7}}x_{{1}}x_{{4}}+\k_{{3}}x_{{6}}+\k_{{5}}x_{{7}}+\k_{{8}}x_{{8}}+\k_{{12}}x_{{9}},  -\k_{{10}}x_{{2}}x_{{5}}+\k_{{9}}x_{{8}}+\k_{{11}}x_
{{9}}, 
\\ & \qquad  \k_{{1}}x_{{1}}x_{{3}}-\k_{{2}}x_{{6}}-\k_{{3}}x_{{6}},  \k_{{4}}x_{{2}}x_{{4}}-\k_{{5}}x_{{7}}-\k_{{6
}}x_{{7}},  \k_{{7}}x_{{1}}x_{{4}}-\k_{{8}}x_{{8}}-\k_{{9}}x_{{8}}, 
\\ & \qquad  \k_{{10}}x_{{2}}x_{{5}}-\k_{{11}}x_{{9}}-\k_{{12}}x_{{9}} ).
 \end{align*}

We apply the algorithm to this network with the matrix $N$ and the vector $v(x)$.

\medskip\noindent
{\bf Step 1. } Mass-action kinetics fulfils assumption  \eqref{eq:assumption}. The function $f(x)$ and $W$ are given above and the matrix $W$  is row reduced.

\medskip\noindent
{\bf Step 2. }  {This network is dissipative since it is a PTM network.}

\medskip\noindent
{\bf Step 3.  } 
The network has four intermediates $AK, A_pK, A_pF, A_{pp}F$. After their elimination, we are left with
a reaction network with two catalysts: $K,F$.
Their elimination yields the {following underlying substrate network}
$$ A \cee{<=>} A_p   \cee{<=>} A_{pp}.$$
This is a monomolecular network with two strongly connected components. By Corollary~\ref{cor:boundary}, there are no boundary equilibria  in any $\mP_c$ for which $\mP_c^+\neq \emptyset$.

\medskip\noindent
{\bf Step 4.  } 
For our choice of $W$, we have $i_1=1,i_2=2,i_3=3$.
The function $\varphi_c(x)$ is thus
 \begin{align*}
\varphi_c(x) & = \big(x_1+x_6+x_8-c_1 ,  x_2+x_7+x_9-c_2, x_3+x_4+x_5+x_6+x_7+x_8+x_9 - c_3, 
 \\ & \qquad  -\k_{{4}}x_{{2}}x_{{4}}-\k_{{7}}x_{{1}}x_{{4}}+\k_{{3}}x_{{6}}+\k_{{5}}x_{{7}}+\k_{{8}}x_{{8}}+\k_{{12}}x_{{9}},  -\k_{{10}}x_{{2}}x_{{5}}+\k_{{9}}x_{{8}}+\k_{{11}}x_
{{9}}, 
\\ & \qquad  \k_{{1}}x_{{1}}x_{{3}}-\k_{{2}}x_{{6}}-\k_{{3}}x_{{6}},  \k_{{4}}x_{{2}}x_{{4}}-\k_{{5}}x_{{7}}-\k_{{6
}}x_{{7}},  \k_{{7}}x_{{1}}x_{{4}}-\k_{{8}}x_{{8}}-\k_{{9}}x_{{8}}, 
\\ & \qquad  \k_{{10}}x_{{2}}x_{{5}}-\k_{{11}}x_{{9}}-\k_{{12}}x_{{9}} \big).
 \end{align*}
The Jacobian matrix $M(x)=J_{\varphi_c}(x)$ is
{\small $$  
\begin{pmatrix}
 1&0&0&0&0&1&0&1&0
\\ 0&1&0&0&0&0&1&0&1\\ 0&0&1&1&1&1
&1&1&1\\ -\k_{{7}}x_{{4}}&-\k_{{4}}x_{{4}}&0&-\k_{{4}}x_{{2}}-\k_{{7}}x_{{1}}&0&\k_{{3}}&\k_{{5}}&\k_{{8}}&\k_{{12}}
\\ 0&-\k_{{10}}x_{{5}}&0&0&-\k_{{10}}x_{{2}}&0&0&\k_{{9}}&\k_{{11}}\\ \k_{{1}}x_{{3}}&0&\k_{{1}}x_{{1}}&0&0&-\k
_{{2}}-\k_{{3}}&0&0&0\\ 0&\k_{{4}}x_{{4}}&0&\k_{{4}}x_{{2}}&0&0&-\k_{{5}}-\k_{{6}}&0&0\\ \k_{{7}}x_{{4}}&0&0&\k
_{{7}}x_{{1}}&0&0&0&-\k_{{8}}-\k_{{9}}&0\\ 0&\k_{{10}}x_{{5}}&0&0&\k_{{10}}x_{{2}}&0&0&0&-\k_{{11}}-\k_{{12}}
\end{pmatrix}. $$}
The determinant of $M(x)$ is a large polynomial. We omit it here. 

\medskip\noindent  
{\bf Step 5.  }  The determinant of $M(x)$ has terms of sign $(-1)^{s+1}=-1$. 
 {We postpone the discussion of the conditions on the reaction rate constants for which all terms have sign $(-1)^s$ to Step 7.}
We proceed to the next step.

\medskip\noindent
{\bf Step 6.  }  This network is a PTM network and has a non-interacting set  with $s=6$ species: 
$$\{X_4, X_5, X_6,X_7,X_8,X_9 \} = \{A_p,   A_{pp}, AK, A_pF, A_pK,  A_{pp}F\}.$$

By solving the equilibrium equations $f_4=f_5=f_6=f_7=f_8=f_9=0$
in the variables $x_4,\dots,x_9$, we find the following positive parameterization of the set of equilibria in terms of $\widehat{x}=(x_1,x_2,x_3)$:
\begin{align*}
x_{4}& =\frac {  \k_{1} \k_{3}( \k_{5}+\k_{6} ) x_{1}x_{3}}{ ( \k_{2}+\k_{3} )\k_{4}\k_{6} x_{2}}, 
& 
x_{5} & =\frac {\k_{1}  \k_{3} ( \k_{5}+\k_{6} )   \k_{7}  \k_{9} ( \k_{11}+\k_{12} ) x_{1}^{2}x_{3}}
{ ( \k_{2}+\k_{3} )\k_{4}\k_{6} ( \k_{8}+\k_{9}) \k_{10}  \k_{12} x_{2}^{2} }, \\
x_{6} & =\frac{\k_{1}x_{1}x_{3}}{\k_{2}+\k_{3}}, &
x_{7} & =\frac{\k_{1}\k_{3}x_{1}x_{3}}{( \k_{2}+\k_{3} )\k_{6} }, \\
x_{8}& =\frac{ \k_{1}\k_{3}  ( \k_{5}+\k_{6})\k_{7} x_{1}^{2}x_{3}}{\k_{2}+\k_{3} )\k_{4} \k_{6} ( \k_{8}+\k_{9})x_{2} }, & 
x_{9}& =\frac {\k_{1}\k_{3} ( \k_{5}+\k_{6})\k_{7}  \k_{9}x_{1}^{2} x_{3}}{
\k_{2}+\k_{3} )\k_{4} \k_{6} ( \k_{8}+\k_{9})\k_{12} x_{2} }.
\end{align*}
{The free variables of this parameterization are the concentrations of the two enzymes and one of the substrates.}
We substitute $x_4,\dots,x_9$  with their expressions in the parameterization in $\det(M(x))$ to  find $a(\widehat{x})$. 
The function $a(\widehat{x})$ is a large rational function with positive denominator which we do not include here.

 \medskip\noindent
{\bf Step 7.  } 
The numerator of $a(\widehat{x})$, the polynomial $p(\widehat{x})$, determines therefore the sign of $a(\widehat{x})$. 
The coefficients are polynomials in $\k_1,\dots,\k_{10}$. 

The polynomial has 15 terms, 9 of which are positive for all values of the reaction rate constants.
The remaining 6 coefficients are polynomials  in $\k_1,\dots,\k_{10}$ that can either be positive or negative.

Five of the six coefficients are of the form $\beta(\k) b_1(\k)$, where
$\beta(\k)$ is a positive polynomial in $\k$ and 
$$ b_1(\k) = \k_3\k_{12} - \k_6 \k_9$$
(thus $b_1(\k)$ is the same for all five coefficients).
These five coefficients correspond to the monomials $x_1^3x_2^2x_3$, $x_1^2x_2^2x_3^2$, $x_1^3x_2x_3^2$, $x_1^2x_2^3x_3$ and $x_1^4x_3^2$.

The remaining coefficient is of the form $\gamma(\k) \alpha(\k)$, where
$\gamma(\k)$ is a positive polynomial in $\k$ and
\begin{align*}
\alpha(\k) & =  \k_{{1}}\k_{{3}}\k_{{4}}\k_{{8}}\k_{{10}}\k_{{12}}+\k_{{1}}\k_{{3}}\k_{{4}}\k_{{9}}\k_{{10}}\k_{{12}}+\k_{{1}}\k_{{3}}\k_{{5}}\k_{{7}}\k_{{10}}\k_{{12}}+\k_{{1}}\k_{{3}}\k_{{6}}\k_{{7}}\k_{{10}}\k_{{12}}
 \\ & \qquad -\k_{{1}}\k_{{4}}\k_{{6}}\k_{{7}}\k_{{9}}\k_{{11}}-\k_{{1}}\k_{{4}}\k_{{6}}\k_{{7}}\k_{{9}}\k_{{12}}-\k_{{2}}\k_{{4}}\k_{{6}}\k_{{7}}\k_{{9}}\k_{{10}}-\k_{{3}}\k_{{4}}\k_{{6}}\k_{{7}}\k_{{9}}\k_{{10}}.
\end{align*}
It corresponds to the monomial $x_1^2x_2^2x_3$.

Since $(-1)^6=1$, part Corollary \ref{cor:main}(A)   {(Corollary 2(A) in the main text)}  tells us that there is a unique positive equilibrium in each stoichiometric compatibility class with non-empty positive part, if 
$$b_1(\k)\geq 0\qquad\textrm{and}\qquad \alpha(\k) \geq 0.$$
The condition $\alpha(\k)\geq 0$ can be rewritten as:
\begin{align*} 
\k_{{1}}\k_{{3}}\k_{{10}}\k_{{12}} \big( \k_{{4}} ( \k_{{9}}+\k_{{8}} ) +\k_{{7}} ( \k_{{6}}+\k_{{5}} )\big) -
\k_{{4}}\k_{{6}}\k_{{7}}\k_{{9}} \big( \k_{{1}} ( \k_{{12}}+\k_{{11}} ) +\k_{{10}} ( \k_{{3}}+\k_{{2}} ) \big)   \geq 0.
  \end{align*}
  Dividing the expression by $\k_1\k_4\k_7\k_{10}$, the condition 
 can be rewritten as
\begin{align*} 
\k_{3}\k_{12}\big(k_{M2} + k_{M3}\big) - \k_{6}\k_{9}\big(k_{M1} + k_{M4}\big) \geq 0,
  \end{align*}
  where
  $$k_{M1}=\frac{\k_{2}+\k_{3}}{\k_{1}}, \qquad k_{M2}=\frac{\k_{5}+\k_6}{\k_{4}}, \qquad 
  k_{M3}=\frac{\k_{8}+ \k_{9}}{\k_{7}}, \qquad k_{M4} =\frac{\k_{11} +  \k_{12}}{\k_{10} }$$ 
are the   Michaelis-Menten constants of $K$ and $F$ for each site.
  Note that $\k_3,\k_6,\k_9,\k_{12}$ are the catalytic constants of phosphorylation of $A$, dephosphorylation of $A_p$, phosphorylation of $A_p$ and dephosphorylation of $A_{pp}$. These are denoted by $k_{c1}, k_{c2}, k_{c3}, k_{c4}$ in the main text by letting 
$$k_{c1}=\k_3, \qquad k_{c2}=\k_6, \qquad k_{c3}=\k_9, \qquad k_{c4}=k_{12}.$$
  
  By letting
  $$ b_2(\k)= \k_{3}\k_{12}\big(k_{M2} + k_{M3}\big) - \k_{6}\k_{9}\big(k_{M1} + k_{M4}\big),$$
  $\alpha(\k)\geq 0$ if and only if $b_2(\k)\geq 0$. Thus   we have proven the condition for unique equilibria given in the second row of Table 1 in the main text.

 Let us consider whether Corollary \ref{cor:main}(B) {(Corollary 2(B) in the main text)}  applies if $b_1(\k)<0$ and/or $\alpha(\k)<0$.
The exponent vectors of the monomials of $p(\widehat{x})$ are:
$$
\begin{array}{llllllll}
(3, 1, 1) &  (1, 3, 1) &  (2, 2, 1) &  (2, 2, 2) &  (2, 3, 0) &  (2, 2, 0) &  (1, 3, 0) &  (3, 1, 2) \\  (2, 3, 1)&  (3, 2, 1) &  (4, 0, 2)&   (4, 0, 1) &  (0, 4, 1)&  (1, 4, 0) &  (0, 4, 0)
\end{array}
$$
The vertices of the convex hull of the exponent vectors are
$$
\begin{array}{llllllllllll}
(2, 3, 0) &  (4, 0, 1) & (2, 2, 0) &  (0, 4, 0) &  (1, 4, 0) &  \mathbf{(3, 2, 1)} &   (4, 0, 2)&   (0, 4, 1) &  (2, 3, 1) &  (2, 2, 2).
\end{array}
$$
The vertex highlighted in bold corresponds to the monomial $x_1^3x_2^2x_3$, whose sign depends on $b_1(\k)$. By Proposition~\ref{prop:newton}, if $b_1(\k)<0$, then there exists
$\widehat{x}$ such that $p(\widehat{x})$ is negative. Corollary \ref{cor:main}(B)  {(Corollary 2(B) in the main text)}   gives that there is a 
stoichiometric compatibility class that  admits positive multiple equilibria. This proves the condition for multistationarity given in the second row of Table 1 in the main text.

The exponent vector of the monomial corresponding to the coefficient $\alpha(\k)$, $(2,2,1)$, is not a vertex of the Newton polytope. In this case it is uncertain whether the condition $\alpha(\k)<0$ is sufficient for multistationarity.

\subsection{Two-substrate enzyme catalysis}

{
  This section contains an additional example to illustrate the application of the algorithm  to a monostationary network  for which a parameterization is required to reach the conclusion.}
 
We consider a mechanism in which an enzyme $E$ binds two substrates, $S_1,S_2$, in an unordered manner in order  to catalyze the reversible  conversion to the product $P$.
A variation of this system was considered in \cite{craciun2006}.
The reactions of the system are:
\begin{align*} 
 E+ S_1 & \cee{<=>[\k_1][\k_2]}   ES_1 &  S_2+ES_1 &  \cee{<=>[\k_5][\k_6]}  ES_{1}S_2 &  ES_{1}S_2 & \cee{<=>[\k_7][\k_8]}   E+P \\
E+S_{2}  & \cee{<=>[\k_3][\k_4]}   ES_2 &   S_1+ES_2 &  \cee{<=>[\k_9][\k_{10}]}  ES_{1}S_2.
\end{align*}
We let 
\begin{align*}
 X_1 & =E, &  X_2  & =S_1,   &  X_3 & =ES_1, & X_4 & =S_2,  & X_5& =ES_2, & X_6& =ES_1S_2, & X_7 &=P.
\end{align*}
 The stoichiometric matrix $N$ of the network and a row reduced matrix $W$ whose rows from  a basis of $\im(N)^\perp$
are
{\small \begin{align*}
N & =\left(
\begin {array}{rrrrrrrrrr} -1&1&-1&1&0&0&0&0&1&-1
\\ -1&1&0&0&0&0&1&-1&0&0\\ 1&-1&0&0
&-1&1&0&0&0&0\\ 0&0&-1&1&-1&1&0&0&0&0
\\ 0&0&1&-1&0&0&1&-1&0&0\\ 0&0&0&0
&1&-1&-1&1&-1&1\\ 0&0&0&0&0&0&0&0&1&-1
  \end{array}\right) \\[10pt]
    W & = \begin{pmatrix}
1&0&1&0&1&1&0\\ 0&1&
1&0&0&1&1\\ 0&0&0&1&1&1&1
\end{pmatrix}.
  \end{align*}}
  The rank of $N$ is $s=4$.
     The matrix $W$ gives rise to the conservation relations
\begin{align*}
c_1 &= x_1+x_3+x_5+x_6, &  c_2  &   =x_2+ x_3+x_6+x_7 & c_3 &=  x_4+x_5+x_6+x_7,
\end{align*}
where $c_1,c_2,c_3,c_4$ correspond to the total amounts of kinase, substrate $S_1$ and substrate $S_2$, respectively.

With mass-action kinetics, the vector of reaction rates is
 $$ v(x)= (\k_{{1}}x_{{1}}x_{{2}},\k_{{2}}x_{{3}},\k_{{3}}x_{{1}}x_{{4}},\k_{{4}}x_{{5}},\k_{{5}}x_{{4}}x_{{3}},\k_{{6}}x_{{6}},\k_{{7}}x_{{6}},\k_{{8}}x_{{2}}x_{{5}},\k_{{9}}x_{{6}},\k_{{10}}x_{{1}}x_{{7}}).$$
The function $f(x)=Nv(x)$ is
 \begin{align*}
          f(x) &= ( -\k_{1}x_{1}x_{2}-\k_{3}x_{1}x_{4}-\k_{10}x_{1}x_{7}+\k_{2}x_{3}+\k_{4}x_{5}+\k_{9}x_{6},  \\ & \qquad 
   -\k_{1}x_{1}x_{2}-\k_{8}x_{2}x_{5}+\k_{2}x_{3}+\k_{7}x_{6},  \k_{1}x_{1}x_{2}-\k_{5}x_{4}x_{3}-\k_{2}x_{3}+\k_{6}x_{6}\\  &  \qquad 
  -\k_{3}x_{1}x_{4}-\k_{5}x_{4}x_{3}+\k_{4}x_{5}+\k_{6}x_{6}, 
 \k_{3}x_{1}x_{4}-\k_{8}x_{2}x_{5}-\k_{4}x_{5}+\k_{7}x_{6}, \\  &  \qquad 
  \k_{5}x_{4}x_{3}+\k_{8}x_{2}x_{5}+\k_{10}x_{1}x_{7}-\k_{6}x_{6}-\k_{7}x_{6}-\k_{9}x_{6},  
 -\k_{10}x_{1}x_{7}+\k_{9}x_{6}).
\end{align*}

We apply the algorithm to this network with the matrix $N$ and the vector $v(x)$.

\medskip\noindent
{\bf Step 1. } Mass-action kinetics fulfils assumption  \eqref{eq:assumption}. The function $f(x)$ and $W$ are given above and the matrix $W$  is row reduced.

\medskip\noindent
{\bf Step 2. }  
The network is conservative since   the concentration of every species is in the support of a conservation relation with positive coefficients. Therefore  the network is dissipative.

\medskip\noindent
{\bf Step 3.  } 
This network has only one intermediate $ES_1S_2$.
Its removal yields the reaction network
\begin{align*} 
 E+ S_1 & \cee{<=>}   ES_1  & 
 S_2+ES_1  & \cee{<=>}   E+P &  S_2+ES_1 &  \cee{<=>}  S_1+ES_2 \\
E+S_{2}  & \cee{<=>}   ES_2 &   S_1+ES_2 &  \cee{<=>}  E+P.
\end{align*}
The conservation relations of this new network are (with the notation above):
\begin{align*}
c_1 &= x_1+x_3+x_5, &  c_2  &   =x_2+ x_3+x_7 & c_3 &=  x_4+x_5+x_7.
\end{align*}
The minimal siphons of the network are
$$ \{ E, ES_1,ES_2\}, \{S_1,ES_1,P \}, \{S_2,ES_2,P \}. $$
These siphons contain the support of the conservation relations for $c_1,c_2,c_3$ respectively. Thus, by {Proposition~2 in the main text and Proposition~\ref{prop:red}}, the  original network does not 
have boundary equilibria in any stoichiometric compatibility class that intersects the positive orthant.

\medskip\noindent
{\bf Step 4.  } 
For our choice of $W$, we have $i_1=1,i_2=2,i_3=4$.
The function $\varphi_c(x)$ is thus
\begin{align*}
\varphi_c(x) & = \Big( 
 x_1+x_3+x_5+x_6-c_1,x_2+ x_3+x_6+x_7-c_2,   \k_{1}x_{1}x_{2}-\k_{5}x_{4}x_{3}-\k_{2}x_{3}+\k_{6}x_{6}, \\ & 
 \qquad  x_4+x_5+x_6+x_7 - c_3,
  \k_{3}x_{1}x_{4}-\k_{8}x_{2}x_{5}-\k_{4}x_{5}+\k_{7}x_{6},  \\ & \qquad 
\k_{5}x_{4}x_{3}+\k_{8}x_{2}x_{5}+\k_{10}x_{1}x_{7}-\k_{6}x_{6}-\k_{7}x_{6}-\k_{9}x_{6}, 
-\k_{10}x_{1}x_{7}+\k_{9}x_{6}
\Big).
\end{align*}
The Jacobian matrix $M(x)=J_{\varphi_c}(x)$ is
{\small $$  \left(
 \begin {array}{ccccccc} 1&0&1&0&1&1&0\\ 
 0&1& 1&0&0&1&1\\ 
 \k_{1}x_{2}&\k_{1}x_{1}&-\k_{5}x_{4}-\k_{2}&-\k_{5}x_{3}&0&\k_{6}&0\\ 
 0&0&0&1&1&1&1\\ 
 \k_{3}x_{4}&-\k_{8}x_{5}&0&\k_{3}x_{1}&-\k_{8}x_{2}-\k_{4}&\k_{7}&0\\ 
 \k_{10}x_{7}&\k_{8}x_{5}&\k_{5}x_{4}&\k_{5}x_{3}&\k_{8}x_{2}&-\k_{6}-\k_{7}-\k_{9}&\k_{10}x_{1}\\ 
 -\k_{10}x_{7}&0&0&0&0&\k_{9}&-\k_{10}x_{1}\end {array}
\right). $$}
 The determinant of $M(x)$ is a large polynomial. We omit it here. 

\medskip\noindent
{\bf Step 5.  }  The determinant of $M(x)$ has terms of sign $(-1)^{s+1}=-1$.  
 {We postpone the discussion of the conditions on the reaction rate constants for which all terms have sign $(-1)^s$ to Step 7.}
We proceed to the next step.

\medskip\noindent
{\bf Step 6.  }  This network is not a PTM system, but has a non-interacting set  with $s=4$ species: 
$$\{X_3, X_5, X_6,X_7 \} = \{ES_1,  ES_2, ES_1S_2,P\}.$$
 By solving the equilibrium equations $f_3=f_5=f_6=f_7=0$
in the variables $x_3,x_5,x_6,x_7$, we find the following positive parameterization of the set of equilibria in terms of $\widehat{x}=(x_1,x_2,x_4)$:
\begin{align*}
x_{3} & =\frac{x_{2}x_{1} \left( \k_{1}\k_{6}\k_{8}x_{2}+\k_{3}\k_{6}\k_{8}x_{4}+\k_{1}\k_{4}\k_{6}+\k_{1}\k_{4}\k_{7}
 \right) }{\k_{2}\k_{6}\k_{8}x_{2}+\k_{4}\k_{5}\k_{7}x_{4}+\k_{2}\k_{4}\k_{6}+\k_{2}\k_{4}\k_{7}}, \\
x_{5}& =\frac{x_{1}x_{4} \left( \k_{1}\k_{5}\k_{7}x_{2}+\k_{3}\k_{5}\k_{7}x_{4}+\k_{2}\k_{3}\k_{6}+\k_{2}\k_{3}\k_{7} \right) }{\k_{2}\k_{6}\k_{8}x_{2}+\k_{4}\k_{5}\k_{7}x_{4}+\k_{2}\k_{4}\k_{6}+\k_{2}\k_{4}\k_{7}}, \\
 x_{6}& =\frac{x_{2}x_{4} \left( \k_{1}\k_{5}\k_{8}x_{2}+\k_{3}\k_{5}\k_{8}x_{4}+\k_{1}\k_{4}\k_{5}+\k_{2}\k_{3}\k_{8} \right) x_{1}}{\k_{2}\k_{6}\k_{8}x_{2}+\k_{4}\k_{5}\k_{7}x_{4}+\k_{2}\k_{4}\k_{6}+\k_{2}\k_{4}\k_{7}},\\
x_{7}& =\frac{\k_{9}x_{2}x_{4} \left( \k_{1}\k_{5}\k_{8}x_{2}+\k_{3}\k_{5}\k_{8}x_{4}+\k_{1}\k_{4}\k_{5}+\k_{2}\k_{3}
\k_{8} \right) }{ \left( \k_{2}\k_{6}\k_{8}x_{2}+\k_{4}\k_{5}\k_{7}x_{4}+\k_{2}\k_{4}\k_{6}+\k_{2}\k_{4}\k_{7} \right) \k_{10}}.
\end{align*}
We substitute $x_3,x_5,x_6,x_7$   with their expressions in the parameterization in $\det(M(x))$ to  find $a(\widehat{x})$. 
The function $a(\widehat{x})$ is a large rational function. 

 \medskip\noindent
{\bf Step 7.  }  The numerator and denominator of $a(\widehat{x})$ are polynomials in $x$ and $\k$ with all coefficients positive. 
By Corollary \ref{cor:main}(A) {(Corollary 2(A) in the main text)}    using $s=4$, there is a unique positive equilibrium in each stoichiometric compatibility class that intersects the positive orthant.


\end{document}